\documentclass[11pt]{article}

\usepackage{graphicx}%
\usepackage{multirow}%
\usepackage{amsmath,amssymb,amsfonts}%
\usepackage{amsthm}%
\usepackage{mathrsfs}%
\usepackage[title]{appendix}%
\usepackage{xcolor}%
\usepackage{textcomp}%
\usepackage{manyfoot}%
\usepackage{booktabs}%
\usepackage{algorithm}%
\usepackage{algorithmicx}%
\usepackage{algpseudocode}%
\usepackage{listings}%

\newtheorem{proposition}{Proposition}%

\usepackage{url}

\usepackage{indentfirst}

\usepackage{circuitikz}

\usepackage{bbm}
\usepackage{lipsum}
\usepackage{amsfonts}
\usepackage{graphicx}
\usepackage{epstopdf,comment}
\usepackage{algorithm}
\usepackage{algpseudocode}
\usepackage{amsmath}
\usepackage{mathtools}
\usepackage{graphicx}
\usepackage{enumerate,url}
\ifpdf
  \DeclareGraphicsExtensions{.eps,.pdf,.png,.jpg}
\else
  \DeclareGraphicsExtensions{.eps}
\fi

\usepackage{amsopn}

\DeclareMathOperator{\spn}{span}
\DeclareMathOperator{\Exp}{Exp}

\newcommand{\Cov}{\mathrm{Cov}}
\newcommand{\Cor}{\mathrm{Cor}}
\newcommand{\norm}[1]{\left \lVert #1 \right \rVert}

\theoremstyle{remark}

\newtheorem{corollary}{Corollary}

\DeclareMathOperator*{\argmin}{arg\,min}

\DeclareMathOperator{\Var}{Var}
\usepackage{tikz,tikz-cd}
\usepackage{color} 
\definecolor{mygreen}{rgb}{0, 0.82 , 0.0}
\definecolor{myorange}{rgb}{0.85, 0.5, 0.0}
\definecolor{mymagenta}{rgb}{1,0,1}
\definecolor{mycyan}{rgb}{0.0, 1.0, 1.0}
\definecolor{mypink}{rgb}{0.97, 0.68, 0.64}
\definecolor{myblue}{rgb}{0.57, 0.61, 0.90}
\definecolor{mylightgreen}{rgb}{0.66, 1, 0.66}
\definecolor{mygrey}{rgb}{0.6, 0.6, 0.6}

\usepackage{authblk}
\title{Covariance Decomposition for Distance Based Species Tree
  Estimation}
\author[1]{Georgios Aliatimis}
\author[2]{Ruriko Yoshida}
\author[3]{Burak Boyac\i}
\author[4]{James A.  Grant}
\affil[1]{{STOR-i Centre for Doctoral Training}, {Lancaster University}, { {Lancaster}, {LA1 4YW}, {UK}}}

\affil[2]{{Department of Operations Research}, {Naval Postgraduate School}, {{1411 Cunningham Road}, {Monterey}, {93943}, {CA}, {USA}}}

\affil[3]{{Management School}, {Lancaster University}, {{Lancaster}, {LA1 4YX},  {UK}}}

\affil[4]{{School of Mathematical Sciences}, {Lancaster University}, {{Lancaster}, {LA1 4YF},  {UK}}}

\begin{document}

\maketitle

\abstract{
In phylogenomics, species-tree methods must contend with two major sources of noise;
stochastic gene-tree variation under the multispecies coalescent model (MSC) and 
finite-sequence substitutional noise. Fast agglomerative methods such as GLASS, STEAC, and METAL
combine multi-locus information via distance-based clustering. We derive the exact covariance matrix of these pairwise distance estimates under a joint MSC-plus-substitution
model and leverage it for reliable confidence estimation, and we algebraically decompose it into components attributable to coalescent variation versus sequence-level stochasticity. 
Our theory identifies parameter regimes 
where one source of variance greatly exceeds the other. For both very low and very high
mutation rates, substitutional noise dominates, while coalescent variance is the primary 
contributor at intermediate mutation rates. Moreover, the interval over which coalescent variance dominates becomes narrower as the species-tree height increases. These results 
imply that in some settings one may legitimately ignore the weaker noise source when 
designing methods or collecting data. In particular, when gene-tree variance is dominant, 
adding more loci is most beneficial, while when substitution noise dominates, longer sequences or imputation
are needed. Finally, leveraging the derived covariance matrix, we implement a Gaussian-sampling procedure to 
generate split support values for METAL trees and demonstrate empirically that this approach yields more reliable 
confidence estimates than traditional bootstrapping.
}

\section{Introduction}

A core problem in phylogenomics is the reconstruction of a species tree from gene alignments. 
A species tree encodes the evolutionary relationships and divergence times among given species, while gene trees represent the evolutionary histories of individual genes. The structure of a gene tree may differ from the species tree due to variation induced by processes such as incomplete lineage sorting (ILS),  horizontal gene transfer and hybridization \cite{maddison1997gene}. Finally, gene alignments are comparisons of DNA, RNA, or protein sequences across different species, from which gene trees may inferred. The task of species tree reconstruction is to infer the species tree from gene alignments. \\

\begin{figure}[!h]
    \centering
    \includegraphics[width=0.8\linewidth]{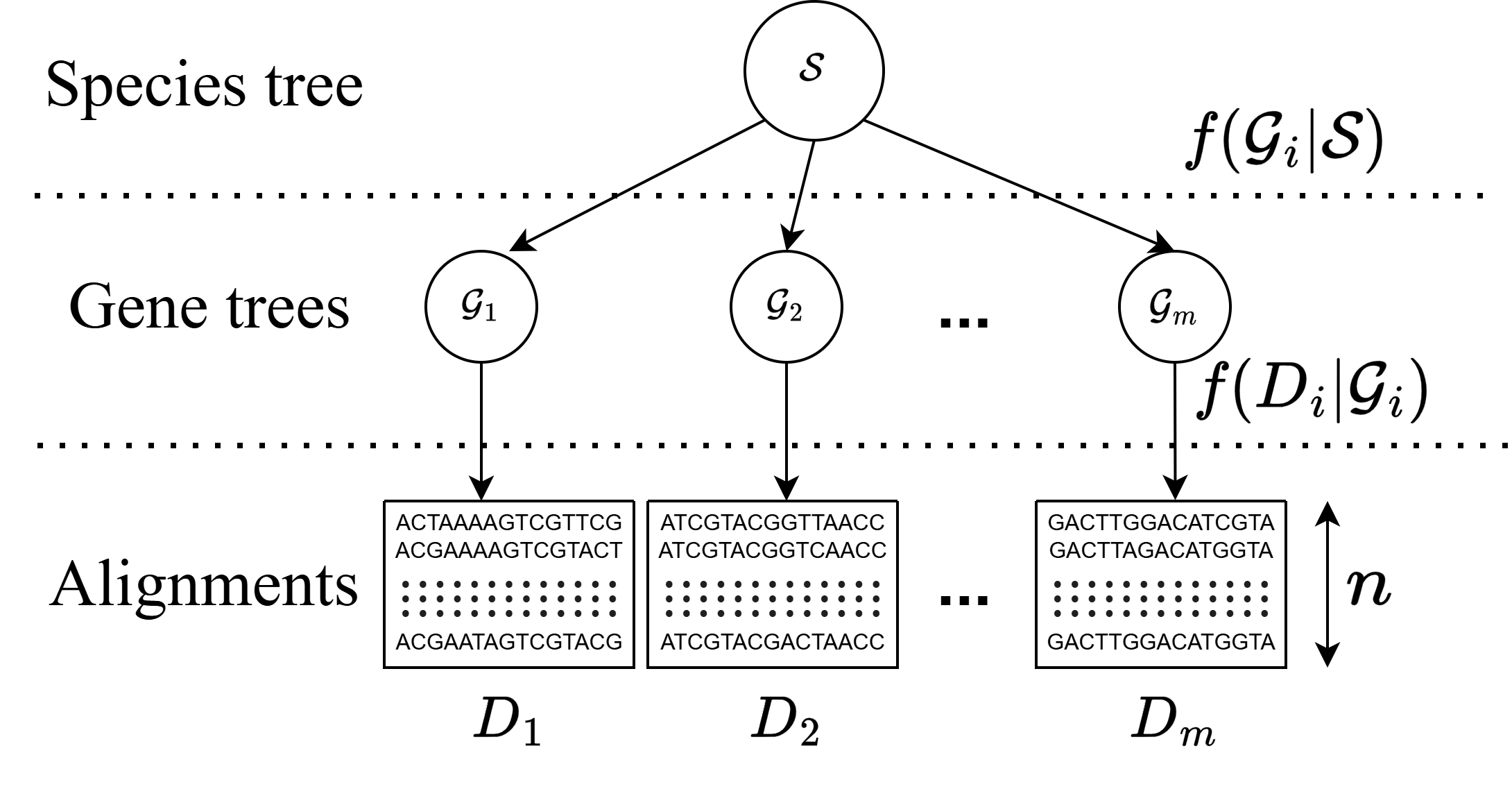}
    \caption{In phylogenetics, the goal is to reconstruct gene trees \(G_i\) for \(i \in [m]\) (intermediate level) from \(m\) gene alignments (bottom level) through a
    nucleotide substitution model $f(D|G)$.
    In phylogenomics, the aim is to infer the species tree \(\mathcal{S}\) (top level) from which these gene trees arose through the MSC model $f(G|\mathcal{S})$.}
    \label{fig:phylogenomics_as_hierarchical_model}
\end{figure}
\par
Hence, species tree reconstruction is commonly formulated as a hierarchical statistical model with three levels, as illustrated in Fig.~\ref{fig:phylogenomics_as_hierarchical_model}. At the top level, the species tree \(\mathcal{S}\) is the parameter of ultimate interest, to be estimated either as a point estimate or via a (approximate) posterior distribution. 
Given \(\mathcal{S}\), the second level consists of \(m\) gene trees \(\{G_i : i \in [m]\}\), whose likelihoods \(f(G_i \mid \mathcal{S})\) are defined by the multispecies coalescent (MSC) model, which captures the stochasticity introduced by incomplete lineage sorting (ILS) under a given species tree~\cite{degnan2009gene}. 
At the bottom level are the gene alignments \(\{D_i : i \in [m]\}\), whose likelihoods \(f(D_i \mid G_i)\) depend only on the corresponding gene trees, and are computed using nucleotide substitution models such as the Jukes-Cantor (JC69) model \cite{jukes1969evolution}, the 
Hasegawa-Kishino-Yano (HKY) model \cite{hasegawa1985dating}, and the Generalized Time Reversible (GTR) model \cite{tavare1986some}.
In practice, we observe only the gene alignments and must infer both the gene trees \(\{G_i: i\in[m]\}\) and more crucially for phylogenomics, the species tree \(\mathcal{S}\).\\
\par
Viewing a species tree reconstruction via a hierarchical model has several conceptual advantages: 
it clarifies the distinct sources of randomness (under coalescent and substitution models),
provides a foundation for Bayesian approaches \cite{yang2012molecular}, and facilitates the development of statistical methods for quantifying uncertainty at each level.
Moreover, it enables the integration of multiple data types and model components within a coherent inferential framework, as seen in widely used tools such as MrBayes \cite{ronquist2012mrbayes}, BEAST \cite{bouckaert2019beast},
and BPP \cite{flouri2018species}. 

\subsection{Species Tree Reconstruction Methods}

In contrast to computationally intensive methods for species tree reconstruction, distance-based approaches based on summary statistics---such as METAL~\cite{dasarathy2014data}, GLASS~\cite{mossel2008incomplete}, and STEAC~\cite{liu2009estimating}---offer simpler and more computationally efficient alternatives. These methods rely on estimated coalescent times or Hamming distances between species pairs, and then define summary statistics over these estimates across all genes. GLASS, which uses the \textit{minimum} coalescent times across genes as its summary statistic, can provide more accurate species tree estimates when gene tree uncertainty is minimal. In contrast, METAL and STEAC, which rely on \textit{averages} as their summary measure, tend to be more robust in the presence of greater variance in gene tree estimates or forms of gene tree heterogeneity other than ILS. \\
\par
Both STEAC and GLASS were originally developed as methods that take gene trees as input
and produce an inferred species tree as output. Their application in phylogenomics 
thus depends on the accurate reconstruction of gene trees from gene sequence data. 
In contrast, 
METAL was specifically designed to infer species trees directly from gene sequences
without requiring intermediate gene-tree estimation.
As a result, METAL 
typically achieves accurate species tree reconstruction with fewer loci.
Notably, when working directly with sequence data, METAL requires only 
 $\mathcal{O}(f^{-2})$ genes to recover the correct species tree topology with a 
a specified high probability,
where $f$ is the shortest internal branch length in the species tree. 
By comparison, GLASS and STEAC require $\mathcal{O}(f^{-3})$ and $\mathcal{O}(f^{-4})$ 
genes respectively~\cite{dasarathy2014data}. 
However, GLASS can outperform METAL in regimes when 
the number of bases per gene is large; as substitutional noise diminishes,
the advantage of GLASS's minimum distance criterion become more pronounced.\\
\par 
Despite the computational efficiency and practical appeal of the aforementioned methods, 
a rigorous understanding of how the two primary sources of uncertainty-- coalescent variation versus
substitutional noise-- combine to affect distance-based species tree inference has been lacking. 
In particular, it remains unclear under which parameter regimes methods like GLASS, STEAC or METAL
will be most reliable, since the relative magnitude of substitutional noise versus gene-tree variance
can shift dramatically with mutation rate, sequence length, and tree height. \\
\par 
Although bootstrap resampling of concatenated sequences is the de facto standard for 
assessing split support, it conflates coalescent and substitutional uncertainties and can yield misleading confidence values under extreme variance regimes. By deriving the full covariance matrix 
of METAL's pairwise distances, we introduce 
a Gaussian-sampling procedure that produces more accurate split support estimates than 
traditional bootstrapping.

\subsection{Contributions}
\par 
Considering the hierarchical nature of species tree reconstruction, 
uncertainty arises from two distinct sources; (1) variation in tree topologies 
governed by MSC, and (2) sequence-level stochasticity introduced by the substitution process used to infer gene trees from alignments (e.g., as JC69). By explicitly deriving the covariance matrices corresponding to each source, we 
obtain a quantitative framework for understanding how uncertainty propagates through 
distance-based estimates. 
In particular, we demonstrate that METAL outperforms GLASS
when substitutional variance dominates the total uncertainty, and vice versa. \\
\par
This decomposition enables informed decisions about study design, such as whether 
to prioritize collecting longer sequences or more independent loci. When coalescent
variance dominates, increasing the number of loci is most effective; when 
substitutional noise is dominant, longer alignments (perhaps via imputation strategies)
are preferable. \\
\par
In Section~\ref{sec:reconstruction_from_gene_trees}, we analyze species tree reconstruction from known gene trees under METAL, GLASS, and STEAC. We compute the covariance structure of METAL and STEAC, and show that METAL effectively interpolates between STEAC and GLASS as a function of mutation rate. In Section~\ref{sec:reconstruction_from_bases}, we study species tree reconstruction from gene base sequences, emphasizing the statistical behavior of METAL and STEAC under finite sequence lengths.
Section~\ref{sec:theory} presents our theoretical results: the decomposition of the total covariance of METAL and STEAC into coalescent and substitutional components, analysis of the spectrum of the METAL covariance matrix, and asymptotic expressions quantifying their relative magnitudes across a range of model parameters.
Section~\ref{sec:sim_results} contains two empirical applications. Section~\ref{sec:STE} compares the performance of METAL, GLASS, and STEAC across regimes where substitutional or coalescent variance dominates, showing that the relative error of GLASS versus METAL depends critically on the proportion of variance attributable to substitutional noise. In Section~\ref{sec:BPC}, we introduce a Gaussian-sampling procedure for computing support values for METAL trees, enabled by our covariance derivations. We show empirically that this method provides more reliable split confidence than traditional bootstrapping.
Finally, Section~\ref{sec:discussion} concludes the paper with a discussion of the implications, limitations, and future directions of this work.

\section{Species tree reconstruction from gene trees} \label{sec:reconstruction_from_gene_trees}
First, we introduce the notation that will be used in the rest of the paper.
The task is to successfully reconstruct the species tree $S = (V,E)$, where 
$V$ and $E$ is the set of nodes and edges in the tree respectively. 
The leaf set $L \subseteq V$ corresponds to the $n=|L|$ extant species examined.
We define $\tau_e$ to be the length of $e \in E$ in coalescent units and the smallest branch length as 
$f = \min_{e \in E} \{\tau_e \}$ and the diameter of the tree as 
$\Delta = \max_{a,b \in L} \tau_{ab} $. \\
\par
For simplicity, we sample a single individual per species, following the approach implemented in agglomerative, distance-based frameworks such as METAL \cite{braun2024testing}. 
We consider $m$ loci with associated gene trees
$G^{(i)}, i \in [m]$ 
with the same leaf set $L$, and define $g_{ab}^{(i)}$ as twice the time
in coalescent units to the most recent ancestor of $a$ and $b$ in $G^{(i)}$.
This is the evolutionary distance between the two leaves.
In this section, we assume that the gene trees are known. 
Each of the three methods examined in this paper uses a dissimilarity map
$d: L^2 \to \mathbb{R}_ +$ and subsequently performs hierarchical clustering on that
distance matrix. The resulting dendrogram is an estimate of the species tree.
Pseudocode for these algorithms can be found in Appendix~\ref{sec:algorithms}, Algorithm~\ref{alg:species_tree_reconstruction}.\\
\par 
The dissimilarity maps are defined as follows:
\begin{align}
    d^{\rm (GLASS)}(a,b) & = \min_{i \in [m]} g_{ab}^{(i)}, \label{eq:glass_def} \\
    d^{\rm (STEAC)}(a,b) & = \frac{1}{m} \sum_{i=1}^{m} g_{ab}^{(i)}, \label{eq:steac_def} \\
    d^{\rm (METAL)}(a,b) & = \frac{1}{m} \sum_{i=1}^{m} p_{ab}^{(i)} = 
    \frac{3}{4}\left( 
        1 - \frac{1}{m} \sum_{i=1}^{m}  \exp\left( -  \mu g_{ab}^{(i)} \right) 
    \right), \label{eq:metal_def} 
\end{align}
for all pairs $(a,b) \in \binom{L}{2}$  
where $p_{ab}^{(i)}$ is the probability that the bases of species $a$ and $b$ differ at any given 
site in locus $i$
under the Jukes-Cantor evolutionary model, and $\mu$ is rate of substitution. \\
\par 
Note that METAL does not require prior knowledge of the mutation rate~$\mu$ when applied to empirical sequence alignments. In practice, one simply computes the normalized Hamming distance $\hat{p}_{ab}^{(i)}$ between alignments of species~$a$ and~$b$ for each locus~$i$, 
which is the proportion of sites that differ between those two alignments,
and then averages over loci to obtain
\[
\hat{d}^{\rm (METAL)}(a,b) = \frac{1}{m} \sum_{i=1}^{m} \hat{p}_{ab}^{(i)}.
\]
As the number of sites per locus tends to infinity, $\hat{p}_{ab}^{(i)} \to p_{ab}^{(i)}$, and thus $\hat{d}^{\rm (METAL)}(a,b) \to d^{\rm (METAL)}(a,b)$. Under the Jukes--Cantor model~\cite{jukes1969evolution}, the $p$-distance $p_{ab}$ 
is related to the Jukes-Cantor distance $d_{\rm JC}(a,b)$, which is measured in expected number
of substitution time units, through
\[
    p_{ab } = \frac{3}{4} \left( 1- \exp\left(-\frac{4}{3} d_{\rm JC}(a,b) \right)\right),
\]
which, yields Equation~\eqref{eq:metal_def} by noting that 
$d_{\rm JC}(a,b) = \frac{3}{4} \mu g_{ab}$.\\
\par 
Proposition~\ref{thm:metal_to_steac_and_glass} shows that, under perfect gene‐tree knowledge, the METAL distance “interpolates’’ between STEAC and GLASS: as the substitution rate $\mu\to0$, it reduces to STEAC’s average coalescent time, whereas as $\mu\to\infty$, it collapses to GLASS’s minimum‐coalescent‐time criterion.\\

\begin{proposition}
    \label{thm:metal_to_steac_and_glass}
    Let $\Hat{T}_{\rm METAL}(\mu), 
    \hat{T}_{\rm STEAC}, 
    \hat{T}_{\rm GLASS}$ be the reconstructed species tree topologies of each method as described in Algorithm \ref{alg:species_tree_reconstruction}, 
    assuming that all methods use the same hierarchical clustering method, and that all the minima attained by
    GLASS are unique.
    Then, 
    \begin{align*}
        \lim_{\mu \to 0} \Hat{T}_{\rm METAL}(\mu) &= \hat{T}_{\rm STEAC}, \\
        \lim_{\mu \to \infty} \Hat{T}_{\rm METAL}(\mu) &= \hat{T}_{\rm GLASS}.
    \end{align*}
    For the second limit, it is assumed that the hierarchical clustering 
    method is single or complete linkage clustering.
\end{proposition}
\begin{proof}
    All proofs can be found in Appendix \ref{sec:proofs}.
\end{proof}

Proposition~\ref{thm:positive_cov} then establishes that, for any four leaves $a,b,c,d\in L$, the pairwise distances produced by STEAC and by METAL are both nonnegatively correlated, with STEAC correlations always at least as large as those of METAL.  The covariance matrix of METAL pairwise distances—assembled from the individual covariance components derived in the proof of Proposition~\ref{thm:positive_cov}—will be used for uncertainty quantification in Section~\ref{sec:theory} and to formulate the total covariance in our simulation results (Section~\ref{sec:sim_results}).

\begin{proposition}
    \label{thm:positive_cov}
    Let $a,b,c,d \in L$ be leaves in the species tree. Then, under MSC, 
    \begin{equation*}
        \label{eq:positive_cov}
        \Cor\left(g_{ab}^{(1)}, g_{cd}^{(1)} \right) \geq 
        \Cor\left(e^{t g_{ab}^{(1)}}, e^{t g_{cd}^{(1)}} \right)
        \geq 0,
        \forall t\leq 0,
    \end{equation*}
    where $\Cor$ denotes the correlation between two random variables. 
    Moreover, $\Cor\left(e^{t g_{ab}^{(1)}}, e^{t g_{cd}^{(1)}} \right)$ as a function of 
    $t$ is continuous and strictly increasing on $(-\infty,0)$ with limit
    \begin{align*}
        \lim_{t \to 0} \Cor\left(e^{t g_{ab}^{(1)}}, e^{t g_{cd}^{(1)}} \right) &=
    \Cor\left(g_{ab}^{(1)}, g_{cd}^{(1)} \right). \\
    \end{align*}
    Finally, 
    \begin{align*}
         \Cor\left( d^{\rm (STEAC)}(a,b), d^{\rm (STEAC)}(c,d) \right)&= 
         \Cor\left(g_{ab}^{(1)}, g_{cd}^{(1)} \right) \\
         \Cor\left( d^{\rm (METAL)}(a,b), d^{\rm (METAL)}(c,d) \right)
         &= \Cor\left(e^{t g_{ab}^{(1)}}, e^{t g_{cd}^{(1)}} \right)
    \end{align*}
\end{proposition}

\section{Gene tree reconstruction from base sequences} \label{sec:reconstruction_from_bases}
In the previous section, we considered properties of the distribution 
of gene trees $G_i | S$ under the MSC model.  
In this section, we investigate statistical properties of gene base sequence
conditional on the gene tree or $\chi^{i}| G_i$.
Let $\chi_{a}^{ij} \in \{A,C,T,G\}$ be the $j^{\rm th}$ base of the $i^{\rm th}$ locus of species $a \in L$, where $i \in [m]$, $j \in [K_i]$, $m$ is the number of loci studied and 
$K_i$ is the sequence length of loci $i$. 
For notational simplicity in the formulae that follow, we assume that all 
loci have the same number of bases $K = K_i, ~ \forall i \in [m]$.
The analysis can be readily extended to the case of loci with varying sequence lengths.
The normalized Hamming distance, $\hat{p}^{(i)}_{ab}$ and Jukes-Cantor distance $\Hat{\nu}_{ab}^{(i)}$ between species $a, b\in L$ are then defined as 
\begin{align*}
    \Hat{p}_{ab}^{(i)} &= \frac{1}{K} \sum_{j=1}^K \mathbbm{1}\left( 
        \chi_{a}^{ij} \neq  \chi_{b}^{ij}
        \right), \text{ and}\\
    \nu_{ab} &= \frac{3}{4}\mu\Hat{g}_{ab}^{(i)} = - \frac{3}{4} \log\left( 1 - \frac{4}{3} \Hat{p}_{ab}^{(i)} \right).
\end{align*}

METAL computes the average of the normalized Hamming distances, while STEAC and GLASS consider
the average and minimum Jukes-Cantor distances respectively. 
However, a problem with using the Jukes-Cantor distances directly is that they are infinite if
$\Hat{p}_{ab}^{(i)} \geq \frac{3}{4}$ for any $a,b \in L$, $i \in [m]$, which has a non-zero
chance of occurring. 

\subsection{Hamming Distances}
\par Given the gene tree $G^{(i)}$ of locus $i$, the distribution of 
$\Hat{p}^{(i)}_{ab}$ is
\begin{align*}
    \Hat{p}_{ab}^{(i)} = \frac{1}{K} \sum_{j=1}^K \mathbbm{1}\left( 
        \chi_{a}^{ij} \neq  \chi_{b}^{ij}
        \right)
    &\sim 
    {\rm Bin}\left(K, p_{ab}^{(i)}\right)
     , \text{ where}\\
    p_{ab}^{(i)} &= \frac{3}{4}\left( 1 - 
    \exp\left(-\mu g_{ab}^{(i)}  \right) \right).
\end{align*}
Effectively, $p_{ab}^{(i)}$ is the true probability that species 
$a$ and $b$ would share the same base in locus $i$ and 
$\Hat{p}_{ab}^{(i)}$.
The number of bases, $K$, is typically large and so we can apply the common Normal approximation of the Binomial distribution, 
\begin{equation*}
    \Hat{p}_{ab}^{(i)}\sim  
    \mathcal{N}\left( p_{ab}^{(i)}, \frac{p_{ab}^{(i)} \left(1-p_{ab}^{(i)}\right) 
     }{K}\right).
\end{equation*}

\begin{figure}[!h]
    \centering
    \includegraphics[width=0.4\linewidth]{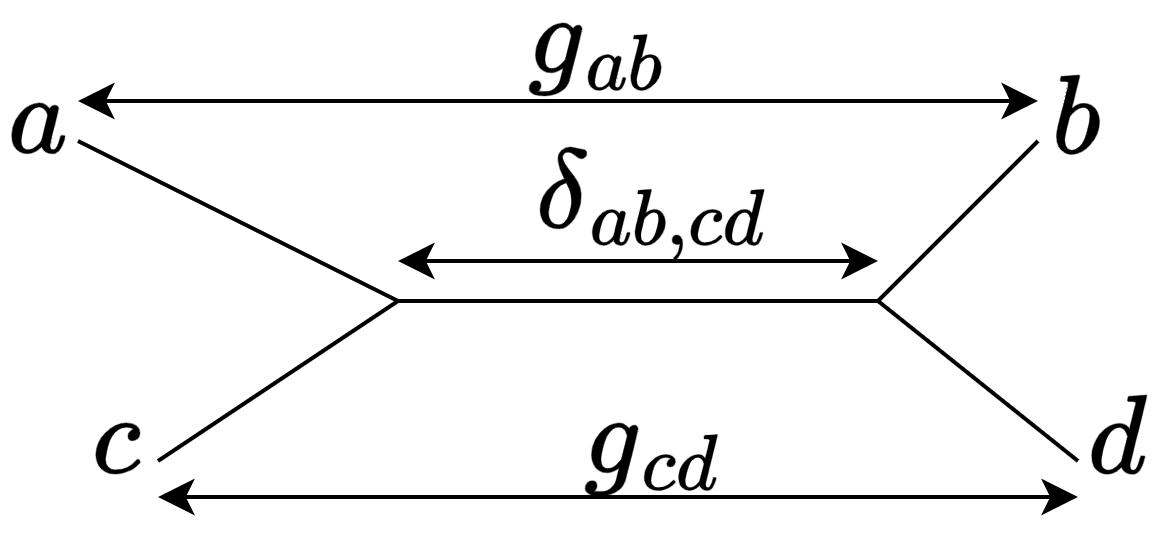}
    \caption{The evolutionary tree used in proposition \ref{thm:evol_cov}. Here $\delta_{ab,cd} = \delta_{ac,bd}$ 
    is the length of the intersection of the 
    shortest path from $a$ to $b$ and the shortest 
    path from $c$ to $d$. 
    All trees with 4 distinct leaves $a,b,c,d \in L$ can be 
    drawn in this way. For cherry trees, the pairs $(a,c), (b,d)$ 
    are sisters/cherries and the root of the 4-leaf tree is 
    contained in the $\delta$-segment.
    For comb trees, without loss of generality, the topology is 
    $(( (a,c),b),d) $ with $a,c$ being sisters and the root of the
    tree contained in the shortest path from $b$ to $d$. 
    Note that since the paths from $a$ to $c$ and $b,d$ are disconnected $\delta_{ac,bd} =0 $. }
    \label{fig:evolutionary_tree}
\end{figure}

While the normal approximation above applies to each pairwise distance $\hat p_{ab}^{(i)}$ individually, we require a joint multivariate description of all $\binom{L}{2}$ distances at locus $i$.  In fact, the multivariate 
extension of that result is
\[
\hat{p}_{\binom{L}{2}}^{(i)} \;\sim\; \mathcal{N}\Bigl(p_{\binom{L}{2}}^{(i)},\,\tfrac{1}{K}\,\Sigma^{(i)}\Bigr),
\]
where $\hat p_{\binom{L}{2}}^{(i)},\,p_{\binom{L}{2}}^{(i)}\in\mathbb{R}^{\binom n2}$ are vectors containing the estimated and true pairwise distances respectively and $\Sigma^{(i)}\in\mathbb{R}^{\binom n2\times\binom n2}$ is the covariance matrix whose entries are of the form given by Proposition~\ref{thm:evol_cov}.  \\

\begin{proposition} 
    \label{thm:evol_cov}
    Let $a,b,c,d \in L$ 
    be leaves of a gene tree, where the pairs $a,c$ and $b,d$ are sisters
    as shown in Fig.~\ref{fig:evolutionary_tree}. Letting $\delta_{ab,cd}$ be the length of the intersection of the shortest path from $a$ to $b$ the shortest path from $c$ to $d$, 
    the covariance is 
    \begin{equation}
        \label{eq:cov_bases}
        \Cov \left(\mathbbm{1}\left(\chi_a^{(i)} \neq \chi_b^{(i)} \right),
        \mathbbm{1}\left(\chi_c^{(i)} \neq \chi_d^{(i)} \right) \right) = 
        \frac{3}{16} 
        e^{-\mu \left(g_{ab}^{(i)} + g_{cd}^{(i)} \right)}
        \left(e^{2\mu \delta_{ab,cd}^{(i)}} + 2 e^{\mu \delta_{ab,cd}^{(i)}} - 3
        \right)
    \end{equation}
\end{proposition}
Using Proposition \ref{thm:evol_cov}, we see in particular, that the diagonal entries of $\Sigma^{(i)}$ coincide with the univariate variances: by 
expanding Equation~\eqref{eq:cov_bases} and noting that $\delta_{ab,ab}=g_{ab}$, we have,
\[
\Var\bigl(\mathbbm{1}(\chi_a^{(i)}\neq\chi_b^{(i)})\bigr)\;=\;
p_{ab}^{(i)}\left(1-p_{ab}^{(i)}\right)\,.
\] 
\subsection{Jukes-Cantor Distances}
\par For the STEAC method, instead of using normalized Hamming distances,
$\hat{p}_{ab}$, it uses estimates for the gene tree branch lengths
\begin{equation}
    \label{eq:steac_branch_lengths}
        \Hat{g}_{ab} = - \frac{1}{\mu} \log\left( 1- \frac{4}{3} \Hat{p}_{ab} \right),
        \forall a,b \in L.
\end{equation}

The major issue with this approach is that if at any loci $i \in [m]$ we get 
$\hat{p}_{ab}^{(i)} \geq 3/4$, then $\hat{g}_{ab}^{(i)} = \infty$ and so 
$\hat{d}_{ab}^{\rm (STEAC)} = \infty$. Clearly, this is undesirable, especially when the
number of genes is large $m \gg 1$ and the number of sites per gene $K$ is low, when 
such events are most common.
Hierarchical clustering methods such as UPGMA would not allow clusters containing
leaves $a$ and $b$ to ever merge. \\
\par 
In this paper we consider the case when $K \gg 1$, when the Delta method is typically 
used to approximate the distribution of $\Hat{g}_{ab}$ in Equation~\eqref{eq:steac_branch_lengths}, which dictates that
\begin{align}
    \hat{g}_{\binom{L}{2}}^{(i)} & \sim \mathcal{N}\left( g_{\binom{L}{2}}^{(i)}, \Sigma_{\rm STEAC}^{(i)} \right), \nonumber\end{align}
where  $\hat{g}_{\binom{L}{2}}^{(i)}$, ${g}_{\binom{L}{2}}^{(i)} \in \mathbb{R}^{\binom{n}{2}}$ are vectors containing the estimated and true Jukes-Cantor distances respectively, and the row $ab$, column $cd$ element of the covariance matrix $\Sigma_{\rm STEAC}^{(i)}$ is \begin{align}
    \Cov\left( \hat{g}_{ab}^{(i)}, \hat{g}_{cd}^{(i)} \right) \nonumber 
    &\approx
    \frac{ \Cov\left( \mathbbm{1}\left( \chi_a^{(i)} \neq \chi_b^{(i)} \right), \mathbbm{1}\left( \chi_c^{(i)} \neq \chi_d^{(i)} \right) \right) }
    {K \mu^2 \left( \frac{3}{4} - p_{ab}^{(i)} \right)\left( \frac{3}{4} - p_{cd}^{(i)} \right) }\\& = 
    \frac{e^{2\mu \delta_{ab,cd}^{(i)}} + 2e^{\mu \delta_{ab,cd}^{(i)}} - 3}{3\mu^2 K},
    \label{eq:steac_sigma}
\end{align}
where the last equality follows from Proposition~\ref{thm:evol_cov} and the definition of
$p_{ab}^{(i)}$. 
In the following section, we make use of these derived covariance matrices 
to compute the total covariance matrix of the METAL and STEAC pair distances.

\section{Theoretical results} \label{sec:theory}
Two main topics are addressed in this section. 
First, we derive and decompose the total covariance of 
the METAL and STEAC estimators into coalescent and substitution components, thereby quantifying the relative contributions of coalescent/genealogical and substitution/mutational uncertainty. 
Second, we analyze the spectral properties of these covariance matrices.

\subsection{Decomposition of Covariance}
We begin by defining the total covariance of METAL distance estimates as
\begin{equation}
    \left(\Sigma_{\rm total}\right)_{ab,cd} \coloneqq
    \Cov\left(\hat{d}^{\rm (METAL)}(a,b), \hat{d}^{\rm (METAL)}(c,d) \right) = 
    \frac{1}{m} \Cov\left(\hat{p}_{ab}^{(1)}, \hat{p}_{cd}^{(1)} \right),
    \label{eq:sigma_total_def}
\end{equation} 
where $\Sigma_{\rm total} \in \mathbb{R}^{\binom{n}{2} \times \binom{n}{2}}$.
The following proposition shows how the total covariance can be decomposed as 
a coalescent and substitution covariance. \\
\begin{proposition} \label{prop:total_covariance_decomposition}
    The total covariance can be expressed as 
    \begin{equation}
        \Sigma_{\rm total} = \frac{1}{m}
        \Sigma_{\rm coal} + \frac{1}{mK}\Sigma_{\rm sub},
        \label{eq:sigma_total_decomposition}
    \end{equation}
    where 
    \begin{align}
        \left(\Sigma_{\rm coal}\right)_{ab,cd}  &= \frac{9}{16} \Cov\left(e^{-\mu g_{ab}^{(1)}}, e^{-\mu g_{cd}^{(1)}}\right), \label{eq:sigma_coal}  \\
        \left(\Sigma_{\rm sub}\right)_{ab,cd}  &= \frac{3}{16} 
         \mathbb{E}_{G^{(1)}| \mathcal{S}} \left( 
            e^{-\mu \left(g_{ab}^{(i)} + g_{cd}^{(i)} \right)}
        \left(e^{2\mu \delta_{ab,cd}^{(i)}} + 2 e^{\mu \delta_{ab,cd}^{(i)}} - 3
        \right)
            \right) , \label{eq:sigma_evol} 
    \end{align}
    and
    $\Sigma_{\rm coal}, \Sigma_{\rm sub}, \Sigma_{\rm total} \in 
    \mathbb{R}^{\binom{n}{2} \times \binom{n}{2}}$ are positive semi-definite 
    matrices. 
\end{proposition}
Expressions for $\Sigma_{\rm coal}$ defined in Equation~\eqref{eq:sigma_coal}
have been derived in the proof of Proposition~\ref{thm:positive_cov}. 
Specific expressions for $\Sigma_{\rm sub}$ defined in Equation~\eqref{eq:sigma_evol}, for all tree topology cases can be 
found in Appendix \ref{sec:cov_calculations}. \\

The decomposition in Proposition~\ref{prop:total_covariance_decomposition} sets the stage for 
Proposition~\ref{prop:covariance_coal_sub_total}, in which we specialize to the case of a star-shaped
species tree and derive explicit formulae and asymptotics for each component of covariance.
These results allow use both to compare the magnitude of coalescent versus substitution uncertainty
and to study how their ratio varies with the mutation rate $\mu$ and tree diameter $\Delta$. \\


\begin{proposition} \label{prop:covariance_coal_sub_total}
    Suppose that the species tree is a star tree with diameter $\Delta$ 
    i.e. $\tau_{ab} = \Delta, \forall a,b \in L=  [n], a\neq b$.
    Then
    \begin{equation*}
        (\Sigma_{\rm mode})_{ab,cd} = \begin{cases}
        \sigma^{(2)}_{\rm mode}(\mu,\Delta), & \mbox{if } |\{a,b\} \cap \{c,d\}| = 2 \\
        \sigma^{(1)}_{\rm mode}(\mu,\Delta), & \mbox{if } |\{a,b\} \cap \{c,d\}| = 1 \\
        \sigma^{(0)}_{\rm mode}(\mu,\Delta), & \mbox{if } |\{a,b\} \cap \{c,d\}| = 0 ,
    \end{cases}
    \end{equation*}
    for ${\rm mode} \in \{ {\rm coal, sub, total} \}$ with
    \begin{equation*}
        \sigma_{\rm total}^{(i)}(\mu,\Delta) = \frac{1}{m} \left( 
            \sigma_{\rm coal}^{(i)}(\mu,\Delta) + \frac{1}{K} \sigma_{\rm sub}^{(i)}(\mu,\Delta)
        \right), \, \forall i \in \{0,1,2\},
    \end{equation*}
    where $m$ is the number of loci and $K$ is the number of sites per loci.\\
    Moreover, $\sigma_{\rm mode}^{(2)} \geq \sigma_{\rm mode}^{(1)}
        \geq \sigma_{\rm mode}^{(0)} > 0$, and 
    \begin{align*}
        \sigma_{\rm coal}^{(i)} &= \begin{cases}
            \mathcal{O}(\mu^2), & \text{ as } \mu \to 0, \\
            \mathcal{O}(\mu^{-3+i} e^{-2\mu\Delta}), & \text{ as } \mu \to \infty
        \end{cases} \\
        \sigma_{\rm sub}^{(i)} &= \begin{cases}
            \mathcal{O}(\mu), & \text{ as } \mu \to 0, \\
            \mathcal{O}(\mu^{-2+i} e^{(-2+i)\mu\Delta}), & \text{ as } \mu \to \infty
        \end{cases} \\
    \end{align*}
    for all $i \in \{0,1,2\}$ and consequently
    \begin{align*}
        \frac{\sigma_{\rm sub}^{(i)} }{\sigma_{\rm coal}^{(i)}} &= 
        \mathcal{O}(\mu^{-1}) \,\,\, \text{ as } \mu \to 0, \forall i \in \{0,1,2\},
        {\rm and }\\ 
        \frac{\sigma_{\rm sub}^{(i)} }{\sigma_{\rm coal}^{(i)}} &= 
        \mathcal{O}(\mu  e^{i \mu \Delta } ) \,\,\, \text{ as } \mu \to \infty, \forall i \in \{0,1,2\}. \\ 
    \end{align*}
\end{proposition}

Proposition~\ref{thm:metal_to_steac_and_glass} shows that in the low‐mutation regime 
(\(\mu\ll1\)), STEAC and METAL agree to first order, since 
\(
p_{ab} \;=\;\tfrac{3}{4}\bigl(1-e^{-\mu g_{ab}}\bigr)\approx \tfrac{3}{4}\,\mu\,g_{ab}.
\)
However, Proposition~\ref{prop:steac_infinite_variance} demonstrates that as soon as \(\mu\ge\tfrac14\), 
the STEAC estimates have a strictly positive probability of diverging.  
Thus, even at moderate mutation rates, STEAC can yield unbounded 
branch‐length estimates—and with many loci this becomes highly probable, 
since a single locus producing an infinite estimate is enough to drive the STEAC average to infinity.
Although one can truncate or discard these infinite estimates, doing so introduces 
downward bias by systematically removing the highest inferred branch lengths.
We conjecture that, in the presence of substitution uncertainty—i.e.\ when gene trees 
must be estimated rather than known a priori—STEAC is uniformly inferior to METAL for all \(\mu\); 
indeed, STEAC was originally developed under the assumption of almost perfect gene‐tree knowledge.  
We assess this conjecture empirically in Section~\ref{sec:STE}. \\

\begin{proposition} \label{prop:steac_infinite_variance}
    For any two distinct leaves $a,b \in L$,
    the variance of the STEAC estimate 
    \(\hat{g}_{ab}\), as defined in Equation~\eqref{eq:steac_branch_lengths} and approximated 
    via the Delta method in Equation~\eqref{eq:steac_sigma}, is finite if and only if 
    \(\mu < \tfrac{1}{4}\). Consequently, when \(\mu \ge \tfrac{1}{4}\), there is a positive 
    probability 
    \[
        p \;=\; \mathbb{P}\bigl(\hat{g}^{(i)}_{ab} = \infty \bigr)
        \;=\; \mathbb{P}\bigl(\hat{p}_{ab}^{(i)} \ge \tfrac{3}{4}\bigr) \;>\; 0.
    \]
    In that case, the probability that all \(m\) STEAC gene branch lengths 
    \(g_{ab}^{(i)}\) for \(i \in [m]\) contributing to \(\hat{d}^{\rm (STEAC)}_{ab}\) are finite is
    \[
        \mathbb{P}\bigl(\hat{d}^{\rm (STEAC)}_{ab} < \infty \bigr)
        \;=\; (1 - p)^{m} \to 0 \quad {\rm as } \quad m \to \infty,
    \]
    where $m$ is the number of genes used to infer the STEAC estimate.
\end{proposition}

\subsection{Principal components of variance}
As we have seen in Proposition~\ref{prop:covariance_coal_sub_total},
the coalescent, substitution and total covariance matrices of the METAL
estimator for the star species tree only depend on 
$|\{a,b\} \cap \{c,d\}| $. In this subsection we study the spectrum of 
these matrices. Specifically, for a covariance matrix
$C \in \mathbb{R}^{\binom{n}{2} \times \binom{n}{2}}$, with entries
\begin{equation}
    C_{ab,cd} =  {\rm Cov}(\Hat{d}_{ab}, \Hat{d}_{cd}) = \begin{cases}
        \alpha , & \mbox{if } |\{a,b\} \cap \{c,d\}| = 2 \\
        \beta, & \mbox{if } |\{a,b\} \cap \{c,d\}| = 1 \\
        \gamma, & \mbox{if } |\{a,b\} \cap \{c,d\}| = 0 ,
    \end{cases}
    \label{eq:c_def}
\end{equation}
for pairs $(a,b), (c,d) \in \binom{L}{2} $
where $\alpha > \beta > \gamma > 0$, Proposition~\ref{thm:spectrum}
derives its spectrum. \\

\begin{proposition}
    \label{thm:spectrum}
    The matrix $C \in \mathbb{R}^{ {n \choose 2} \times  {n \choose 2} }, n\geq 3$, defined in 
    Equation~\eqref{eq:c_def}, 
    has eigenvalues 
    \begin{itemize}
        \item the dominant eigenvalue 
        $\gamma{n \choose 2} +
        2(\beta-\gamma)(n-1) + \gamma + \alpha - 2\beta$ 
        with multiplicity $1$ and corresponding
        first principal component ${\rm PC1} ={\bf 1}$, 
        \item $(\beta-\gamma)(n-2) +\gamma + \alpha - 2 \beta $ with 
        corresponding principal components ${\rm PC}2, \dots, {\rm PC}{n}$
        dependent on $n$ and
        independent of $\alpha,\beta,\gamma$, and
        \item $\gamma + \alpha - 2\beta$ with multiplicity $n(n-3)/2$.
    \end{itemize}
    Note that these eigenvalues are presented in decreasing magnitude
    and that $C$ is positive definite if and only if 
    $\alpha + \gamma > 2\beta$.
    It follows that the variance ratio of the principal component tends to 
    $\frac{\gamma}{\alpha}$ as $n \to \infty$. 
    The variance ratio of the first $n$ most dominant 
    eigenvectors tends to $\frac{2\beta- \gamma}{\alpha}$
    as $n \to \infty$. \\
\end{proposition} 

\par Proposition~\ref{thm:spectrum} shows that each covariance matrix in our model has exactly three distinct eigenvalues.  In particular, the all‐ones vector 
$
\mathbf{1} 
$
is an eigenvector corresponding to the largest eigenvalue.  From the perspective of tropical geometry, $\mathbf{1}$ plays a canonical role: the tropical projective torus 
$
\mathbb{R}^{\binom{n}{2}} \big/ \mathbb{R}\mathbf{1}
$
is precisely the quotient that identifies any vector $\mathbf{v}$ with its translate $\mathbf{v} + c\,\mathbf{1}$ for all real $c$.  Equivalently, adding a constant to every coordinate does not change the point in the toric quotient. 
Statistical analysis on the tropical projective torus has been explored in various works, e.g.~\cite{aliatimis2024tropical, barnhill2023clustering, lee2022tropical}. \\
\par Corollary~\ref{cor} (to Proposition~\ref{thm:spectrum}) further quantifies that up to one‐third of the total variance of our distance‐vector estimates lies along the $\mathbf{1}$‐direction. 
Since the tropical projective torus “modes out’’ precisely the $\mathbf{1}$‐direction, restricting our inference to $\mathbb{R}^{\binom{n}{2}} / \mathbb{R}\mathbf{1}$ removes this single principal‐component subspace. As a result, the remaining variance—now confined to the $(\binom{n}{2}-1)$‐dimensional tropical torus—drops by up to one‐third.  Working in this quotient removes a large portion of extraneous variation while preserving all information relevant to tree topology, since adding any constant multiple of \(\mathbf{1}\) to the distance vector does not alter the inferred tree topology. \\
\par
Moreover, Corollary~\ref{cor} shows that when \(\mu \ll 1\), the relative magnitude of substitution versus coalescent uncertainty scales with \(\mu K\), the expected number of substitutions per locus: as \(\mu K\) increases, substitution noise becomes negligible.  In particular, if 
\begin{equation}
    K^{-1} \ll \mu \ll 
    \frac{1}{2\Delta} \log\left(
        \frac{1+\sqrt{8K\Delta}}{2}
    \right)
    ,
    \label{eq:coalescent_interval}
\end{equation}
then the coalescent covariance term dominates the total uncertainty.  
Note that as the number of bases per gene $K$ increases, the interval widens 
on both ends, whereas increasing $\Delta$ narrows the interval from the upper bound.
Figure~\ref{fig:coal+sub=total} confirms this behavior: substitution‐driven variance is largest for very small and very large \(\mu\), whereas for intermediate \(\mu\) the coalescent variance is predominant.  Figure~\ref{fig:sub/coal} plots the ratio of Frobenius norms \(\|\Sigma_{\rm sub}\|_F/\|\Sigma_{\rm coal}\|_F\), which attains its minimum at \(\mu = \mathcal{O}(1/\Delta)\) when \(\Delta\) is large.  Finally, Figure~\ref{fig:variance_ratio} displays the cumulative variance explained by the leading \(n\) principal components of \(\Sigma_{\rm sub}\).  Remarkably, for \(\mu \ll 1\) and large \(\Delta\), these \(n\) components capture nearly all of the METAL estimator’s variance, while for larger \(\mu\) the explained fraction vanishes, in agreement with Proposition~\ref{prop:coal_infinite_tree}.\\
\par 
It is worth mentioning that Proposition~\ref{prop:coal_infinite_tree} further establishes that as \(\mu\to\infty\), the limiting coalescent correlation matrix, as defined in Equation~\eqref{eq:cor_coal},
has rank \(n-1\), and its nonzero principal components are the split-indicator vectors 
which are exactly the generators of the unique maximal cone of the ultrametric fan that contains
the species-tree ultrametric. In other words, all variation of the METAL distance vector lies within the cone corresponding to the true species‐tree ultrametric. \\

\begin{corollary} \label{cor}
    Consider the 
    matrices $\Sigma_{\rm coal}(\mu), \Sigma_{\rm sub}(\mu), \Sigma_{\rm total}(\mu) \in 
    \mathbb{R}^{\binom{n}{2} \times \binom{n}{2}}$ defined in 
    Equations~\eqref{eq:sigma_total_def}, \eqref{eq:sigma_coal}, \eqref{eq:sigma_evol}
    for a star species tree with diameter $\Delta$.
    They have the following asympototic properties
    \begin{equation*}
        \frac{ \mathrm{Tr}(\Sigma_{\rm sub} )}{ \mathrm{Tr}(\Sigma_{\rm coal} )}, 
        \sqrt[\binom{n}{2}]{ \frac{\det(\Sigma_{\rm sub} )}{ \det(\Sigma_{\rm coal} )}}, 
        \frac{\norm{\Sigma_{\rm sub}}_2 }{\norm{\Sigma_{\rm coal}}_2 },
        \frac{\norm{\Sigma_{\rm sub}}_{\rm F} }{\norm{\Sigma_{\rm coal}}_{\rm F} }= 
        \begin{cases}
            \mathcal{O}(\mu^{-1}K^{-1}), & \text{ as } \mu \to 0, \\
            \mathcal{O}(\mu e^{2\mu\Delta} K^{-1}
            ), & \text{ as } \mu \to \infty
        \end{cases} \\
    \end{equation*}
    As $\mu \to 0$ and $n \to \infty$,
    the first principal component PC1={\bf 1} explains $2/(6+3\Delta)$ of the variance of 
    $\Sigma_{\rm evol}$ and $\Sigma_{\rm total}$, 
    and $2/9$ of the variance of $\Sigma_{\rm coal}$,
    while the first $n$ principal components explain $(4+3\Delta)/(6+3\Delta)$ 
    and $4/9$ of the variance 
    respectively. 
\end{corollary}

\begin{figure}[ht]
    \centering
    \includegraphics[width=\linewidth]{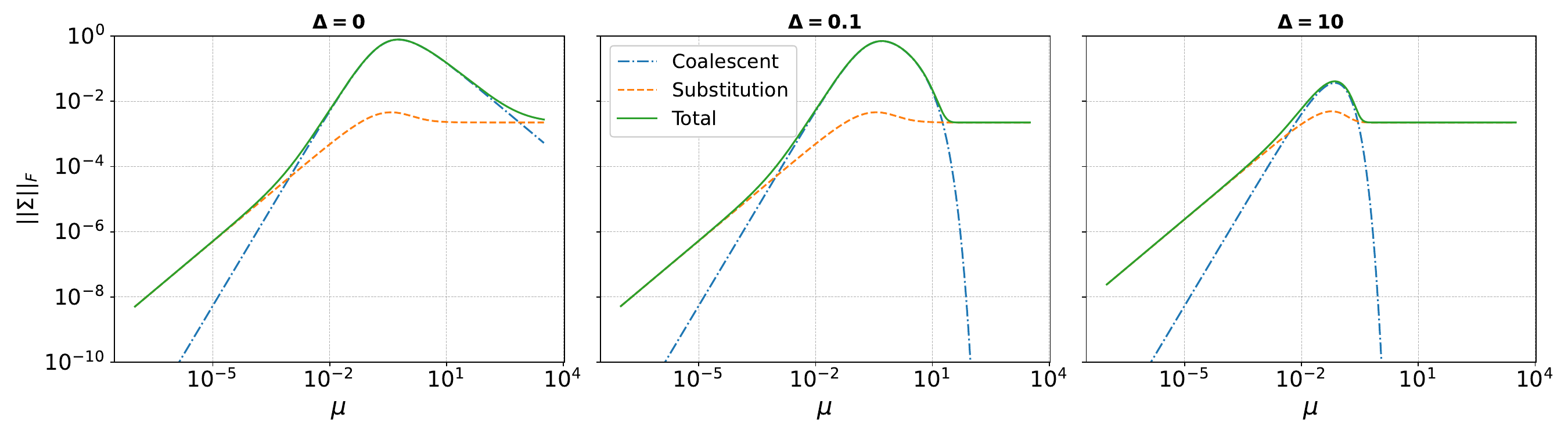}
    \caption{Frobenius norms of \(\Sigma_{\rm coal}(\mu)\), \(\Sigma_{\rm sub}(\mu)\), and \(\Sigma_{\rm total}(\mu)\) as functions of the mutation rate \(\mu\), shown for three values of the species‐tree diameter \(\Delta\) with \(K = 1000\). For \(\mu \ll 1\) and for \(\mu \gg 1\), the substitution variance \(\Sigma_{\rm sub}\) dominates; in the intermediate \(\mu\) range, the coalescent variance \(\Sigma_{\rm coal}\) prevails. As \(\Delta\) increases, the interval in which \(\Sigma_{\rm coal}\) dominates becomes narrower.}
    \label{fig:coal+sub=total}
\end{figure}

\begin{figure}[ht]
    \centering
    \includegraphics[width=.7\linewidth]{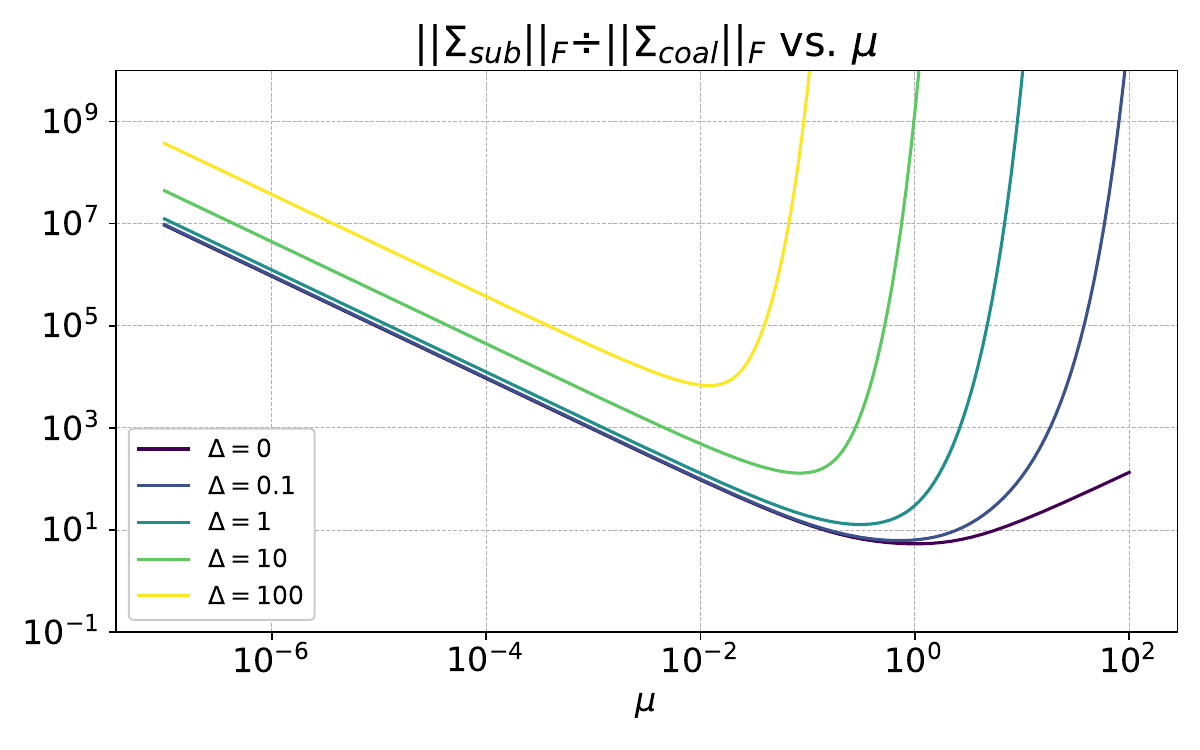}
    \caption{Ratio \(\|\Sigma_{\rm sub}(\mu)\|_F \big/ \|\Sigma_{\rm coal}(\mu)\|_F\) plotted against \(\mu\) for several values of \(\Delta\). Substitution variance (\(\Sigma_{\rm sub}\)) dominates when \(\mu \ll 1\) or \(\mu\,\Delta \gg 1\), whereas coalescent variance  dominates in the intermediate \(\mu\) regime.}
    \label{fig:sub/coal}
\end{figure}

\begin{proposition} \label{prop:coal_infinite_tree}
    Consider a species tree $S$ with positive pair distances $\tau_{i,j} = \rho_{i,j} \Delta$, 
    where $\Delta$ is the diameter of the tree and $\rho_{ij}  \in (0,1]$. 
    From the covariance matrix defined in Equation~\eqref{eq:sigma_coal}, if we
    construct the correlation matrix $C_{\rm coal}(\Delta) = D_{\rm coal}^{-1} \Sigma_{\rm coal}(\Delta) D_{\rm coal}^{-1} $, 
    where $D_{\rm coal} = \sqrt{{\rm diag}(\Sigma_{\rm coal}(\Delta))}$ and consider 
    the limiting case for large trees, then we let 
    \begin{equation}
        C_{\rm coal}^{\infty} \coloneqq \lim_{\Delta \to \infty} C_{\rm coal}(\Delta).
        \label{eq:cor_coal}
    \end{equation}
    Let $i \in V\backslash L$ 
    be an internal node of the species tree $S$, and the sets $L_i,R_i \subset L = [n]$ are the sets of leaf nodes left and right of $i$  respectively. Define the vector, 
    \begin{align*}
        v^{(i)}_{kl} & \coloneqq
        \mathbb{I} \left(
            i \text{ is the most recent common ancestor of } k \text{ and } l 
        \right)
        \\ &= 
        \mathbb{I} \left( 
            (k,l) \in (L_i \times R_i) \cup (R_i \times L_i) 
        \right), 
    \end{align*}
which are also known as split-indicator vectors.
Then, $C_{\rm coal}^{\infty} v^{(i)} = |L_i| |R_i|  v^{(i)} $   i.e. $v^{(i)}$ is an eigenvector with corresponding eigenvalues $|L_i| |R_i|$.  These are exactly $n-1$ principal components corresponding to distinct internal nodes in $V\backslash L$. 
The remaining eigenvalues are zero, and so ${\rm rank}(C_{\rm coal}^{\infty}) = n-1$.
Finally, let the species tree pairwise distance vector
$d^{\rm S} \in \mathbb{R}^{\binom{n}{2} \times \binom{n}{2}}$, 
$d^{\rm S}_{ab} = \tau_{ab},\, \forall a,b \in L$ be an ultrametric i.e.
$d \in \mathcal{U}_n$ and suppose it does not lie on topological boundaries. 
Then,
there exists $\delta_i > 0$ such that for all $\epsilon_i \in (0,\delta_i)$, 
\begin{equation*}
    d^{\rm S} + \sum_{i = 1}^{n-1} \epsilon_i v^{(i)} \in  \mathcal{U}_n
\end{equation*}
Finally, $C_{\rm evol}^{\infty} = C_{\rm total}^{\infty} = I_{\binom{n}{2}}$, where 
the correlation matrices are defined as above, and 
$I_N \in \mathbb{R}^{N\times N}$ is the identity 
matrix.
\end{proposition}

\begin{figure}[htbp]
    \centering
    \includegraphics[width=0.9\linewidth]{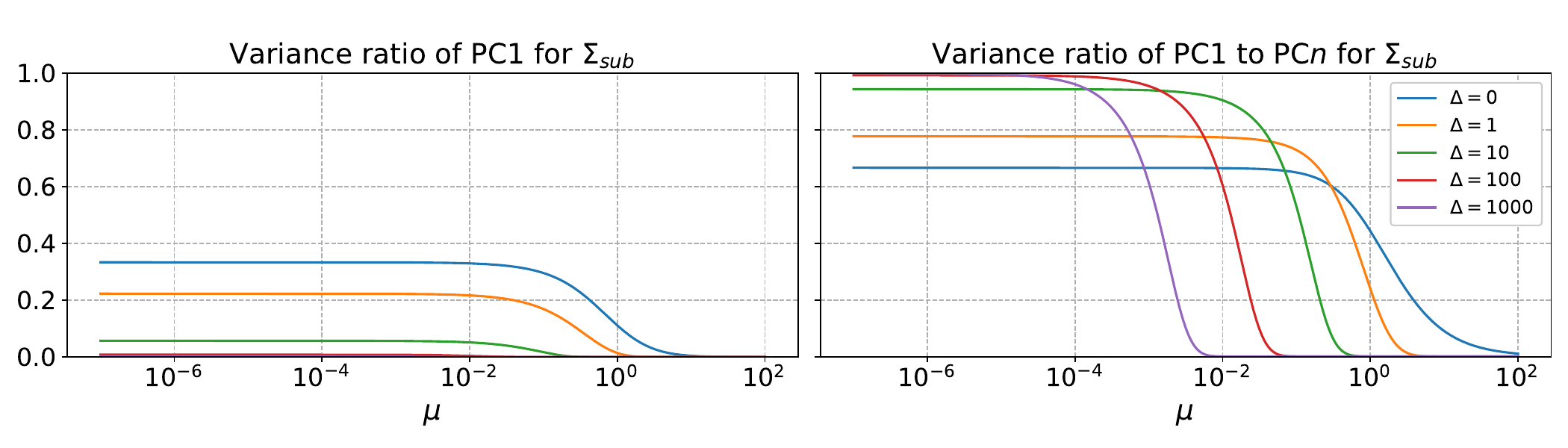} \\
    \caption{Variance ratios for \(\Sigma_{\rm sub}\) plotted against mutation rate \(\mu\) and species‐tree diameter \(\Delta\).  (Left) Fraction of total variance explained by the first principal component.  (Right) Fraction of total variance explained by the first \(n\) principal components.  For corresponding plots of \(\Sigma_{\rm coal}\) and \(\Sigma_{\rm total}\), see Figure~\ref{fig:complete_variance_ratio} in the Appendix.}
    \label{fig:variance_ratio}
\end{figure}

\section{Simulation Results} \label{sec:sim_results}
Two applications of the derived covariance matrix in phylogenetic analysis are considered in this
section. In the first application (Section \ref{sec:STE}), our derivations of the covariance matrix enable
us to quantify uncertainty from the substitution model, which predicts the relative 
performance of GLASS against METAL. In the second application (Section \ref{sec:BPC}), the covariance 
matrix is used to compute node‐support values as confidence for each bipartition of 
the estimated METAL tree. We compare our method to the standard bootstrap approach.
Our repository can be found in \cite{aliatimis2024metal}.




\subsection{Species‐Tree Estimation} \label{sec:STE}

We begin by simulating data under the multi‐species coalescent (MSC) combined with a site‐substitution model.  Using default parameters 
\[
    \mu = 1,\quad \Delta = 1,\quad K = 100,\quad n = 10,\quad m = 100,
\]
we first generate a species tree of diameter \(\Delta\) (in coalescent units) via MSC.  Then \(m\) gene trees are sampled under MSC on that species tree.  Finally, sequences of length \(K\) are evolved along each gene tree under a substitution model with rate \(\mu\), producing \(m\) alignments. \\
\par From these \(m\) alignments, we infer the species tree using three approaches (see Algorithm~\ref{alg:species_tree_reconstruction}): STEAC, METAL, and GLASS.  In STEAC and GLASS, 
each gene’s pairwise Hamming distances \(\hat p^{(i)}_{ab}\) are transformed via the inverse Jukes–Cantor formula to branch‐length estimates \(\hat g^{(i)}_{ab}\). STEAC then averages 
across all \(m\) genes before applying UPGMA. 
GLASS takes the minimum branch gene tree branch lengths
before applying UPGMA.  METAL concatenates all \(m\) alignments, computes Hamming distances on the full superalignment, and reconstructs a UPGMA tree from that distance matrix. \\
\begin{figure}[ht]
    \centering
    \includegraphics[width=0.8\linewidth]{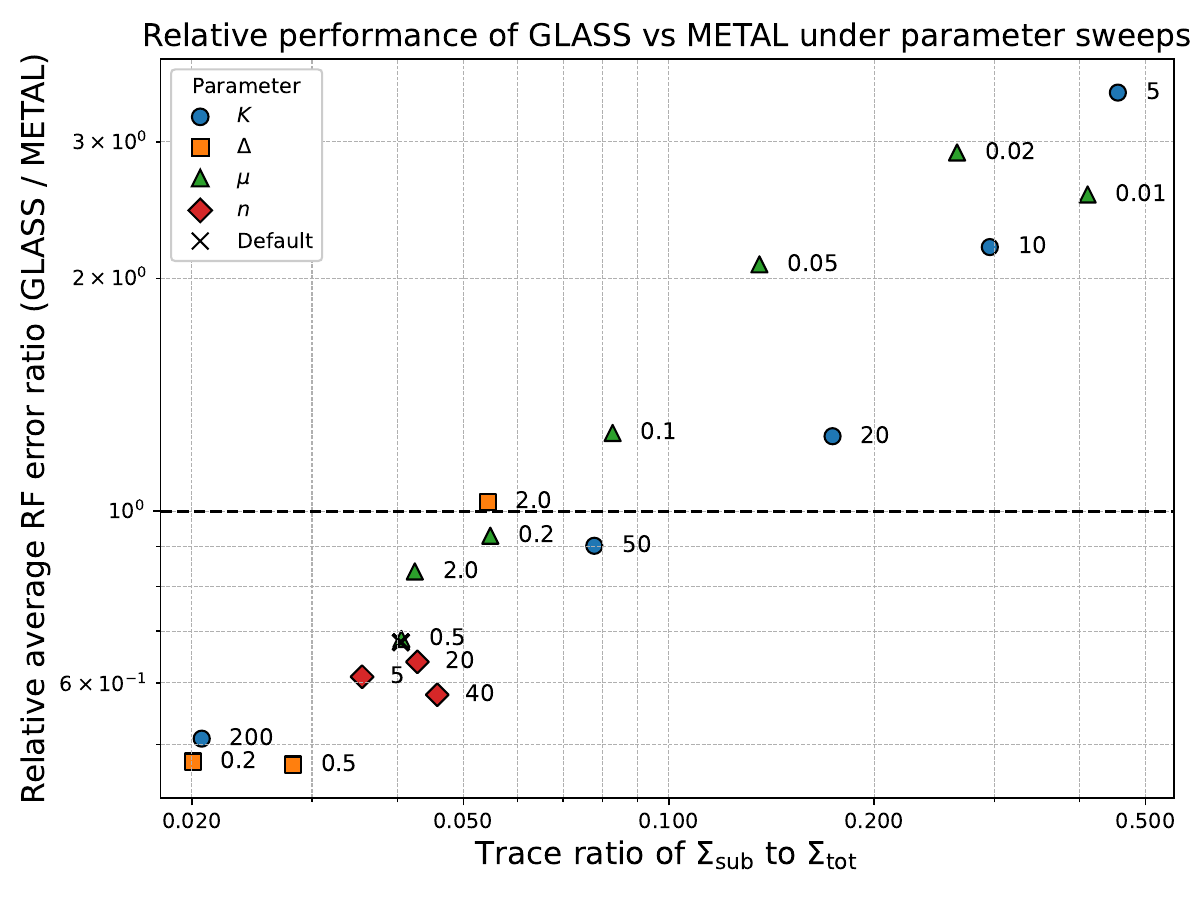}
    \caption{Comparison of GLASS and METAL performance as a function of the ratio 
    \(\mathrm{tr}(\Sigma_{\mathrm{sub}})/\mathrm{tr}(\Sigma_{\mathrm{total}})\). 
    When substitution‐model variance constitutes a large fraction of total variance, METAL outperforms GLASS.}
    \label{fig:glass_vs_metal}
\end{figure} 

\par Figure~\ref{fig:glass_vs_metal} illustrates how the relative accuracy of GLASS versus METAL depends on the fraction of variance contributed by the substitution model.  As this ratio increases, METAL (which explicitly pools substitution noise across genes) consistently outperforms GLASS. To quantify reconstruction accuracy, each inferred tree \(\widehat T_{\mathrm{method}}\) is compared against the true species tree \(T_{\mathrm{true}}\) via the Robinson–Foulds (RF) distance.
Over \(R\) replicates (here \(R=100\)), we obtain RF distances \(\mathrm{RF}_{\mathrm{STEAC}}^{(r)},\,\mathrm{RF}_{\mathrm{METAL}}^{(r)},\,\mathrm{RF}_{\mathrm{GLASS}}^{(r)}\).  The mean RF error for each method is
\[
    \overline{\mathrm{RF}}_{\mathrm{method}}
    \;=\; \frac{1}{R}\sum_{r=1}^{R} \mathrm{RF}_{\mathrm{method}}^{(r)}.
\]

\par We consistently observe 
\[
    \overline{\mathrm{RF}}_{\mathrm{STEAC}} \;>\; \overline{\mathrm{RF}}_{\mathrm{METAL}},
\]
so STEAC underperforms METAL in all tested regimes (see Figure~\ref{fig:steac_vs_metal} in the Appendix).  In contrast, 
\[
    \frac{\overline{\mathrm{RF}}_{\mathrm{GLASS}}}{\overline{\mathrm{RF}}_{\mathrm{METAL}}}
\]
increases with 
\(\mathrm{tr}(\Sigma_{\mathrm{sub}})/\mathrm{tr}(\Sigma_{\mathrm{total}})\).  When substitution variance dominates, GLASS yields higher RF error than METAL.  However, if coalescent variance from MSC is the primary source (for instance, when \(K\) is large or \(\mu\) lies in an intermediate regime), GLASS approaches the Kullback–Leibler‐optimal species‐tree estimator and outperforms METAL, i.e.
\[
    \overline{\mathrm{RF}}_{\mathrm{GLASS}} \;<\; \overline{\mathrm{RF}}_{\mathrm{METAL}}.
\] \\

\par In summary, the ratio
\[
    \frac{\mathrm{tr}(\Sigma_{\mathrm{sub}})}{\mathrm{tr}(\Sigma_{\mathrm{total}})}
\]
determines the preferable method: if this ratio exceeds a threshold (when \(\mu\) is very low or very high, \(\Delta\) is large, or \(K\) is small), METAL—by aggregating substitution noise—yields lower RF error.  Otherwise, for moderate \(\mu\) and large \(K\), the coalescent term dominates and GLASS (leveraging accurate gene‐tree estimates) gives superior accuracy.

\subsection{Bipartition Confidence} \label{sec:BPC}

The procedure begins with data generation: we first simulate a species tree, then simulate gene trees under the multi‐species coalescent (MSC) model using \texttt{DendroPy}~\cite{Moreno2024}, and finally perform sequence evolution along those gene trees using \texttt{Pyvolve}~\cite{SpielmanWilke2015}. This pipeline is illustrated in Fig.~\ref{fig:phylogenomics_as_hierarchical_model}. \\
\par Once we have simulated data, we obtain the METAL estimate by computing Hamming distances between the concatenated sequences. These distances serve as pairwise dissimilarities, and we reconstruct a tree via the UPGMA algorithm, following the METAL methodology. \\
\par Our main objective is to assign confidence scores to tree splits (bipartitions). The standard approach uses bootstrapping: sites from the concatenated alignment are resampled (with replacement) to generate pseudo‐replicate alignments. For each replicate, we compute a new METAL tree, yielding a collection of bootstrap trees. The proportion of replicates in which a given split appears provides its bootstrap support. \\
\par As an alternative, we propose a Gaussian‐sampling approach. The METAL estimator of the average pairwise distances can, for a sufficiently large number of genes, be approximated by a multivariate normal distribution. The mean of this distribution is the METAL estimate itself, and the covariance is derived from the METAL tree (in coalescent units to scale Hamming distances appropriately). The quality of this normal approximation improves as the number of genes increases. \\
\par We draw samples from this multivariate distribution to obtain synthetic vectors of pairwise distances. Each sampled distance matrix is then converted into a tree via UPGMA, exactly as in the bootstrap procedure. We estimate split support by computing the frequency of each bipartition across these “Gaussian‐sampled” trees. \\
\par Figure~\ref{fig:three_trees} provides an illustrative comparison. The leftmost tree is the METAL estimate with bootstrap‐derived support values, the middle tree shows the same METAL topology but with support values from Gaussian sampling, and the rightmost tree is the true species topology. Note that Gaussian sampling typically yields more conservative (lower) support values than bootstrapping. \\

\begin{figure}[h]
    \centering
    \includegraphics[width=\linewidth]{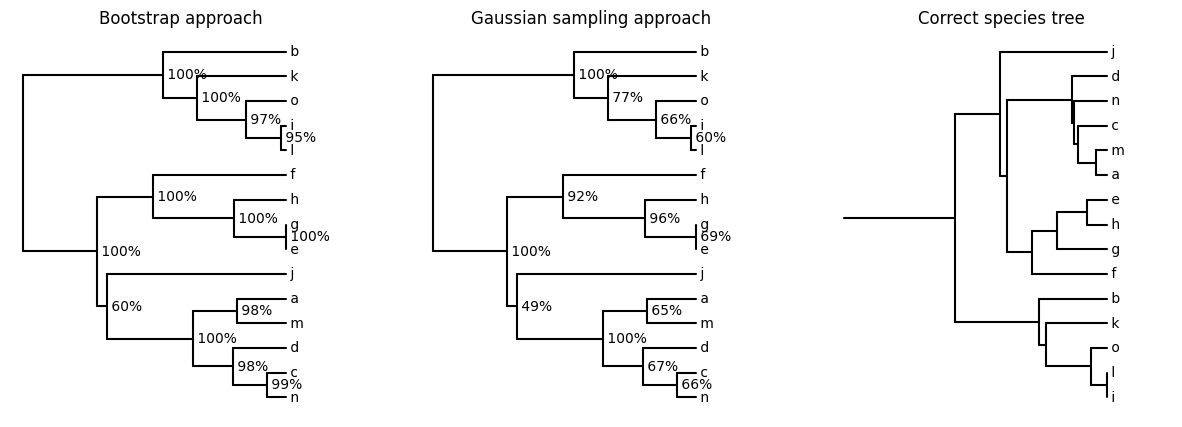}
    \caption{Example of METAL estimate (left) with bootstrap support (percentages at each node), METAL estimate (center) with Gaussian‐sampling support, and the true species tree (right). Bootstrap supports tend to be higher than Gaussian supports, but may also overestimate confidence for incorrect splits.}
    \label{fig:three_trees}
\end{figure}

To compare the two confidence‐estimation methods quantitatively, we exploit the fact that the true species tree is known in simulation. With $n = 40$ taxa, there are $n - 2 = 38$ non‐trivial splits. We label each split in the METAL estimate as correct (1) if it appears in the true tree, or incorrect (0) otherwise. Using the support scores assigned by each method, we compute the area under the ROC curve (AUC) to measure how well each set of support values discriminates correct splits from incorrect ones. \\

This experiment is repeated 100 times. In each series of 100 replicates, we vary exactly one simulation parameter (e.g., mutation rate $\mu$) while holding all other parameters (species‐tree diameter $\Delta$, number of sites per gene $K$, and number of genes $m$, number of taxa $n$) fixed. We then repeat separately for each of the other parameters ($\Delta$, $K$, $m$, and $n$), always keeping the remaining factors constant. \\

Figure~\ref{fig:auc_boxplots} shows boxplots of AUC scores for bootstrapping versus Gaussian sampling. In the left panel, we vary the number of genes $m$; in the right panel, we vary the number of sites per gene $K$. In both cases, the Gaussian‐sampling approach yields higher median AUC and lower variance than bootstrapping. 
Similar figures for the other parameters 
can be found in Appendix~\ref{sec:plots}. 
What is interesting about parameters $m$ and $K$ is that the 
concatenated sequence length has $mK$ bases which is 
As the number of genes \(m\) or the number of sites per gene \(K\) increases, the runtime of the bootstrap procedure grows in proportion to 
\(m \cdot K\), the size of the concatenated sequence. In contrast, the computational cost of the Gaussian‐sampling approach is essentially unaffected by changes in \(m\) or \(K\), since it depends only on the number of taxa \(n\). Moreover, although bootstrapping requires more time as \(m\) and \(K\) increase, its ability to correctly classify true splits does not improve—in fact, its classification accuracy becomes more erratic. \\

\begin{figure}[h]
  \centering
    \includegraphics[width=\textwidth]{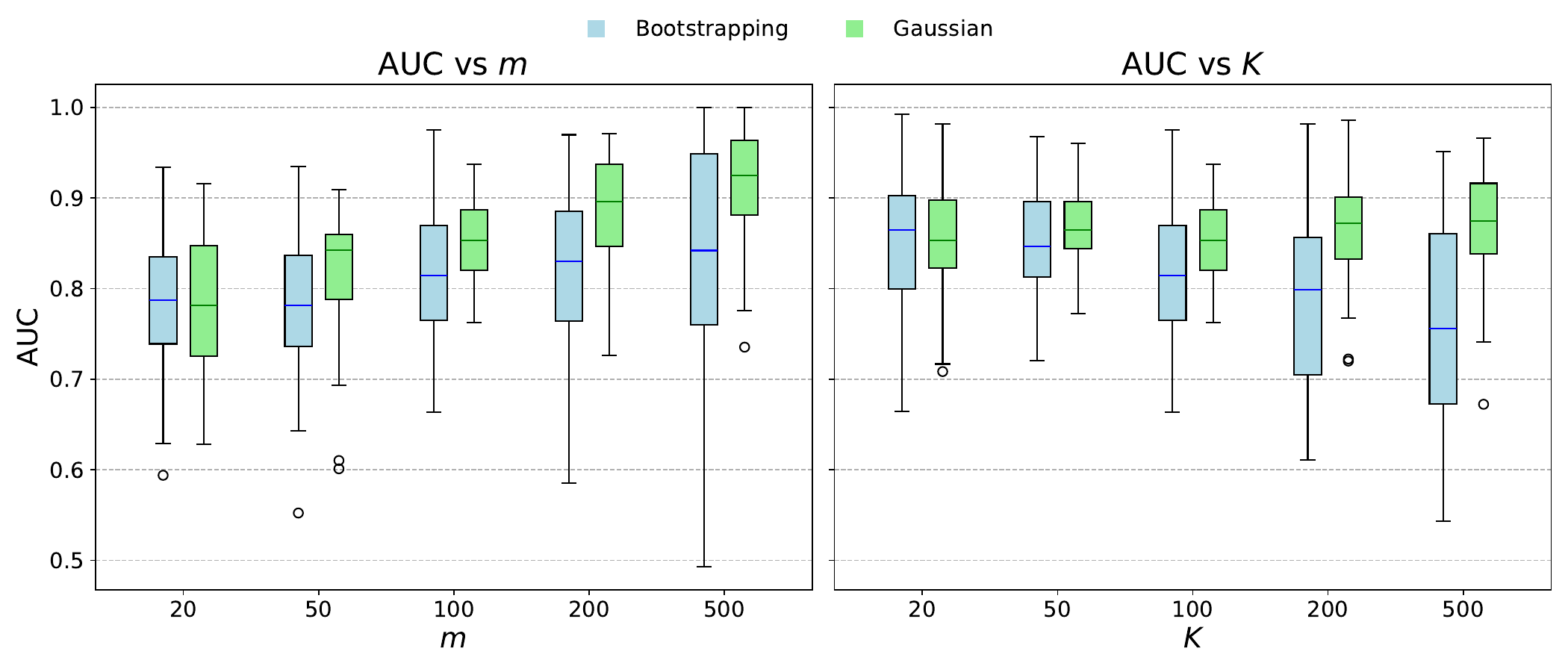}
  \caption{Boxplots of AUC values comparing bootstrap and Gaussian‐sampling support scores.  (Left:) AUC distributions as the number of genes $m$ varies (with $\mu$, $\Delta$, $n$, and $K$ held constant).  (Right:) AUC distributions as the number of sites per gene $K$ varies (with $\mu$, $\Delta$, $n$, and $m$ held constant). In both scenarios, Gaussian sampling outperforms bootstrapping in terms of median AUC and variability.}
  \label{fig:auc_boxplots}
\end{figure}

\subsubsection{Computational Complexity Comparison}
For bootstrapping, each replicate requires recomputing the METAL tree from a resampled alignment of total length $m \cdot K$. Computing all pairwise Hamming distances is $\mathcal{O}(n^2 \cdot m \cdot K)$, since there are $\binom{n}{2} = \mathcal{O}(n^2)$ taxon pairs. To ensure METAL is accurate with probability $1 - \varepsilon$, one needs 
\[
  m \;=\; O\!\bigl(\tfrac{\log(n/\varepsilon)}{f^2}\bigr),
\]
where $f$ is the shortest internal branch length. Since $f \leq \Delta / (n - 1)$, one obtains $m = \Omega(n^2 \log n)$. Hence each bootstrap replicate costs $\mathcal{O}(n^4 \log n \cdot K)$, and obtaining $B$ replicates is $B$ times that.

In contrast, the Gaussian sampling approach requires computing a covariance matrix over 
\[
  N = \binom{n}{2} \;=\; \mathcal{O}(n^2)
\]
entries, resulting in a full covariance matrix of size $N \times N$, containing $\mathcal{O}(n^4)$ elements. Computing the Cholesky decomposition of this matrix takes $\mathcal{O}(N^3) = \mathcal{O}(n^6)$ time, but only once. Each subsequent multivariate‐normal sample (to produce one synthetic distance vector) costs $\mathcal{O}(N^2) = \mathcal{O}(n^4)$. Crucially, this Gaussian method is independent of the sequence length $K$, making it far more efficient when $K$ is large.\\
\par
Nonetheless, when the number of taxa is very large $(n \gg1)$, 
the one-time cost of the Cholesky decomposition becomes prohibitive, and
bootstrapping may be more efficient. However, if the mutation rate 
is low 
$\mu \ll 1$ or if $\mu \Delta \gg 1$, we can instead leverage 
the theoretical results
of Section~\ref{sec:theory}, specifically Corollary~\ref{cor} and
Proposition~\ref{prop:coal_infinite_tree}. These results describe the
asymptotic spectrum of the total covariance matrix of the METAL estimator. 
In both of these limiting cases, we have shown 
in Proposition~\ref{prop:covariance_coal_sub_total}
that $\Sigma_{\rm coal} \ll \Sigma_{\rm sub}$, and so that 
$\Sigma_{\rm total} \approx \Sigma_{\rm sub}$.
For $\mu \ll 1$, there are three distinct eigenvalues with known variance ratios. As illustrated in Figure~\ref{fig:principal_components}, the asymptotic
predictions of the eigenvalues of Corollary~\ref{cor} 
closely match the empirical ones of a randomly generated species tree with
$\Delta=1, n =30, \mu=0.2$. 
Knowing the principal components in advance allows us to bypass the Cholesky 
decomposition entirely, eliminating the most computationally expensive step of the method.
\begin{figure}[h]
    \centering
    \includegraphics[width=0.6\linewidth]{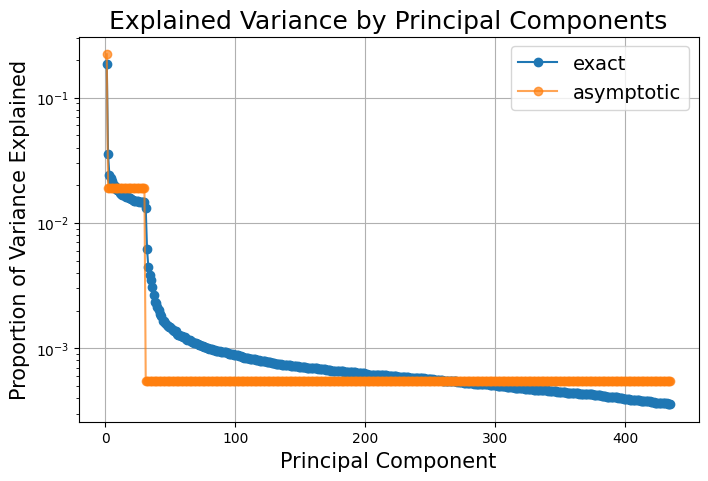}
    \caption{Explained variance by principal components of 
    the total covariance matrix 
    for a species tree with 
    $\Delta=1, n =30, \mu=0.2, K= 100$. The exact curve (blue) shows the empirical proportion of variance explained by each principal component, while the "asymptotic" curve (orange) represents theoretical predictions from Corollary~\ref{cor}. The close agreement validates the asymptotic spectrum in the low mutation rate regime ( \(\mu \ll 1\) ).}
    \label{fig:principal_components}
\end{figure}

\section{Discussion} \label{sec:discussion}

In this work, we explicitly formulate the covariant matrix of all pairwise distances under the MSC and substitution models for the species tree reconstruction problem.  
METAL constructs a distance matrix from concatenated multi-locus data~\cite{dasarathy2014data}, 
and then applies a clustering algorithm to infer the species tree. Our first main result is the
derivation of the exact covariance matrix of METAL's pairwise distance estimates. This covariance 
naturally decomposes into a term arising from coalescent genealogy variance and a term from sequence-substitution variance. We give explicit formulas for both components. We then analyze
the asymptotic regime as the number of loci grows large; we show that METAL distances concetrate
around expectations, and we study the spectrum of the limiting covariance matrix.
The spectral analysis suggests that much of the phylogenetic signal lies in the first few principal
components of the distance matrix when the number of leaves is large. \\
\par 
Our simulation study illustrates how the relative performance of distance-based methods depends on the underlying source of variation. In particular, METAL tends to outperform GLASS when the substitutional 
noise dominates, while GLASS becomes more accurate as 
coalescence variance takes over. This sensitivity reflects the 
design of each estimator; METAL effectively aggregates substitutional 
uncertainty across loci, whereas GLASS leverages gene-tree precision
when it is available. Likewise, our results show METAL retaining strong performance
even under rate variation and short branches. Finally, we introduce a new Gaussian-sampling procedure to estimate METAL split-support values, which uses the derived covariance to quantify support for inferred clades. \\
\par 
There are several natural extensions of our work. Our 
current analysis 
assumed a strict molecular clock, so that all lineages evolve at the same rate and
that the gene and species trees are equidistant. Relaxing this assumption is 
important for real data. Generalizing our covariance derivation to non-clock 
settings would be a major next step of our current analysis. \\
\par 
In conclusion, we have offered the first systematic decomposition of the covariance structure in distance-based 
methods like METAL and STEAC, separating the contributions of gene-tree (coalescent) variability from sequence-level
(substitutional) noise. Our framework not only illuminates how these two sources of uncertainty compare to each other, but also lays the groundwork for extending this analysis to a broader class of estimators. 
An important direction for future research is to determine the extent to which our findings on the balance 
between coalescent and substitutional effects hold under alternative models and estimators. 
Addressing this question will deepen our understanding of estimator performance across diverse biological scenarios.

\section*{Acknowledgments}

RY is partially funded by NSF Division of Mathematical Sciences: Statistics Program DMS 1916037. GA is funded by EPSRC through the STOR-i Centre for Doctoral Training under grant EP/L015692/1.

\begin{appendices}
\section{Proofs} \label{sec:proofs}

\begin{proof}[\bf Proof of Proposition \ref{thm:metal_to_steac_and_glass}]
    The definitions of the dissimilarity metrics can be found in Equations~\eqref{eq:glass_def}--\eqref{eq:metal_def}.
    We start by considering low mutation rates. Since, 
    \begin{equation*}
        \lim_{\mu \to 0}\frac{3}{4} \cdot \frac{1 - \exp\left( - \frac{4}{3} \mu  g_{ab}^{(i)}  \right) }{\mu} = g_{ab}^{(i)},\,\,
        \forall a,b \in L
    \end{equation*}
    and by expressing Equation~\eqref{eq:metal_def} as the partial sum of 
    terms $1 - \exp\left( - \frac{4}{3} \mu  g_{ab}^{(i)} \right)$, it follows that
    \begin{equation*}
        \lim_{\mu \to 0} \frac{d^{\rm (METAL)}(a,b)}{\mu} = \frac{1}{m} \sum_{i=1}^{m} g_{ab}^{(i)} = d^{\rm (STEAC)}(a,b), \,\,
        \forall a,b \in L .
    \end{equation*}
    Therefore, in the limit $\mu \to 0$, the pairwise METAL distances
    become proportional to the STEAC distances. 
    Since hierarchical clustering dendrograms are invariant under scalar multiplication of the distance matrix, 
    \begin{equation*} \lim_{\mu \to 0} \Hat{T}_{\rm METAL}(\mu) = \hat{T}_{\rm STEAC}.
    \end{equation*}
    For high mutation rates, note that if the GLASS minimum coalescence times are unique, i.e. 
    for any given pair of leaves $a,b \in L, a \neq b$
    there exists $i^* \in  [m]$ such that
    $\min_{j \in [m]} g_{ab}^{(j)} = g_{ab}^{(i^*)} < g_{ab}^{(i)}$ for all 
    $i \in [m]\backslash \{i^*\}$, then 
    \begin{equation*}
        \lim_{\mu \to \infty} m \frac{1- \frac{4}{3}d^{\rm (METAL)}(a,b)}
        {\exp\left(- \frac{4}{3} \mu d^{\rm (GLASS)}(a,b) \right)} = \lim_{\mu \to \infty}
        \sum_{i=1}^m \exp\left(-\frac{4}{3} \mu 
        \left(g_{ab}^{(i)}-g_{ab}^{(i^*)} \right)
        \right) = 1.
    \end{equation*}
    In other words,
    \begin{equation*}
        d^{\rm (METAL)}(a,b) = \frac{3}{4} \left( 1- 
        \left( \frac{1}{m} + o(1) \right)
        \exp\left(- \frac{4}{3} \mu d^{\rm (GLASS)}(a,b) \right) 
        \right).
    \end{equation*}
    Note that $f(x) = \frac{3}{4} \left( 1 - a \exp\left(- \frac{4}{3} \mu x \right)
    \right)$ is a strictly increasing function. 
    Hence, for a sufficiently large $\mu$, 
    \begin{equation*}
        d^{\rm (GLASS)}(a,b) < d^{\rm (GLASS)}(c,d) \Rightarrow 
        d^{\rm (METAL)}(a,b) < d^{\rm (METAL)}(c,d).
    \end{equation*}
    Since the number of combinations $a,b,c,d \in L$ is finite, there exists a sufficiently
    large $\mu$ such that the above equation holds for all $a,b,c,d \in L$. 
    Hence, the asymptotic ordering of METAL distances agrees with the order of GLASS distances. 
    Another way of reaching the same conclusion is using the monotone admissibility of 
    hierarchical clustering methods such as single linkage or complete
    linkage
    \cite{dugad1998unsupervised}, 
    which allows monotonic transformation of the elements of the distance matrix without 
    altering the clustering.
\end{proof}

\begin{proof}[{\bf Proof of Proposition \ref{thm:positive_cov}}]
Instead of writing $g_{ab}^{(1)}$ to refer to locus $1$ specifically, 
we write $g_{ab} $ for the remainder of the proof.
First, we address the trivial case where $a=b$ or $c=d$. 
Assume that wlog $a=b$, then $g_{ab}= 0$ and
$e^{t g_{ab}} = 1$, which implies that 
\[
    \Cov(e^{tg_{aa}}) = \Cov(g_{aa}) = 0,
\]
which implies that the corresponding correlations will be zero too, validating the results 
immediately in this case. 

For the rest of the proof, assume that 
$a \neq b$ and $c \neq d$.
\par
There are four cases to consider regarding the choice of the two pairs.
\begin{itemize}
    \item {\bf Two leaves only}, where wlog $a=c$ and $b=d$. 
    \item {\bf Three leaves only}, where wlog $d=a$ and 
    $S_{ab} \leq S_{ac} = S_{bc}$. 
    \item {\bf Four leaves forming a cherry tree}, where 
    $S_{ab} \leq S_{cd} \leq S_{ad} = S_{bd} = S_{ac} = S_{bc}$. 
    \item {\bf Four leaves forming a comb tree}, where 
    $S_{ab} \leq S_{ac} = S_{bc} \leq S_{ad} = S_{bd} = S_{cd}$.
\end{itemize}
These are the only cases to be considered; if $a,b,c,d$ do not satisfy them, they can always be 
permuted so that at least one of the conditions above are satisfied.

Fig. \ref{fig:all_tree_shapes} illustrates those four cases.
\begin{figure}[!h]
    \centering
    \includegraphics[width=\linewidth]{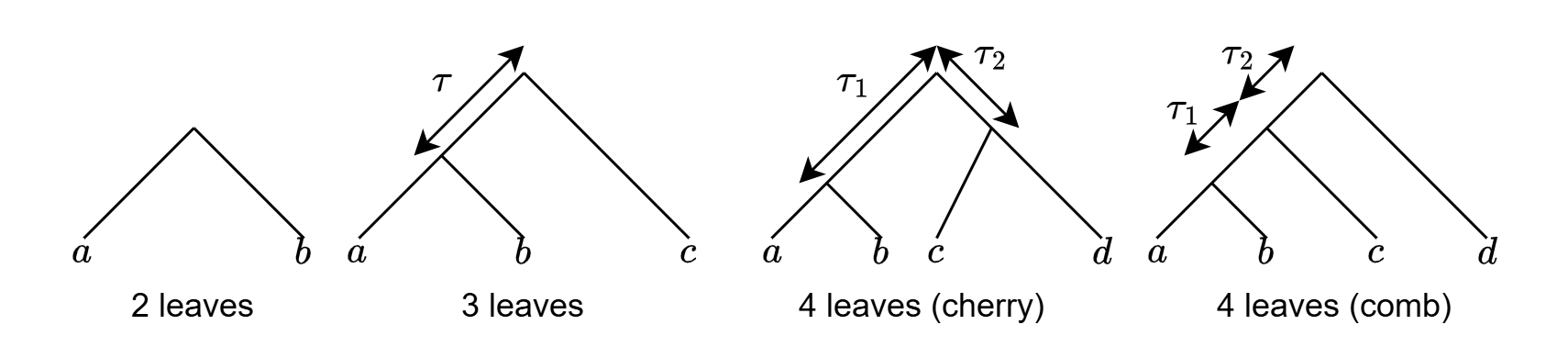}
    \caption{The four possible shapes of tree containing 
    $a,b,c,d \in L$ with $a\neq b, c\neq d$.}
    \label{fig:all_tree_shapes}
\end{figure}

{\bf Two leaf tree} \, For the case where $(a,b) = (c,d)$, the covariance expressions
becomes variances, which are positive. Under the MSC 
$g_{ab} = S_{ab} + 2 E_1$, where $E_1 \sim \Exp(1)$, and so 
\begin{align}
    \label{eq:var_gab}
    \Cov(g_{ab},g_{ab}) = \Var(g_{ab}) & = 4,\\
    \Cov(e^{t g_{ab} },e^{t g_{ab} }) = 
    \Var(e^{t g_{ab} }) &= \mathbb{E}\left( e^{2t g_{ab}} \right) -
    \left( \mathbb{E}( e^{t g_{ab}}) \right)^2
    \nonumber \\&= 
    e^{2tS_{ab}} \mathbb{E}\left( e^{2t E_1 } \right) 
    - e^{2tS_{ab}} \mathbb{E}\left( e^{t E_1 } \right)^2
    \nonumber \\&= 
    \frac{4t^2e^{2tS_{ab}}}{(1-4t)(1-2t)^2},
    \label{eq:var_etgab}
\end{align}
where for the last result we use that 
the moment generating function (MGF) of $\Exp(1)$ is 
$M(t) = \mathbb{E}\left( e^{t E_1 } \right) = 1/(1-t)$ 
and $\mathbb{E}\left( e^{2t E_1 } \right)  = 1/(1-2t)$.
\\ 
\par 
For the remaining three cases, we consider the variable
\[
E_{vw,xy} = \frac{ \left( g_{vw} + g_{xy} - S_{vw} - S_{xy} \right)}{2},
\]
where \( v, w, x, y \in L \), and its moment-generating function (MGF)
\[
M_{vw,xy}(t) = \mathbb{E}\left( \exp\left( t E_{vw,xy} \right) \right).
\]
The moment‐generating function allows access to the covariances of interest, as the two auxiliary results below (Equations~\eqref{eq:second_der},\eqref{eq:cov_mgf}) show. We now show these results. \\
\par 
For each pair \((i,j)\), define
\[
Y_{ij} := \frac{g_{ij}-S_{ij}}{2} \sim \mathrm{Exp}(1),
\]
with
\(\mathbb{E}[Y_{ij}]=1\), \(\mathrm{Var}(Y_{ij})=1\) and hence \(\mathbb{E}[Y_{ij}^2]=2\).
Writing
\[
E_{vw,xy} = Y_{vw}+Y_{xy},
\]
we have
\[
\frac{\partial^2M_{vw,xy}}{\partial t^2}(0)
= \mathbb{E}\bigl[(Y_{vw}+Y_{xy})^2\bigr]
= \mathbb{E}[Y_{vw}^2] + \mathbb{E}[Y_{xy}^2] + 2\,\mathbb{E}[Y_{vw}Y_{xy}].
\]
But,
\[
\mathbb{E}[Y_{vw}Y_{xy}]
= \Cov(Y_{vw},Y_{xy}) + \mathbb{E}[Y_{vw}]\;\mathbb{E}[Y_{xy}]
= \Cov(Y_{vw},Y_{xy}) + 1.
\]
Therefore
\[
\frac{\partial^2M_{vw,xy}}{\partial t^2}(0)
= 2 + 2 + 2\bigl(\Cov(Y_{vw},Y_{xy}) + 1\bigr)
= 6 + 2\,\Cov(Y_{vw},Y_{xy}).
\]
Noting that \(\Cov(Y_{vw},Y_{xy})=\tfrac14\,\Cov(g_{vw},g_{xy})\) recovers
\begin{equation}
\label{eq:second_der}
\frac{\partial^2M_{vw,xy}}{\partial t^2}(0)
=6+\tfrac12\,\Cov\bigl(g_{vw},g_{xy}\bigr).
\end{equation}
Similarly, one shows that
\begin{equation}
\label{eq:cov_mgf}
\Cov\bigl(e^{t g_{vw}},e^{t g_{xy}}\bigr)
= \Bigl(M_{vw,xy}(2t)-\tfrac1{(1-2t)^2}\Bigr)\,e^{t(S_{vw}+S_{xy})}.
\end{equation}
Indeed, since \(g_{ij}=S_{ij}+2Y_{ij}\) with \(Y_{ij}\sim\Exp(1)\),
\[
\mathbb{E}[e^{t g_{vw}}e^{t g_{xy}}]
= e^{t(S_{vw}+S_{xy})}\,\mathbb{E}[e^{2t(Y_{vw}+Y_{xy})}]
= e^{t(S_{vw}+S_{xy})}\,M_{vw,xy}(2t),
\]
and 
\[
\mathbb{E}[e^{t g_{ij}}]
= e^{tS_{ij}}\,(1-2t)^{-1},
\]
where $1/(1-2t)^2$ is the MGF of $2\Exp(1)$, 
so 
\(\mathbb{E}[e^{t g_{vw}}]\,
 \mathbb{E}[e^{t g_{xy}}]
= e^{t(S_{vw}+S_{xy})}(1-2t)^{-2}\).
Subtracting gives \eqref{eq:cov_mgf}. \\
\par
{\bf Three leaf tree}
Consider the case of a three-leaf tree where $S_{ab} \leq S_{ac} = S_{bc}$, 
where $\tau = \frac{S_{ac} - S_{ab}}{2}$ as shown in the second leftmost tree in 
Fig. \ref{fig:all_tree_shapes}. 
Define the joint random variable $(T,{\bf p})$, where 
$T$ is the first time that a pair coalesces and ${\bf p}$ is that pair.
In other words, 
\begin{align*}
    T &= \min_{x,y \in \binom{\{a,b,c\}}{2}} \{g_{xy}\}, \\
    {\bf p} &= \argmin_{(x,y) \in \binom{\{a,b,c\}}{2}} \{g_{xy}\}.
\end{align*}
Since the first speciation event happens at time $S_{ab}/2$, 
we require that $T \geq S_{ab}/2$. If $T \in (S_{ab}/2, S_{ac}/2)$, then ${\bf p}=(a,b)$.
The joint distribution function of $(T,{\bf p})$ is 
\begin{equation*}
    f_{T,{\bf p}}(T,{\bf p}) = \begin{cases}
        \exp\left(-\left(T-\frac{S_{ab}}{2}\right)\right), & 
        {\rm if} \,\, T \in \left( \frac{S_{ab}}{2}, \frac{S_{ac}}{2} \right), 
        {\bf p} = (a,b) \\
        \exp(-\tau) \exp\left( - 3 \left( T-\frac{S_{ac}}{2} \right) \right), & 
        {\rm if} \,\, T \geq \frac{S_{ac}}{2}, \, \forall {\bf p} \in 
        {\{a,b,c\} \choose 2}, \\ 
        0, & {\rm otherwise},
    \end{cases}
\end{equation*}
where for the remainder of this proof $\binom{\mathcal{A}}{2}$ is the set containing pairs of distinct 
elements of the set $\mathcal{A}$. 

Note that $T-\frac{S_{ac}}{2}|T>\frac{S_{ac}}{2} \sim \Exp(3)$, but since there are 
three choices of pairs, the pdf of $\Exp(3)$ is divided by 3 for each pair in the joint 
pdf above. Moreover, the conditioning 
probability $\mathbb{P}(T>\frac{S_{ac}}{2}) = \exp(-\tau)$ is included.
\par 
Conditional on $(T,{\bf p})$,
\begin{equation*}
    E_{ab,ac} |(T,{\bf p}) \overset{d}{=} \begin{cases}
        \left( T- \frac{S_{ab}}{2} \right) + \Exp(1), & 
        {\rm if} \,\, T \in \left( \frac{S_{ab}}{2}, \frac{S_{ac}}{2} \right), 
        {\bf p} = (a,b) \\
        2\left(T-\frac{S_{ac}}{2} \right) + \tau + 2 \Exp(1), & 
        {\rm if} \,\, T \geq \frac{S_{ac}}{2}, 
        {\bf p} = (b,c) \\
        2\left(T-\frac{S_{ac}}{2} \right) + \tau + \Exp(1), & 
        {\rm if} \,\, T \geq \frac{S_{ac}}{2}, 
        {\bf p} = (a,b) \,{\rm or}\, (a,c) \\
    \end{cases}
\end{equation*}
and so 
\begin{equation*}
    \mathbb{E}\left(e^{t E_{ab,ac}}|(T,{\bf p})\right) = \begin{cases}
        \frac{\exp\left(t \left(T- \frac{S_{ab}}{2}\right)\right)}{1-t}, & 
        {\rm if} \,\, T \in \left( \frac{S_{ab}}{2}, \frac{S_{ac}}{2} \right), 
        {\bf p} = (a,b) \\
        \frac{\exp{\left(2t \left( T -  \frac{S_{ac}}{2} \right) + t \tau \right)}}{1-2t}, & 
        {\rm if} \,\, T \geq \frac{S_{ac}}{2}, 
        {\bf p} = (b,c) \\
        \frac{\exp{\left(2t \left( T -  \frac{S_{ac}}{2} \right) + t \tau \right)}}{1-t}, & 
        {\rm if} \,\, T \geq \frac{S_{ac}}{2}, 
        {\bf p} = (a,b) \,{\rm or}\, (a,c) \\
    \end{cases}
\end{equation*}
Finally, combining those two results gives us the MGF, 
\begin{align*}
    M_{ab,ac}(t) &= \mathbb{E}_T \left( 
        \mathbb{E}( \exp(t E_{ab,ac} ) | T 
    \right)\\ &=
    \int_{\frac{S_{ab}}{2}}^{\infty} \sum_{{\bf p} \in \{(a,b),(b,c),(a,c)\}}
    f_{T,{\bf p}}(T,{\bf p}) \mathbb{E}\left(e^{t E_{ab,ac}|(T,{\bf p})} \right) \, dT
    \\
    & = \int_{S_{ab}/2}^{S_{ac}/2} e^{ -\left(T- \frac{S_{ab}}{2}\right)} 
    \frac{e^{t \left(T- \frac{S_{ab}}{2}\right)}}{1-t}
    \, dT  \\ & \,\,\,\,\,\, + 
    \int_{S_{ac}/2}^{+\infty} e^{-\tau}
    e^{\left( - 3 \left( T-\frac{S_{ac}}{2} \right) \right)} 
    \frac{e^{\left(2t \left( T -  \frac{S_{ac}}{2} \right) + t \tau \right)}}{1-2t}
    \, dT  \,\,\,\,\,\,\,\,\,\,\,\, ({\rm for } \,\, {\bf p} = (b,c))
    \\ & \,\,\,\,\,\, + 
    \int_{S_{ac}/2}^{+\infty} 2e^{-\tau}
    e^{\left( - 3 \left( T-\frac{S_{ac}}{2} \right) \right)} 
    \frac{e^{\left(2t \left( T -  \frac{S_{ac}}{2} \right) + t \tau \right)}}{1-t}
    \, dT  \,\,\,\,\,\,\,\,\, ({\rm for } \,\, {\bf p} \neq (b,c)) \\
    & = 
    \frac{1}{(1-t)^2} + e^{(t-1)\tau} \frac{t^2}{(1-t)^2(1-2t)(3-2t) }.
\end{align*}

Taking the second derivative of this expression with respect to $t$ at $t=0$ yields 
\begin{equation*}
    M_{ab,cd}''(0) = 6 + \frac{2}{3} e^{-\tau},
\end{equation*}
and so,
using equations 
\eqref{eq:second_der}, and \eqref{eq:cov_mgf}, it follows that 
\begin{align}
    \Cov(g_{ab},g_{ac}) &= \Cov(g_{ab},g_{bc})  = \frac{4}{3} e^{-\tau} \nonumber \\
    \Cov(e^{tg_{ab}},e^{tg_{ac}}) &= \Cov(e^{tg_{ab}},e^{tg_{bc}})
    =  e^{-\tau + 2tS_{ac} } \frac{4t^2}{(1-2t)^2 (1-4t) (3-4t)} \nonumber 
\end{align}
which are all positive. Note that the reason these results hold for the pairs $(a,b)$ and $(b,c)$ 
is that $M_{ab,bc}(t) = M_{ab,ac}(t)$, since $a$ and $b$ are interchangeable.
To get the correlations, we need the expression of the variances from Equations~\eqref{eq:var_gab}, \eqref{eq:var_etgab}, which yield 
\begin{align}
    \Cor(g_{ab},g_{ac}) &= \Cor(g_{ab},g_{bc})  = \frac{1}{3} e^{-\tau} \\
    \Cor(e^{tg_{ab}},e^{tg_{ac}}) &= \Cor(e^{tg_{ab}},e^{tg_{bc}})
    =   \frac{e^{(2t-1)\tau}}{3-4t},
\end{align}
\par 
We also need to examine $M_{ac,bc}(t)$ which is distinct from the other two cases.
Here, 
\begin{equation*}
    E_{ac,bc} |(T,{\bf p}) \overset{d}{=} \begin{cases}
        2 \Exp(1), & 
        {\rm if} \,\, T \in \left( \frac{S_{ab}}{2}, \frac{S_{ac}}{2} \right), 
        {\bf p} = (a,b) \\
        2\left(T-\frac{S_{ac}}{2} \right) + 2 \Exp(1), & 
        {\rm if} \,\, T \geq \frac{S_{ac}}{2}, 
        {\bf p} = (a,b) \\
        2\left(T-\frac{S_{ac}}{2} \right) + \Exp(1), & 
        {\rm if} \,\, T \geq \frac{S_{ac}}{2}, 
        {\bf p} \neq (a,b). \\
    \end{cases}
\end{equation*}
Following the same approach as before, 
\begin{align*}
    M_{ac,bc}(t) &=  \int_{S_{ab}/2}^{S_{ac}/2} e^{ -\left(T- \frac{S_{ab}}{2}\right)} 
    \frac{1}{1-2t}
    \, dT  \\ &+ 
    \int_{S_{ac}/2}^{+\infty} e^{-\tau}
    e^{\left( - 3 \left( T-\frac{S_{ac}}{2} \right) \right)} 
    \frac{e^{2t \left( T -  \frac{S_{ac}}{2} \right)}}{1-2t}
    \, dT  \,\,\,\,\, ({\rm for } \,\, {\bf p} = (a,b))
    \\ &+ 
    2\int_{S_{ac}/2}^{+\infty} e^{-\tau}
    e^{\left( - 3 \left( T-\frac{S_{ac}}{2} \right) \right)} 
    \frac{e^{2t \left( T -  \frac{S_{ac}}{2} \right)}}{1-t}
    \, dT  \,\,\,\,\, ({\rm for } \,\, {\bf p} \neq (a,b)) \\
    & = 
    \frac{1}{1-2t} - e^{-\tau} \frac{2t^2}{(1-t)(1-2t)(3-2t) }.
\end{align*}
Hence, it follows that 
\begin{align}
    \Cov(g_{ac},g_{bc}) &= 4 - \frac{8}{3} e^{-\tau} \geq 0 \\
    \Cov(e^{tg_{ac}},e^{tg_{bc}}) &= \frac{4t^2}{(1-2t) (1-4t)}
    e^{2tS_{ac}} \left( \frac{1}{1-2t}- e^{-\tau}
        \frac{2}{3-4t} 
    \right)\\
    & \geq
    e^{2tS_{ac}}  \frac{4t^2}{(1-2t)^2 (1-4t)(3-4t)} \geq 0
    \nonumber
    \\
    \Cor(g_{ac},g_{bc}) &= 1- \frac{2}{3} e^{-\tau} \\
    \Cor(e^{tg_{ab}},e^{tg_{ac}}) &
    =   1-  \frac{2(1-2t)}{3-4t}e^{-\tau}.
\end{align}
{\bf Cherry tree}
We continue by considering the cherry tree as shown in Figure~\ref{fig:all_tree_shapes}, with 
$S_{ab} \leq S_{cd} \leq S_{ad} = S_{bd} = S_{ac} = S_{bc} = \Delta$, 
where $\Delta$ is the diameter of the subtree containing $a,b,c,d$. We use this for notational convenience 
for the remainder of the proof and should not be confused with the diameter of the species tree contatining 
all leaves. We also defined,
$\tau_1 = \frac{\Delta-S_{ab}}{2}$ and 
$\tau_2 = \frac{\Delta-S_{cd}}{2} \leq \tau_1$. 
Once again, define the joint random variable $(T,{\bf p})$ as before. For the cherry tree case,
it has density function
\begin{equation*}
    f_{T,{\bf p}}(T,{\bf p}) = \begin{cases} 
        \exp\left(-\left(T-\frac{S_{ab}}{2}\right)\right), & 
        {\rm if} \,\, T \in \left( \frac{S_{ab}}{2}, \frac{S_{ac}}{2} \right), 
        {\bf p} = (a,b) \\
        \exp(- (\tau_1-\tau_2))
        \exp\left(-2\left(T-\frac{S_{ac}}{2}\right)\right), & 
        {\rm if} \,\, T \in \left( \frac{S_{ac}}{2}, \frac{\Delta}{2} \right), 
        {\bf p} = (a,b) \,{\rm or }\, (c,d) \\
        \exp(-\tau_1 - \tau_2)) 
        \exp\left( - 6 \left( T-\frac{S_{ac}}{2} \right) \right), & 
        {\rm if} \,\, T \geq \frac{\Delta}{2} \, \forall {\bf p} 
        \in {\{a,b,c,d\} \choose 2} \\
        0, & 
        {\rm otherwise.}
    \end{cases}
\end{equation*}
Conditional on $(T,{\bf p})$ and the species tree $S$,
\begin{equation*}
    E_{ab, cd} |(T,{\bf p}), S \overset{d}{=} \begin{cases}
        \left(T-\frac{S_{ab}}{2}\right) + \Exp(1), & 
        {\rm if} \,\, T \in \left( \frac{S_{ab}}{2}, \frac{S_{cd}}{2} \right), 
        {\bf p} = (a,b) \\
        2\left(T-\frac{S_{cd}}{2} \right) +
        \tau_1 - \tau_2+ 
        \Exp(1), & 
        {\rm if} \,\, T \in (\frac{S_{cd}}{2}, \frac{\Delta}{2}), 
        {\bf p} = (a,b), (c,d) \\
        2\left(T-\frac{\Delta}{2} \right) + \tau_1 + \tau_2 + \Exp(1), & 
        {\rm if} \,\, T \geq \frac{\Delta}{2}, 
        {\bf p} = (a,b), (c,d) \\
        2\left(T-\frac{\Delta}{2} \right) + \tau_1 + \tau_2 + 
        E_{xy,xz}|S=* , & 
        {\rm if} \,\, T \geq \frac{\Delta}{2}, 
        {\bf p} \neq (a,b), (c,d).
    \end{cases}
\end{equation*}
Note the last case where the first pair to coalesce is neither $(a,b)$ nor 
$(a,c)$. Let's assume without loss of generality that the first pair was $(a,c)$. 
Since species $a$ and $c$ have coalesced, we can treat them as one leaf $a=c$ in a 
gene tree of $3$ leaves $a=c, b,d$ that has not been resolved yet. This gene tree is 
generated by the Multispecies coalescent model assuming that the species tree is the 
star tree, denoted $S=*$, since all three species can coalesce after $T\geq \Delta/2$. 
Therefore, 
\[
    E_{ab, cd} |(T,{\bf p}) \overset{d}{=} 2(T-\Delta/2) + \tau_1 + \tau_2 + E_{ab,ad}|S=*.
\]
We have already computed the MGF of $E$ for three-leaf trees. When $S=*$, the internal 
branch length vanishes i.e. $\tau = 0$.
Hence, 
\[
    \mathbb{E}\left(e^{t (E_{xy,xz}|S=* )}\right) = 
    \frac{1}{(1-t)^2} + \frac{t^2}{(1-t)^2(1-2t)(3-2t) } 
    = \frac{3-5t}{(1-t)(1-2t)(3-2t)}, {\rm and}
\]

\begin{equation*}
    \mathbb{E}\left(e^{t E_{ab,cd}}|(T,{\bf p})\right) = \begin{cases}
        \frac{\exp\left(t \left(T- \frac{S_{ab}}{2}\right)\right)}{1-t},& 
        {\rm if} \,\, T \in \left( \frac{S_{ab}}{2}, \frac{S_{cd}}{2} \right), 
        {\bf p} = (a,b) \\
        \frac{\exp{\left(2t \left( T -  \frac{S_{cd}}{2} \right) + t (\tau_1 - \tau_2 )
        \right)}}{1-t}, & 
        {\rm if} \,\, T \in \left( \frac{S_{ac}}{2}, \frac{\Delta}{2} \right), 
        {\bf p} = (a,b) \,{\rm or }\, (c,d) \\
        \frac{\exp{\left(2t \left( T -  \frac{\Delta}{2} \right) + 
        t (\tau_1+\tau_2) \right)}}{1-t}, & 
        {\rm if} \,\, T \geq \frac{\Delta}{2}, 
        {\bf p} = (a,b) \,{\rm or }\, (c,d) \\ 
        \frac{\exp{\left(2t \left( T -  \frac{\Delta}{2} \right) + 
        t (\tau_1+\tau_2) \right)} (3-5t) }{(1-t)(1-2t)(3-2t)}, & 
        {\rm if} \,\, T \geq \frac{\Delta}{2}, 
        {\bf p} \neq (a,b) \,{\rm or }\, (c,d) .
    \end{cases}
\end{equation*}
Hence, 
\begin{align*}
    M_{ab,cd}(t) &=  \int_{S_{ab}/2}^{S_{ac}/2} e^{ -\left(T- \frac{S_{ab}}{2}\right)} 
    \frac{e^{ t\left(T-\frac{S_{ab}}{2} \right) }}{1-t}
    \, dT  \\ &+ 2 
    \int_{S_{ac}/2}^{\Delta/2} e^{-(\tau_1 - \tau_2 )} 
    e^{\left( - 2 \left( T-\frac{S_{ac}}{2} \right) \right)} 
    \frac{e^{2t \left( T -  \frac{S_{ac}}{2} \right)} 
    e^{ t(\tau_1 - \tau_2) }
    }{1-2t}
    \, dT \\ &+ 
    2\int_{\Delta/2}^{+\infty} e^{-\tau_1 - \tau_2}
    e^{\left( - 6 \left( T-\frac{\Delta}{2} \right) \right)} 
    \frac{e^{2t \left( T -  \frac{\Delta}{2} \right)}e^{ t(\tau_1 + \tau_2) }}{1-t}
    \, dT  \,\,\,\,\,\,\,\,\,\,\,\,\,\,\,\,\,\,\,\,\,\,\,\,\,\,
    ({\rm for } \,\, {\bf p} = (a,b), (c,d))
    \\ &+ 
    4\int_{\Delta/2}^{+\infty} e^{-\tau_1 - \tau_2}
    e^{\left( - 6 \left( T-\frac{\Delta}{2} \right) \right)} 
    \frac{e^{2t \left( T -  \frac{\Delta}{2} \right)}e^{ t(\tau_1 + \tau_2)}(3-5t)}
    {(1-t)(1-2t)(3-2t)}
    \, dT  \,\,\,\,\,\, ({\rm for } \,\, {\bf p} \neq (a,b), (c,d) ) \\
    & = 
    \frac{1}{(1-t)^2} + e^{-\tau_1 - \tau_2} \frac{2t^2}{(1-t)^2(1-2t)(3-2t)(3-t) }.
\end{align*}
From this we conclude that 
\begin{align}
    \Cov(g_{ab},g_{cd}) &= \frac{8}{9} e^{-(\tau_1+\tau_2)} \geq 0 \\
    \Cov(e^{tg_{ab}},e^{tg_{cd}}) &=
    e^{t(S_{ab}+S_{cd}) - (\tau_1 + \tau_2)} \frac{8t^2}{(1-2t)^2(1-4t)(3-4t)(3-2t) }
    \\
    \Cor(g_{ac},g_{bc}) &= \frac{2}{9} e^{-(\tau_1+\tau_2)} \\
    \Cor(e^{tg_{ab}},e^{tg_{ac}}) &
    =  e^{-(\tau_1+\tau_2)} \frac{2}{(3-4t)(3-2t)}.
\end{align}
A different approach is used for finding the MGF of $E_{ac,bd}$. Define the random 
variable $N \in \{0,1,2\}$ to be the number of pairs that coalesce before time $\Delta/2$. 
The only pairs that could coalesce in that time interval are ${\bf p} = (a,b)$ or $(c,d)$ 
and the probabilities of each one of these pairs coalescing is independent of the other.
Therefore,
\begin{align*}
    \mathbb{P}(N=0) &= e^{-\tau_1 - \tau_2}, \\
    \mathbb{P}(N=1) &= (1-e^{-\tau_1})  e^{-\tau_2} + (1-e^{-\tau_2})  e^{-\tau_1}, \\
    \mathbb{P}(N=2) &= (1-e^{-\tau_1})  (1-e^{-\tau_2}).
\end{align*}
If $N=0$, then $E_{ac,bd} = E_{xy,zw} | S=*$ i.e. the MGF of $E_{ac,bd}$ will be 
the same as $M_{ab,cd}(t)$ with $\tau_1 = \tau_2=0$ or
\begin{equation*}
    \mathbb{E}\left(e^{tE_{ac,bd}}|N=0\right) = 
    \frac{1}{(1-t)^2} +  \frac{2t^2}{(1-t)^2(1-2t)(3-2t)(3-t) }.
\end{equation*}
If $N=1$, either $(a,b)$ or $(c,d)$ have coalesced before time $\Delta/2$
but not both, and so 
at time $\Delta/2$ there will be three nodes that have not coalesced yet.
Consequently, $E_{ac,bd} = E_{xy,xw} | S=*$ i.e. the MGF of $E_{ac,bd}$ will be 
the same as $M_{ab,ac}(t)$ with $\tau = 0$ or
\begin{equation*}
    \mathbb{E}\left(e^{tE_{ac,bd}}|N=1\right) = 
    \frac{1}{(1-t)^2} + \frac{t^2}{(1-t)^2(1-2t)(3-2t)}.
\end{equation*}
For $N=2$, both $(a,b)$ and $(c,d)$ have coalesced, $E_{ac,bd} = 2\Exp(1)$, 
and thus
\begin{equation*}
    \mathbb{E}\left(e^{tE_{ac,bd}}|N=2\right) = 
    \frac{1}{1-2t} = \frac{1}{(1-t)^2} + \frac{t^2}{(1-t)^2 (1-2t)} .
\end{equation*}
Combining those three results, 
\begin{align*}
    M_{ac,bd}(t) &= \mathbb{E}_N(\mathbb{E}(e^{tE_{ac,bd}}|N) ) 
    = \sum_{n=0}^{2} \mathbb{E}(e^{tE_{ac,bd}}|N=n) \mathbb{P}(N=n)\\&=
    \frac{1}{(1-t)^2} + 
    \frac{2t^2 e^{-\tau_1-\tau_2}}{(1-t)^2(1-2t)(3-2t)(3-t) } \\&+ 
    \frac{t^2 \left((1-e^{-\tau_1})  e^{-\tau_2} + (1-e^{-\tau_2})  e^{-\tau_1}\right)}
    {(1-t)^2(1-2t)(3-2t)} + \frac{t^2 (1-e^{-\tau_1})  (1-e^{-\tau_2})}{(1-t)^2 (1-2t)}.
\end{align*}
The same argument can be made for the pair $(a,d)$ and $(b,c)$, since $a,b$ are 
interchangeable. We conclude that 
\begin{align}
    \Cor(g_{ad},g_{bc}) = \Cor(g_{ac},g_{bd}) &= 1 - \frac{2}{3} \left(e^{-\tau_1} 
    + e^{-\tau_2} \right) + \frac{5}{9} e^{-(\tau_1+\tau_2)} 
    \label{eq:cor_3_steac}\\
    \Cor(e^{tg_{ad}},e^{tg_{bc}}) = \Cor(e^{tg_{ac}},e^{tg_{bd}}) &
    =  1
    - \frac{2(1-2t)}{3-4t} (e^{-\tau_1} + e^{-\tau_2}) 
    \nonumber
    \\ & +
    e^{-(\tau_1+\tau_2)} \frac{(5-4t)(1-2t)}{(3-4t) (3-2t)}.
    \label{eq:cor_3_metal}
\end{align}
We now need to prove that the correlations in Equation~\eqref{eq:cor_3_steac} and \eqref{eq:cor_3_metal}
are non-negative and that the correlation of $\Cor(e^{tg_{ac}},e^{tg_{bd}})$ is an increasing 
function of $t \in (-\infty, 0]$. Note that since,
\begin{align*}
    \lim_{t \to - \infty } \Cor(e^{tg_{ad}},e^{tg_{bc}}) &= 
    (1-e^{-\tau_1}) ( 1 - e^{-\tau_2} ) \geq 0, \, {\rm and} \\ 
    \lim_{t \to 0 } \Cor(e^{tg_{ad}},e^{tg_{bc}}) &= \Cor(g_{ad},g_{bc}),
\end{align*}
proving the monotonicity of $\Cor(e^{tg_{ac}},e^{tg_{bd}})$ on $(-\infty, 0]$ is sufficient for 
the positivity of \eqref{eq:cor_3_steac} and \eqref{eq:cor_3_metal}.\\
To do that, we compute the derivative,
\begin{align*}
    \frac{d(\Cor(e^{tg_{ac}},e^{tg_{bd}}))}{dt} &= \frac{4}{(3-4t)^2} \left(
    (e^{-\tau_1} + e^{-\tau_2}) - \frac{8t^2 - 16t + 9}{(3-2t)^2} e^{-\tau_1-\tau_2}
    \right) \\& \geq
    \frac{4}{(3-4t)^2} \left(
        (e^{-\tau_1} + e^{-\tau_2}) - 2 e^{-\tau_1-\tau_2}
    \right) \geq 0 .
\end{align*}
The first inequality comes from the observation that 
\begin{equation*}
    \frac{8t^2 - 16t + 9}{(3-2t)^2} \geq 2 ~~\Leftrightarrow ~~ t < \frac{9}{8} .
\end{equation*}
The second inequality stems from the fact that 
the function \[g:[0,1]^2 \to \mathbb{R},\, g(x,y)=x+y-2xy\] has non-negative range, which can 
be proved by rewriting $g$ as 
\[
    g(x,y) = 1/2 - \left(x-\frac{1}{2} \right) \left(2y-1 \right),
\]
which is clearly minimized at $x,y = (1,1)$, and so $\min_{(x,y) \in [0,1]^2} g(x,y) = g(1,1) = 0$.
This concludes the proof for the case of the cherry tree.
\par
{\bf Comb tree} Finally, we conclude this proof with the last tree shape of Fig.
\ref{fig:all_tree_shapes}, the comb tree, where $S_{ab} 
\leq S_{ac} = S_{bc} \leq S_{ad} = S_{bd} = S_{cd}= \Delta$
and $\tau_1 = \frac{S_{ac} - S_{ab}}{2}, \, \tau_2 = \frac{\Delta - S_{ac}}{2} $. \\
The joint density function of $(T,{\bf p})$ is
\begin{equation*}
    f_{T,{\bf p}}(T,{\bf p}) = \begin{cases} 
        \exp\left(-\left(T-\frac{S_{ab}}{2}\right)\right), & 
        {\rm if} \,\, T \in \left( \frac{S_{ab}}{2}, \frac{S_{ac}}{2} \right), 
        {\bf p} = (a,b) \\
        \exp(- \tau_1)
        \exp\left(-3\left(T-\frac{S_{ac}}{2}\right)\right), & 
        {\rm if} \,\, T \in \left( \frac{S_{ac}}{2}, \frac{\Delta}{2} \right), 
        \forall {\bf p} 
        \in {\{a,b,c\} \choose 2} \\ 
        \exp(-\tau_1 - 3\tau_2)) 
        \exp\left( - 6 \left( T-\frac{\Delta}{2} \right) \right), & 
        {\rm if} \,\, T \geq \frac{\Delta}{2} \, \forall {\bf p} 
        \in {\{a,b,c,d\} \choose 2} \\
        0, & 
        {\rm otherwise.}
    \end{cases}
\end{equation*} Furthermore,
\begin{equation*}
    E_{ab, cd} |(T,{\bf p}), S \overset{d}{=} \begin{cases}
        \left(T-\frac{S_{ab}}{2}\right) + \Exp(1), & 
        {\rm if} \,\, T \in \left( \frac{S_{ab}}{2}, \frac{S_{ac}}{2} \right), 
        {\bf p} = (a,b) \\
        \left(T-\frac{S_{ab}}{2}\right) + 
        \Exp(1), & 
        {\rm if} \,\, T \in (\frac{S_{ac}}{2}, \frac{\Delta}{2}), 
        {\bf p} = (a,b) \\
        \left(T-\frac{S_{ab}}{2}\right) + E_{xy,xz} | S=\mathcal{A}_{\Delta/2 - T}, & 
        {\rm if} \,\, T \in (\frac{S_{ac}}{2}, \frac{\Delta}{2}), 
        {\bf p} = (a,c), (b,c)\\
        2\left(T-\frac{\Delta}{2} \right) + \tau_1 + \tau_2 + 
        \Exp(1) , & 
        {\rm if} \,\, T \geq \frac{\Delta}{2}, 
        {\bf p} = (a,b), (c,d) \\ 
        2\left(T-\frac{\Delta}{2} \right) + \tau_1 + \tau_2 + 
        E_{xy,xz}|S=* , & 
        {\rm if} \,\, T \geq \frac{\Delta}{2}, 
        {\bf p} \neq (a,b), (c,d).
    \end{cases}
\end{equation*}
where the three-leaf tree $\mathcal{A}_{\tau}$, for $\tau>0$, 
is drawn in Fig. \ref{fig:three_leaf_tree}.

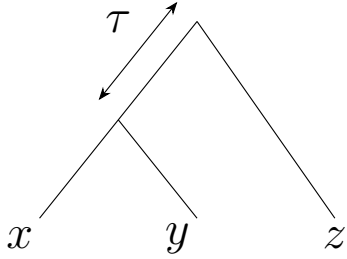
\begin{figure}[!h]
\centering
\resizebox{0.4\textwidth}{!}{
\begin{circuitikz}
\tikzstyle{every node}=[font=\LARGE]
\draw [short] (3.75,12.75) -- (5.75,15.25);
\draw [short] (5.75,15.25) -- (7.5,12.75);
\draw [short] (4.75,14) -- (5.75,12.75);
\node [font=\LARGE] at (3.5,12.5) {$x$};
\node [font=\LARGE] at (5.5,12.5) {$y$};
\node [font=\LARGE] at (7.5,12.5) {$z$};
\draw [<->, >=Stealth] (4.5,14.25) -- (5.5,15.5);
\node [font=\LARGE] at (4.75,15.25) {$\tau$};
\end{circuitikz}
}
\caption{Three-leaf tree $\mathcal{A}_{\tau}$ with internal branch length $\tau$. 
Note that $A_{0} = *$.}
\label{fig:three_leaf_tree}
\end{figure}

The MGF of $E_{xy,xz} | S=\mathcal{A}_{\tau}$ and $E_{xy,xz} | S=*$ has already been 
computed before (three-leaf tree) as
\begin{align*}
    M_{xy,xz|\mathcal{A}_{\tau}}(t) &=
    \frac{1}{(1-t)^2} + e^{(t-1)\tau} \frac{t^2}{(1-t)^2(1-2t)(3-2t) }, \\
    M_{xy,xz|*}(t) &= M_{xy,xz|\mathcal{A}_{0}}(t) = 
    \frac{3-5t}{(1-t)(1-2t)(3-2t)}.
\end{align*}

Hence, 
\begin{align*}
    M_{ab,cd}(t) &=  \int_{S_{ab}/2}^{S_{ac}/2} e^{ -\left(T- \frac{S_{ab}}{2}\right)} 
    \frac{e^{ t\left(T-\frac{S_{ab}}{2} \right) }}{1-t}
    \, dT  \\ &+ 
    \int_{S_{ac}/2}^{\Delta/2} e^{-\tau_1} 
    e^{\left( - 3 \left( T-\frac{S_{ac}}{2} \right) \right)} 
    \frac{e^{t \left( T -  \frac{S_{ac}}{2} \right)} 
    e^{ t\tau_1 }
    }{1-t}
    \, dT \,\,\,\,\, ({\rm for } \,\, {\bf p} = (a,b))\\
    &+ 
    \int_{S_{ac}/2}^{\Delta/2} e^{-\tau_1} 
    e^{\left( - 3 \left( T-\frac{S_{ac}}{2} \right) \right)} 
    e^{2t \left( T -  \frac{S_{ac}}{2} \right)} 
    e^{ t\tau_1 }
    M_{xy,xz| \mathcal{A}_{\Delta/2-T}}(t) 
    \, dT \,\,\,\,\, ({\rm for } \,\, {\bf p} = (a,c),(b,c))\\
    &+ 
    2\int_{\Delta/2}^{+\infty} e^{-\tau_1 - 3\tau_2}
    e^{\left( - 6 \left( T-\frac{\Delta}{2} \right) \right)} 
    \frac{e^{2t \left( T -  \frac{\Delta}{2} \right)}e^{ t(\tau_1 + \tau_2) }}{1-t}
    \, dT  \,\,\,\,\, ({\rm for } \,\, {\bf p} = (a,b), (c,d))
    \\ &+ 
    4\int_{\Delta/2}^{+\infty} e^{-\tau_1 - 3\tau_2}
    e^{\left( - 6 \left( T-\frac{\Delta}{2} \right) \right)} 
    e^{2t \left( T -  \frac{\Delta}{2} \right)}e^{ t(\tau_1 + \tau_2)}
    M_{xy,xz|*}(t)
    \, dT  \,\,\,\,\, ({\rm for } \,\, {\bf p} \neq (a,b), (c,d) ) \\
    & = 
    \frac{1}{(1-t)^2}+ e^{(t-1) (\tau_1+ \tau_2)}
    \frac{t^2 \left(
        (3-t) - (1-t) e^{-2 \tau_2}
    \right)}{(1-t)^2(1-2t)(3-2t)(3-t)},
\end{align*}
and so
\begin{align}
    \Cor(g_{ab},g_{cd}) &=\frac{1}{3} e^{-\tau_1-\tau_2} -
    \frac{1}{9} e^{-\tau_1 - 3\tau_2} \in \left(0, \frac{1}{3} \right), \\
    \Cor(e^{tg_{ab}},e^{tg_{cd}})&
    = \frac{e^{(2t-1) (\tau_1+\tau_2)}}{3-4t} \left(
        1 - \frac{1-2t}{3-2t} e^{-2\tau_2}
    \right).
\end{align}
Note that $\Cor(e^{tg_{ab}},e^{tg_{cd}})$ is an increasing function of $t$ on 
$(-\infty,0)$, since the first term is increasing and $- \frac{1-2t}{3-2t}$ is also 
increasing.

\begin{equation*}
    E_{ac, bd} |(T,{\bf p}), S \overset{d}{=} \begin{cases}
        E_{xy,xz} | S = \mathcal{A}_{\tau_2} , & 
        {\rm if} \,\, T \in \left( \frac{S_{ab}}{2}, \frac{S_{ac}}{2} \right), 
        {\bf p} = (a,b) \\
        E_{ab,cd}|(T,{\bf p}), S, &
        {\rm if} \,\, T \geq \frac{S_{ac}}{2}.
    \end{cases}
\end{equation*}
The last case follows from the fact that $a,b,c$ are interchangeable conditional on $(a,b)$
not coalescing before time $S_{ac}/2$. Moreover, note that 
$E_{ab,cd}|T > S_{ac}/2, S$ has MGF equal to the above $M_{ab,cd}(t)$, but with 
$\tau_1=0$ i.e. 
\begin{equation*}
    M_{ab,cd|\tau_1=0}(t) = \frac{1}{(1-t)^2}+ e^{(t-1) \tau_2}
    \frac{t^2 \left(
        (3-t) - (1-t) e^{-2 \tau_2}
    \right)}{(1-t)^2(1-2t)(3-2t)(3-t)} .
\end{equation*}
Hence, it follows that 
\begin{align*}
    M_{ac,bd}(t) &= \left(1-e^{-\tau_1}\right) 
    M_{xy,xz|\mathcal{A}_{\tau_2}}(t) + e^{-\tau_2} M_{ab,cd|\tau_1=0}(t) \\
    &= \frac{1}{(1-t)^2} + \frac{t^2}{(1-t)^2(1-2t)(3-2t)} e^{(t-1)\tau_2}
    \left(
        1 - e^{-2\tau_2-\tau_1} \frac{1-t}{3-t}
    \right) .
\end{align*}
Also note that $M_{ac,bd}(t) = M_{ad,bc}(t)$. It follows that the correlations are 
as follows
\begin{align}
    \Cor(g_{ac},g_{bd}) = \Cor(g_{ad},g_{bc}) &=\frac{1}{3} e^{-\tau_2} -
    \frac{1}{9} e^{-\tau_1 - 3\tau_2} \in \left(0, \frac{1}{3} \right), \\
    \Cor(e^{tg_{ac}},e^{tg_{bd}}) = \Cor(e^{tg_{ad}},e^{tg_{bc}})&
    = \frac{e^{(2t-1) \tau_2}}{3-4t} \left(
        1 - \frac{1-2t}{3-2t} e^{-2\tau_2- \tau_1}
    \right) .
\end{align}
Similarly, the correlation is an increasing function of $t$.
This concludes the proof for the case of the comb tree.
\end{proof}

\begin{proof}[{\bf Proof of Proposition \ref{thm:evol_cov}}]
    The covariance of interest can be rewritten as follows, using the properties of indicator random variables,
    \begin{align}
        {\rm Cov}( \mathbbm{1}( \chi_a \neq \chi_b ), \mathbbm{1}( \chi_c \neq \chi_d ) )  
        &= 
        {\rm Cov}( \mathbbm{1}( \chi_a = \chi_b ), \mathbbm{1}( \chi_c = \chi_d ) ) 
        \nonumber
        \\ &=
        \mathbb{P}(\chi_a = \chi_b, \chi_c = \chi_d) - 
        \mathbb{P}(\chi_a = \chi_b)
        \mathbb{P}(\chi_c = \chi_d).
        \label{eq:def_cov_bases}
    \end{align}
    Under the Jukes-Cantor evolutionary model, when there is a substitution, it is equally likely that the current base is substituted by any one of the four 
    bases A, C, T, G. Since the probability of no substitution  
    across a branch of time length $t$ is $\exp(-\mu t)$, the probability  that 
    a randomly selected base stays the same across the same branch is 
    $$\exp(-\mu t) + \frac{1}{4}(1-\exp(-\mu t)) = \frac{1}{4} +
    \frac{3}{4} \exp(-\mu t).$$
    It follows that,
    \begin{equation}
        \label{eq:prob_common_base}
        \mathbb{P}(\chi_i = \chi_j) = 
         \frac{3}{4} e^{-\mu g_{ij}} + \frac{1}{4}, \,
         \forall i,j \in L.
    \end{equation}
    The second term of Equation~\eqref{eq:def_cov_bases} can then be computed 
    through a product of terms computed from Equation~\eqref{eq:prob_common_base}.
    We now focus on the first term of Equation~\eqref{eq:def_cov_bases}.
    Let 
    $r$ and $q$ be the internal nodes connecting $a$ and $c$ and $b$ and $d$ respectively, such that the length of 
    the path from $q$ to $r$ is equal to $\delta_{ab,cd}$.
    Similarly to Equation~\eqref{eq:prob_common_base}, 
    \begin{equation}
    \label{eq:prob_common_base_delta}
        \mathbb{P}(\chi_r = \chi_q) = \frac{3}{4} e^{-\mu \delta_{ab,cd} } 
        + \frac{1}{4}. 
    \end{equation}
    
    We replace the states $A,C,G,T $ with numbers $0,1,2,3$ respectively. 
    Define the random variable $N_{v,w} \equiv |\chi_{v} - \chi_{w}| \mod 4$
    and note that 
    \begin{align*}
        \{\chi_r = \chi_q\} &= \{ N_{rq} \equiv 0 \mod 4 \}, \\
        \{\chi_a = \chi_b\} &= \{ N_{ab} \equiv 0 \mod 4 \}, \\
        \{\chi_c = \chi_d\} &= \{ N_{cd} \equiv 0 \mod 4 \},\quad {\rm where} \\ 
        N_{ab} & \equiv N_{ar} + N_{rq} + N_{qb} \mod 4, \\
        N_{cd} & \equiv N_{cr} + N_{rq} + N_{qd} \mod 4
    \end{align*}
    Moreover, $N_{ar}, N_{rq}, N_{qb}, N_{cr}, N_{qd}$ are all independent of each 
    other since their corresponding paths are disjoint. Hence, $N_{ar} + N_{qb}$ 
    is independent of $N_{cr} + N_{qd}$.
    For the remainder of the 
    proof $\mod 4$ will be dropped when $\equiv$ is used.\\
    \par 
    Conditioning on $N_{rq}$,
    \begin{align}
        \mathbb{P}(\chi_a = \chi_b, \chi_c = \chi_d)
        &=   \sum_{i=0}^3
        \mathbb{P}(\chi_a = \chi_b, \chi_c = \chi_d | N_{rq} \equiv i) \mathbb{P}(N_{rq} \equiv i) \nonumber \\
        &=  \sum_{i=0}^3
        \mathbb{P}(N_{ar} + N_{qb} \equiv -i, N_{cr} + N_{qd} \equiv -i) \mathbb{P}(N_{rq} \equiv i ) \nonumber \\
        &= \sum_{i=0}^3
        \mathbb{P}(N_{ar} + N_{qb} \equiv -i) \mathbb{P}(N_{cr} + N_{qd} \equiv -i) \mathbb{P}(N_{rq} \equiv i )  \nonumber \\ &=
        \mathbb{P}(N_{ar} + N_{qb} \equiv 0) \mathbb{P}(N_{cr} + N_{qd} \equiv 0) \mathbb{P}(\chi_r = \chi_q ) \nonumber \\  & \quad + 
        \frac{1}{9}
        \mathbb{P}(N_{ar} + N_{qb} \not \equiv 0) \mathbb{P}(N_{cr} + N_{qd} \not \equiv 0) \mathbb{P}(\chi_r \neq \chi_q ), \label{eq:first_term}
    \end{align}
    where the third equality holds because of the independence of $N_{ar} + N_{qb}$ 
    and $N_{cr} + N_{qd}$, and the fourth equality holds because all the terms in 
    the summand are equal and each probabilities
    involving sums of $N$ being $\equiv i$
    are a third of the corresponding probabilities with $\not \equiv 0$.\\
    \par 
    Now the only thing left is to find the distribution of 
    $N_{ar}+ N_{qb}$ and $N_{cr} + N_{qd}$. Without loss of generality, we will 
    focus on the former. The probability that there are no substitutions across
    the paths from $a$ to $r$ and $q$ to $b$ is $\exp(-\mu (g_{ab} - \delta_{ab,cd}))$,
    since the length of the sum of the paths is $g_{ab} - \delta_{ab,cd}$.
    If there has been no substitution along those paths, then $N_{ar}+ N_{qb} \equiv 0$. Otherwise, if there has been at least one substitution, 
    $N_{ar}+ N_{qb} \equiv i$, where $i=0,1,2,3$
    are all equiprobable. Therefore, 
    \begin{align}
        \mathbb{P}(N_{ar}+ N_{qb} \equiv 0) &= 
        \frac{1}{4} + \frac{3}{4} e^{-\mu (g_{ab} - \delta_{ab,cd})} 
        \label{eq:nar+nqb}\\
        \mathbb{P}(N_{cr} + N_{qd} \equiv 0) &= 
        \frac{1}{4} + \frac{3}{4} e^{-\mu (g_{cd} - \delta_{ab,cd})} 
        \label{eq:ncr+nqd}
    \end{align}
    Substituting Equations~\eqref{eq:nar+nqb}--\eqref{eq:ncr+nqd} and 
    Equation~\eqref{eq:prob_common_base_delta} into Equation~\eqref{eq:first_term}
    gives 
    \begin{align*}
        \mathbb{P}(\chi_a = \chi_b, \chi_c = \chi_d) &= 
        \frac{1}{16} \left(
            1 + 3e^{-\mu (g_{ab}+g_{cd} - 2\delta_{ab,cd})} + 3e^{-\mu g_{ab}}
            \right. \\ & \left. \quad
            +
            3 e^{-\mu g_{cd}} + 6e^{-\mu (g_{ab} + g_{cd} -\delta_{ab,cd}) }
        \right)
    \end{align*}
    Substituting the above to Equation~\eqref{eq:def_cov_bases}
    yields the desired result. 

\end{proof}

\begin{proof}[{\bf Proof of Proposition \ref{prop:total_covariance_decomposition}}]
    Using the law of total covariance, 
    the total covariance between pairs of METAL distance estimates 
    for a single loci is 
    \begin{align}
        &\Cov\left(\hat{p}_{ab}^{(1)}, \hat{p}_{cd}^{(1)} \right) 
        \nonumber \\
        &=
        \mathbb{E}_{G^{(1)}| \mathcal{S}}
        \left( \Cov\left(\hat{p}_{ab}^{(1)}, \hat{p}_{cd}^{(1)} | G^{(1)}\right) \right) +
        \Cov_{G^{(1)}| \mathcal{S}}
        \left(\mathbb{E}\left(\hat{p}_{ab}^{(1)}|G^{(1)} \right), 
        \mathbb{E}\left(\hat{p}_{cd}^{(1)}|G^{(1)} \right)
        \right)  
        \nonumber 
        \\ &= 
        \frac{3}{16 K } 
         \mathbb{E}_{G^{(1)}| \mathcal{S}} \left( 
            e^{-\mu \left(g_{ab}^{(i)} + g_{cd}^{(i)} \right)}
        \left(e^{2\mu \delta_{ab,cd}^{(i)}} + 2 e^{\mu \delta_{ab,cd}^{(i)}} - 3
        \right)
            \right)+ 
        \frac{9}{16} \Cov\left(e^{-\mu g_{ab}}, e^{-\mu g_{cd}}\right).
        \label{eq:cov_decomposition}
    \end{align}
    To derive the last line we use Proposition~\ref{thm:evol_cov} for the first term and 
    the Jukes-Cantor formula
    \begin{equation*}
        \mathbb{E}\left(\hat{p}_{xy}^{(1)}|G^{(1)} \right) = p^{(1)}_{xy | G^{(1)}} = 
        \frac{3}{4} \left( 1- e^{-\mu g_{xy}^{(1)} }\right) .
    \end{equation*}
    for the second term.
    Using the definitions in Equations~\eqref{eq:sigma_total_def},~\eqref{eq:sigma_coal}, \eqref{eq:sigma_evol},
    Equation~\eqref{eq:sigma_total_decomposition} 
    is equivalent to \eqref{eq:cov_decomposition}. \\
    \par
    It is easy to see that $\Sigma_{\rm coal}$ is positive semi-definite, as 
    a covariance matrix of the random vector 
    $\exp{\left(-\mu g_{\binom{L}{2}} \right)}$. For $\Sigma_{\rm sub}$,
    observe that  
    \[
        \Sigma_{\hat{p}_{\binom{L}{2}}^{(1)} | G^{(1)} } \coloneqq
        \Cov\left(\hat{p}_{ab}^{(1)}, \hat{p}_{cd}^{(1)} | G^{(1)}\right),
    \]
    is positive semidefinite i.e. for any vector ${\bf v} \in \mathbb{R}^{\binom{n}{2}}$, ${\bf v} \Sigma_{\hat{p}_{\binom{L}{2}}^{(1)} | G^{(1)} } {\bf v}^T \geq 0$
    and so 
    \[
        {\bf v} \Sigma_{\rm sub} {\bf v}^T = \mathbb{E}\left(
            {\bf v} \Sigma_{\hat{p}_{\binom{L}{2}}^{(1)} | G^{(1)} } v^T
        \right) \geq 0
    \]
    This proves that $\Sigma_{\rm sub}$ is a positive semi-definite matrix.
\end{proof}

\begin{proof}[{\bf Proof of Proposition \ref{prop:covariance_coal_sub_total}}]
Using the formulas from the Proof of Proposition~\ref{thm:positive_cov} for the coalescent 
covariance matrix and the formulas from Appendix~\ref{sec:cov_calculations} for the substitution covariance matrix.
Observe that for the star species tree, 
the three leaf covariances agree with each other for both 
the coalescent and the substitution matrix. So do the results for the four leaf covariances.
In other words, $(\Sigma_{\rm coal})_{ab,cd}$ and $(\Sigma_{\rm sub})_{ab,cd}$ only depend 
on the total number of leaves $|\{a,b,c,d\}| \in \{2,3,4\}$.
Also, note that $|\{a,b\} \cap \{c,d\}| = 4 - |\{a,b,c,d\}| \in \{0,1,2\}$ and so 
$\sigma^{(0)},\sigma^{(1)},\sigma^{(2)}$ now refer to 4-leaf, 3-leaf, and 2-leaf covariances
respectively.\\
\par
First, we derive the formulas of $\sigma_{\rm mode}^{(i)}$ for 
${\rm mode} \in \{ {\rm coal, sub} \}, i \in \{0,1,2\}$.
\begin{align}
    \sigma_{\rm coal}^{(0)}(\mu) & = 
    \frac{8\mu^2}{(1+4\mu)(1+2\mu)^2(3+4\mu)(3+2\mu)} e^{-2\mu\Delta}
    \label{eq:sigma_coal_0} \\
    \sigma_{\rm coal}^{(1)}(\mu) & = 
    \frac{4\mu^2}{(1+4\mu)(1+2\mu)^2(3+4\mu)} e^{-2\mu\Delta}
    = \frac{3+2\mu}{2}\sigma_{\rm coal}^{(0)}(\mu)
    \label{eq:sigma_coal_1} \\
    \sigma_{\rm coal}^{(2)}(\mu) & = 
    \frac{4\mu^2}{(1+4\mu)(1+2\mu)^2} e^{-2\mu\Delta} 
    = (3+4\mu)\sigma_{\rm coal}^{(1)}(\mu)
    \label{eq:sigma_coal_2} \\[1ex]
    \sigma_{\rm sub}^{(0)}(\mu) & = 
    \frac{e^{-2\mu\Delta}}{(1+2\mu)(3+2\mu)} - 
    \sigma_{\rm coal}^{(0)}(\mu) - \frac{e^{-2\mu\Delta}}{(1+2\mu)^2} \nonumber \\
    & = \frac{4\mu (8\mu^2+ 17\mu + 6)}{3(\mu+1)(2\mu+1)(2\mu+3)(4\mu+1)(4\mu+3)} e^{-2\mu\Delta}
    \label{eq:sigma_sub_0} \\
    \sigma_{\rm sub}^{(1)}(\mu) & = 
    \frac{2}{3} \frac{e^{-3\mu\Delta/2}}{(1+2\mu)(1+\mu)} + 
    \frac{1}{3} \frac{e^{-\mu\Delta}}{1+2\mu}
    - \sigma_{\rm coal}^{(1)}(\mu) - \frac{e^{-2\mu\Delta}}{(1+2\mu)^2}  \nonumber \\
    & \geq e^{-2\mu\Delta} \frac{2\mu(8\mu^2 + 17\mu + 6)}{3(\mu+1)(4\mu+1)(4\mu+3)(2\mu+1)} 
    = \frac{3+2\mu}{2}\sigma_{\rm sub}^{(0)}(\mu)
    \label{eq:sigma_sub_1} \\
    \sigma_{\rm sub}^{(2)}(\mu) & = 
    \frac{2}{3} \frac{e^{-\mu \Delta}}{1+2\mu} +
    \frac{1}{3} - \sigma_{\rm coal}^{(2)}(\mu) - \frac{e^{-2\mu\Delta}}{(1+2\mu)^2} 
    \geq \frac{8}{3} \frac{\mu(\mu+1)}{(4\mu+1)(2\mu+1)} e^{-2\mu\Delta}
    \label{eq:sigma_sub_2}
\end{align}

The inequalities are derived by using $e^{-i \mu \Delta} \leq e^{-2 \mu \Delta}$ for all $i \leq 2$.
From here, the asymptotic results as $\mu \to 0$ and $\mu \to \infty$ are easy to derive. 
It is also clear that $\sigma_{\rm coal}^{(i)}$ is decreasing with $i \in \{0,1,2\}$, and that
$\sigma_{\rm sub}^{(1)}(\mu) \geq \sigma_{\rm sub}^{(0)}(\mu)$. We now prove that 
$\sigma_{\rm sub}^{(2)}(\mu) \geq \sigma_{\rm sub}^{(1)}(\mu)$, 
\begin{align*}
    \sigma_{\rm sub}^{(2)}(\mu) - \sigma_{\rm sub}^{(1)}(\mu) & 
    \geq \frac{1}{3} e^{-3\mu\Delta/2} \frac{2\mu(\mu+2)}{(1+2\mu) (1+\mu)}
    - 2(1 + 2\mu)\sigma_{\rm coal}^{(1)}(\mu) \\
    & \geq e^{-2\mu \Delta} \frac{2 \left(16\mu^3 + 36\mu^2 + 23\mu + 6\right)}{3 ( \mu + 1 )( 2\mu + 1 )( 4\mu + 1 )( 4\mu + 3 )} >0.
\end{align*} 
This concludes the proof.
\end{proof}

\begin{proof}[{\bf Proof of Proposition \ref{prop:steac_infinite_variance}}]
    We need to compute the variance of $\hat{g}_{ab}$ using the law of total variance by conditioning on the gene tree branch length $g_{ab}$
    \begin{align*}
        \Var(\hat{g}_{ab}) &= \mathbb{E}(\Var(\Hat{g}_{ab})| g_{ab} ) + \Var\left(\mathbb{E}\left(
        \hat{g}_{ab} | g_{ab} \right) \right) \\ 
        &= \mathbb{E}\left(\frac{e^{2\mu \delta_{ab,cd}^{(i)}} + 2e^{\mu \delta_{ab,cd}^{(i)}} - 3}{3\mu^2 K} \right) +
        \Var\left(g_{ab} \right),
    \end{align*}
    where the variance and the expectation of $\hat{g}_{ab}  |g_{ab}$ is given by the 
    Delta method in Equation~\eqref{eq:steac_sigma}, which is assumed to hold in this Proposition.
    Since $g_{ab} = S_{ab} + 2 \Exp(1)$ and the variance of $\Exp(1)$ is $1$, the second 
    term is equal to 4. 
    For the first term, note that $\delta_{ab,ab} = g_{ab}$ and since the MGF of the exponential 
    distribution $\Exp(1)$ is $M(z) = 1/(1-z)$, it follows that
    \[ \mathbb{E}\bigl(e^{2\mu(S_{ab} + 2\Exp(1))} \bigr)= \frac{e^{\mu S_ab} }{1-4\mu}\]
    is finite if and only if $\mu < 1/4$. Clearly, the same applies to $\Var(\hat{g}_{ab})$.\\
    Finally, we will prove that if 
    $ \mathbb{P}\bigl(\hat{p}_{ab}^{(i)} \ge \tfrac{3}{4}\bigr)  = 0$,
    then the variance $\Var(\hat{g}_{ab})$ is finite. 
    Note that $\hat{p}_{ab}^{(j)} \in \{j/K: i \in \{0,1,\dots,K\} \}$, and so
    \[
        \left\{\hat{p}_{ab}^{(i)} < \tfrac{3}{4} \right\} = \left\{\hat{p}_{ab}^{(i)} \leq \frac{ 
        \lceil \tfrac{3}{4} K \rceil - 1}{K}   \coloneqq \frac{3}{4} - \epsilon \right\},
    \]
    where $\lceil x \rceil$ is the ceiling of $x$ and $\epsilon > 0$. Hence, 
    \begin{equation*}
        1 =  \mathbb{P}\bigl(\hat{p}_{ab}^{(i)} < \tfrac{3}{4}\bigr) = 
        \mathbb{P}\bigl(\hat{p}_{ab}^{(i)} < \tfrac{3}{4} - \epsilon \bigr) = 
        \mathbb{P}\bigl(\hat{g}_{ab}^{(i)} < M_{\epsilon} \bigr),
    \end{equation*}
    where $M_\epsilon > 0$ is a large constant.
    Consequently, if $\mathbb{P}\bigl(\hat{p}_{ab}^{(i)} \ge \tfrac{3}{4}\bigr) =0 $,
    then$\hat{g}_{ab} \in [0,M_{\epsilon}]$ and so $\Var(\hat{g}_{ab})$ is finite.
    The contrapositive, implies that when $\mu \geq 1/4$ and the variance is infinite, 
    the probability $p =\mathbb{P}\bigl(\hat{p}_{ab}^{(i)} \ge \tfrac{3}{4}\bigr) $ is positive.
    The remainder of the theorem follows from here. 
\end{proof}

\begin{proof}[{\bf Proof of Proposition \ref{thm:spectrum}}]
    First we find the spectrum of the following matrix; 
    \begin{equation*}
        A \in \mathbb{R}^{ {n \choose 2} \times  {n \choose 2} }, 
        A_{ab,cd} = |\{a,b\} \cap \{c,d\}| \in \{0,1,2\}
    \end{equation*}
    We need to prove the following results for $A$
    \begin{enumerate}[i.]
        \item $r(A) \le n$.  
        \item Take two distinct leaves $x,y \in [n]$. 
        Consider the vector $v^{(xy)} \in \mathbb{R}^{n \choose 2}$
        with $v^{(xy)}_{xy} = 2n - 4$, $v^{(xy)}_{xz} = v^{(xy)}_{yz} = n-4$ 
        and $v^{(xy)}_{zw} = -4$ for all distinct $z,w \neq x,y$.
        Then $v^{(xy)}$ is an eigenvector of $A$ with eigenvalue $n-2$. 
        In particular, the eigenvectors $v^{(12)}, v^{(13)}, \dots, v^{(1n)}$ are linearly
        independent and so $\dim(E_{n-2}(A) ) \geq n-1$,
        where $E_{\rho}(Y)$ is the eigenspace of matrix $Y$ 
        associated with eigenvalue $\rho$.
    \end{enumerate}
    The first result shows that the dimension of 
    the eigenspace $E_{0}(A)$ is at least 
    ${n \choose 2} - n = n(n-3)/2$ by the rank-nullity theorem. 
    The second result shows that 
    \(
    \dim\left(E_{n-2}(A)\right) \geq n-1
    \). The dominant eigenvalue is $2n - 2$ (with corresponding eigenvector 
    ${\bf 1}$) and the corresponding eigenspace 
    has dimension at least $1$. 
    Since the sum of the dimensions of 
    eigenspaces is ${n \choose 2}$, it follows that 
    \begin{align*}
        \dim\left(E_{2n-2}(A)\right) & = 1, \\ 
        \dim\left(E_{n-2}(A)\right) & = n-1, \\ 
        \dim\left(E_{0}(A)\right) & = \frac{n(n-3)}{2}.
    \end{align*}
    \newline
\textit{Proof of i.} 
First, observe that the sum of rows is a multiple of the all-ones vector 
${\bf 1} = (1,\dots, 1)$, since
\begin{equation*}
    \sum_{a<b} A_{ab, cd} = \sum_{a<b} |\{a,b\} \cap \{c,d\}| = 2(n-1), \, \, 
    \forall (c,d) \in \binom{[n]}{2}.
\end{equation*}
Hence, it suffices to prove that the space spanned
by the first $n-1$ rows $A_{1l, .}$ 
for $l \in \{2,\dots,n\}$ along with the 
row $(1,\dots,1)$ 
(which is the sum of all the rows of $A$) 
contains all the other rows
$A_{xy, .}$ for $x, y \neq 1$. 
To do that, our goal is to write $A_{xy, .}$ as a linear combination 
of the first $n-1$ rows the vector $(1,\dots,1)$ for any $x,y >1$, and proving 
element-wise that it holds.
Consider the coefficients 
$\lambda^{}_j = \lambda + 
\left| \{x,y\} \cap \{j\} \right|, j \in \{2,\dots,n\}  
$, 
corresponding to row $A_{1j, .}$, 
where $\lambda$ is a free parameter 
to be determined. 
Also let $ -2 \lambda $ be the coefficient 
corresponding to the newly created row ${\bf 1} = (1,1,\dots,1)$. 
Then, for all components $zw$ with $z \neq w$ and $z,w \in [n]$,
\begin{align*} 
    \sum_{j=2}^n \lambda^{}_j A_{1j, zw} + (-2\lambda) {\bf 1}_{zw} &= 
    - 2 \lambda +
    \sum_{j=2}^n \left( \lambda + \left|\{j\} \cap \{x,y\} \right| \right)
     \left|\{1,j \} \cap  \{z,w\} \right| 
    \\&=
    - 2 \lambda +
    \sum_{j=2}^n \left( \lambda + \left|\{j\} \cap \{x,y\} \right| \right)
     \left( \left|\{1\} \cap  \{z,w\} \right| + \left|\{j\} \cap  \{z,w\} \right| \right)
    \\&=
    - 2 \lambda - 2 \lambda \left|\{1\} \cap  \{z,w\} \right| 
    \\& \,\,\,\, + 
    \sum_{j=1}^n \left( \lambda + \left|\{j\} \cap \{x,y\} \right| \right)
     \left( \left|\{1\} \cap  \{z,w\} \right| + \left|\{j\} \cap  \{z,w\} \right| \right)
     \\&=
     - 2 \lambda - 2 \lambda \left|\{1\} \cap  \{z,w\} \right|
     (\lambda n + 2)  \left|\{1\} \cap  \{z,w\} \right| + 2 \lambda 
     \\& \,\,\,\, + 
    \sum_{j=1}^n  
    \left| \{ j \} \cap \{z,w\}  \right|
    \cdot \left|\{j\} \cap \{x,y\} \right|
    \\ &=
    \left|\{1\} \cap  \{z,w\} \right| \left( \lambda (n-2) + 2 \right) + 
    \left|\{z,w\} \cap \{x,y\}\right| = A_{xy,zw}
\end{align*}
For the last equality to hold, note that first term vanishes
if we choose $\lambda = \frac{-2}{n-2}$ 
and the summation term is equal to $\left|\{z,w\} \cap \{x,y\}\right|$. 
Hence, we have proved the following result
\begin{equation*}
    \sum_{j=2}^n \lambda^{(x,y)}_j A_{1j, .} + \lambda_1^{(x,y)} {\bf 1} = A_{xy, .}, 
    \forall x,y \in \{ 2, 3 , \dots, n \}, x \neq y
\end{equation*}
which means that for all distinct $x,y \in [n]$ we have 
$A_{xy,.} \in  \spn\{ A_{12,.}, \dots, A_{1n,.}, {\bf 1} \}$. 
Hence, the rank of 
$A$ cannot exceed $n$. \\ 
\newline
\textit{Proof of ii.} 
We now need to verify that $A v^{(xy)} = (n-2) v^{(xy)}$. First, observe that 
\begin{equation*}
    (A^2)_{ab,cd} = \sum_{x<y} A_{ab,xy} A_{xy,cd} = \begin{cases}
        2n,  & {\rm if} |\{a,b\} \cap \{c,d\} | = 2,  \\
        n+2, & {\rm if}  |\{a,b\} \cap \{c,d\} | = 1, \\
        4, & {\rm if}  |\{a,b\} \cap \{c,d\} | = 0 \\
    \end{cases} 
\end{equation*}
and so it follows that 
\begin{equation*}
    A^2 = (n-2) A + 4  \cdot {\bf 1} {\bf 1}^T ,
\end{equation*}
where ${\bf 1}  = (1, \dots,1)$ is a column vector i.e. ${\bf 1} {\bf 1}^T \in \mathbb{R}^{ \binom{n}{2} \times \binom{n}{2}}$ is a matrix of all-ones.
Given that $A {\bf 1} = (2n-2) {\bf 1}$, we conclude that 
\begin{equation*}
    A( nA - 4  \cdot {\bf 1} {\bf 1}^T ) = (n-2) ( nA - 4 \cdot {\bf 1} {\bf 1}^T),
\end{equation*}
which implies that the columns of $V \coloneqq nA - 4 \cdot  {\bf 1} {\bf 1}^T$ 
are all eigenvectors of A with eigenvalue $n-2$. 
Note that $v^{(xy)} = V_{., xy}$ i.e. the vectors constructed in statement (ii) are
precisely the columns of $V$ and hence elements in the eigenspace $E_{n-2}(A)$.\\
\par 
To prove that $v^{(12)}, \dots, v^{(1n)}$ are linearly
independent, consider the square matrix $B \in  \mathbb{R}^{(n-1) \times (n-1)}$ with 
$B_{i,j} = v_{i}^{(1j)}$. It suffices to prove that $B$ is non-singular. 
Note that $B_{i,i} = 2n-4 > n-4 = B_{i,j} > 0$
for all $j \neq i$ and $i \in \{2,\dots,n\}$ and denote them as 
$B_{ii}= x, B_{ij} = y$ for all $i \in [N], j \neq i$ where $N=n-1$ and 
$x > y > 0$. Observing that $B$ is a circulant matrix, its eigenvalues are 
\begin{equation*}
    \rho_k = \sum_{j=0}^{N-1} c_j \omega^{kj}, k \in \{0,\dots, N-1\}
\end{equation*}
where $c_0 = x, \, c_1= \dots = c_{N-1} = y$, and $\omega$ is the $N^{\rm th}$-root
of unity. Hence, $\rho_0 = x+ (n-1) y > 0$ and $\rho_k = x + y  \sum_{j=1}^{N-1} \omega^{kj} = x-y > 0$ for $k \in \{1,\dots,N-1\}$. All eigenvalues are positive and
so $B$ is non-singular.
\newline 
Having found the spectrum of $A$, the task is to now find the spectrum of $C$. 
Observe that 
\[ 
    C = (\beta - \gamma) A + ( \gamma + \alpha - 2\beta) I + \gamma {\bf 1},
\]
where ${\bf 1}$ is a matrix whose entries are all 1.\\ 
Every eigenvalue and eigenvector that is not ${\bf 1}$ is transformed 
\begin{align*}
    \Tilde{ \lambda } &= (\beta - \gamma) \lambda + (\gamma + \alpha - 2\beta) \\ 
    \Tilde{v} & = v - {\bf 1}
    \frac{\gamma v^T {\bf 1}}{ (\beta-\gamma) (2n-2-\lambda) + \gamma n}
\end{align*}
In particular, if $Av = \lambda v$, then $C \Tilde{v} = \Tilde{\lambda} \Tilde{v}$. 
This argument doesn't work for $v = {\bf 1}$ since $\Tilde{v} = 0$. But it's straightforward to compute the eigenvalue directly 
\begin{equation*}
    C{\bf 1} = \left( (\beta - \gamma) (2n-2) +  \gamma + \alpha - 2 \beta + 
        {n \choose 2} \gamma
    \right)
    {\bf 1}
\end{equation*}
The eigenvalue $n-2$ of $A$ becomes $(\beta-\gamma) (n-2) + \gamma +\alpha - 2\beta$ 
and the corresponding eigenvectors do not change because $v^T {\bf 1} = 0$ for 
all the vectors in $E_{n-2}(A)$. The remaining eigenvalue is $\gamma+\alpha - 2\beta$.
\end{proof}

\begin{proof}[{\bf Proof of Corollary \ref{cor}}]
    For the trace and Frobenius norm, the results of Propositions~\ref{prop:covariance_coal_sub_total} suffice to prove them. For the determinant, the spectral norm, and the results about variance ratios, 
    the spectrum of those matrices is given by Proposition~\ref{thm:spectrum}.\\ 
    Regarding the variance ratio, note that Proposition~\ref{thm:spectrum} provides limits of 
    variance ratios as $n \to \infty$. The variance ratio of PC1 is $\gamma/\alpha$ and 
    the variance variance of the first $n$ prinipal components is $\frac{2\beta-\gamma}{\alpha}$. 
    For the covariance matrix $\Sigma_{\rm mode}$ where mode is either coal or sub, 
    $\alpha = \sigma_{\rm mode}^{(2)}, 
    \beta = \sigma_{\rm mode}^{(1)}, 
    \gamma = \sigma_{\rm mode}^{(0)}$. \\
    We start with $\Sigma_{\rm coal}$, whose components $\alpha,\beta,\gamma$ are given by
    Equations~\eqref{eq:sigma_coal_0}--\eqref{eq:sigma_coal_2}. As $\mu \to 0$,
    \begin{equation*}
        \frac{\sigma_{\rm coal}^{(i) }}{\mu^2} \to \begin{cases}
            4, & {\rm if } \quad i=2\\ 
            \frac{4}{3}, & {\rm if } \quad i=1\\ 
            \frac{8}{9}, & {\rm if } \quad i=0
        \end{cases}.
    \end{equation*}
    From that it follows that 
    \begin{align*}
        \lim_{\mu \to 0} \frac{\gamma_{\rm coal}}{\alpha_{\rm coal}} &= \frac{2}{9}\\
        \lim_{\mu \to 0} \frac{2\beta_{\rm coal}-\gamma_{\rm coal}}{\alpha_{\rm coal}} &= \frac{4}{9}
    \end{align*}
    Similarly, for $\Sigma_{\rm sub}$, as $\mu \to 0$, 
    \begin{equation*}
        \frac{\sigma_{\rm sub}^{(i) }}{\mu} \to \begin{cases}
            \frac{8}{3} + \frac{4}{3}\Delta , & {\rm if } \quad i=2\\ 
            \frac{4}{3}+\frac{2}{3}\Delta, & {\rm if } \quad i=1\\ 
            \frac{8}{9}, & {\rm if } \quad i=0
        \end{cases}.
    \end{equation*}
    From that it follows that 
    \begin{align*}
        \lim_{\mu \to 0} \frac{\gamma_{\rm sub}}{\alpha_{\rm sub}} &= \frac{2}{6+3\Delta}\\
        \lim_{\mu \to 0} \frac{2\beta_{\rm sub}-\gamma_{\rm sub}}{\alpha_{\rm sub}} &= \frac{4+3\Delta}{6+3\Delta}
    \end{align*}
    This concludes the proof of the Corollary.
\end{proof}

\begin{proof}[{\bf Proof of Proposition \ref{prop:coal_infinite_tree}}]
    Two statements need to be proved; first, that $v^{(i)}$ is an eigenvector
    of $C^{\infty}_{\rm coal}$
    for all 
    interior nodes $i$, and second, that the remaining eigenvalues are zero.
    Observe that $C^{\infty}_{\rm coal}$ only takes values $\{0,1\}$. 
    The value of $\left(C^{\infty}_{\rm coal}\right)_{ab,cd}$ 
    is $1$ if and only if the topology of the leaves $\{a,b,c,d\}$ is either 
    $( (a,c), (b,d) )$ or $( (a,d), (b,c) )$. Hence, 
    \begin{align*}
        \left( C^{\infty}_{\rm coal} v^{(i)} \right)_{ab} &= \sum_{cd} 
        \mathbb{I}\left(
            ( (a,d), (b,c) ) \text{ or } ( (a,c), (b,d) ) 
        \right)  \cdot 
        \mathbb{I}\left(
             (c,d) \in (L_i \times R_i)
        \right) \\ &= 
        \sum_{c \in L_i, d \in R_i} 
        \mathbb{I}\left(
            ( (a,d), (b,c) ) \text{ or } ( (a,c), (b,d) ) 
        \right) \\ &= 
        \sum_{c \in L_i, d \in R_i} 
        \mathbb{I}\left(
            (a,d), (b,c) 
        \right) + 
        \sum_{c \in L_i, d \in R_i} 
        \mathbb{I}\left(
            (a,c), (b,d) 
        \right) \\ &= 
        \mathbb{I}\left( 
            (a,b) \in R_i \times L_i
        \right) |L_i| |R_i|+ 
        \mathbb{I}\left( 
            (a,b) \in L_i \times R_i 
        \right) |L_i| |R_i|
        \\ & =
        |L_i| |R_i| \mathbb{I}\left( 
            (a,b) \in (L_i \times R_i) \cup (R_i \times L_i)
        \right) 
        \\ & = 
        |L_i| |R_i| v^{(i)}_{ab}
    \end{align*}
    To prove the second statement, it suffices to prove that
    $\sum_{i} |L_i| |R_i| = \text{tr}(C^{\infty}_{\rm coal})$ where the sum is over all 
    internal nodes; 
    that's because the trace is the sum of eigenvalues and $C^{\infty}_{\rm coal}$ is a
    symmetric semi-positive definite matrix. 
    The diagonal of $C^{\infty}_{\rm coal}$ only includes ones, and so 
    $\text{tr}(C^{\infty}_{\rm coal}) = {m \choose 2}$. Let $r$ be the root of the tree
    with $|L_r| + |R_r| = m$. Denote $\xi_{\mathcal{A}}$ be the sum of eigenvalues 
    (corresponding to interior nodes) for the tree formed from the leaf set $\mathcal{A} \subset [m]$.
    It follows that 
    \begin{equation*}
        \xi_{[m]} = |L_r| ( m - |L_r| ) + \xi_{L_i} + \xi_{R_i}
    \end{equation*}
    Using strong induction, we write $\xi_{L_r} = {|L_r| \choose 2}$ and $\xi_{R_r} = {m - |L_r| \choose 2}$. 
    Substituting those to the formula above yields $\xi_{[m]} = {m \choose 2}$. 
    That concludes the inductive step.
\end{proof}

\section{Covariance calculations} \label{sec:cov_calculations}


Recall that using Equation~\eqref{eq:cov_decomposition} we decompose the covariance into an MSC and substitution component. Formulae for the former can be found in the Proof of Proposition~\ref{thm:positive_cov}. 
In this section of the Appendix, the task is to compute the latter component, namely the substitution 
covariance matrix for all four subtree topology cases shown in Figure~\ref{fig:all_tree_shapes}. 
This substitution component can be further decomposed as follows
\begin{align*}
    \frac{16}{9} K \left( \Sigma_{\rm sub} \right)_{ab,cd} &= 
         \mathbb{E}_{G^{(1)}| \mathcal{S}} \left( 
            e^{-\mu \left(g_{ab} + g_{cd} \right)}
        \left(\frac{1}{3}e^{2\mu \delta_{ab,cd}} + \frac{2}{3} e^{\mu \delta_{ab,cd}} - 1
        \right)
            \right)
        \\& = 
        \frac{1}{3} e^{(2)}_{ab,cd} + \frac{2}{3} e^{(1)}_{ab,cd}
        \\& \,\,\,\,\,\,\,\,\, -\left(
        \Cov\left(e^{-\mu g_{ab}}, e^{-\mu g_{cd}}\right)
        + \frac{e^{- \mu (S_{ab} + S_{cd}))}}{(1+2\mu)^2 }
        \right), \quad {\rm where}\\
        e^{(1)}_{ab,cd} & \coloneqq \mathbb{E}_{G^{(1)}| \mathcal{S}} \left( e^{-\mu(g_{ab} + g_{cd} - \delta_{ab,cd})} \right) \\
        e^{(2)}_{ab,cd} & \coloneqq \mathbb{E}_{G^{(1)}| \mathcal{S}} \left( e^{-\mu(g_{ab} + g_{cd} - 2\delta_{ab,cd})} \right)
\end{align*}
where $\Cov\left(e^{-\mu g_{ab}}, e^{-\mu g_{cd}}\right)$ is the covariance under MSC derived in Proposition~\ref{thm:positive_cov}. Therefore, the task now is to compute $e^{(1)}_{ab,cd}$ and
$e^{(2)}_{ab,cd}$
for all different tree topologies. 
\subsection{Computing $e^{(1)}$ terms}
First, we will state all the results for $e^{(1)}_{ab,cd}$, and then provide their proofs.

\begin{itemize}
    \item {\bf Two-leaves.} If $\{a,b\} = \{c,d\}$, then $\delta_{ab,ab} = g_{ab}$ and so 
    \begin{equation*}
        e^{(1)}_{ab,ab} = \frac{e^{-\mu S_{ab} }}{ 1+2\mu}.
    \end{equation*}
    \item {\bf Three-leaf tree} with tree topology $((a,b),c)$ has 
    \begin{equation*}
        e^{(1)}_{ab,ac} = e^{(1)}_{ac,bc} = e^{(1)}_{ab,bc}= \frac{e^{-\mu\left( \Delta + \frac{S_{ab}}{2} \right) } }{(1+\mu) (1+2\mu)}
    \end{equation*}
    \item {\bf Four leaves (Cherry tree)} with tree topology $((a,b),(c,d))$ has
    \begin{align*}
        e^{(1)}_{ab,cd} &= \frac{e^{-\mu( S_{ab} + S_{cd}) }}{(1+2\mu)^2} 
        + e^{-(2\mu+1) \Delta + \frac{S_{ab} + S_{cd}}{2}} 
        \frac{2\mu}{(\mu + 1)(2\mu + 1)^2(2\mu + 3)} \\
        e^{(1)}_{ac,bd} = e^{(1)}_{ad,bc} &= 
        \frac{e^{-\mu\left(\Delta + \frac{S_{ab} + S_{cd}}{2} \right) }}{(1+2\mu)(1+\mu)^2} 
        + e^{-(2\mu+1) \Delta + \frac{S_{ab} + S_{cd}}{2}} 
        \frac{\mu(\mu + 2)}{(\mu + 1)^2(2\mu + 1)(2\mu + 3)}
    \end{align*}
    \item {\bf Four leaves (Comb tree)} with tree topology $(((a,b), c),d)$ has
    \begin{align*}
        e^{(1)}_{ab,cd} &= \frac{e^{-\mu (S_{ab} + \Delta)}}{(1+2\mu)^2} 
        + e^{-\mu(\Delta + S_{ac}) - \frac{S_{ac} - S_{ab}}{2}} 
        \frac{2\mu}{(\mu + 1)(2\mu + 1)^2(2\mu + 3)} \\
        e^{(1)}_{ac,bd} = e^{(1)}_{ad,bc} &= 
        \frac{e^{-\mu\left(\Delta + \frac{S_{ab} + S_{ac}}{2} \right) }}{(1+2\mu)(1+\mu)^2} 
        + e^{-\mu(\Delta + S_{ac}) - \frac{S_{ac} - S_{ab}}{2}} 
        \frac{\mu(\mu + 2)}{(\mu + 1)^2(2\mu + 1)(2\mu + 3)}
    \end{align*}
\end{itemize}
Here and for the remainder of the Appendix, $\Delta$ is the diameter of the subtree with leaves $a,b,c,d$ and should 
not be confused with the diameter of the total tree, which may include more leaves and thus be larger. A more 
accurate notation would be $\Delta_{abcd}$, but we will write $\Delta$ for convenience.

For the {\bf 2-leaf} case, note that $\delta_{ab,ab} = g_{ab}$ and so 
\[
    e^{(1)}_{ab,ab} = \mathbb{E}(e^{-\mu g_{ab}}) = e^{\mu S_{ab}} \mathbb{E}\left(e^{-2\mu E_1} \right) = 
    \frac{ e^{\mu S_{ab}}}{1+2\mu},
\]
where $E_1 \coloneqq (g_{ab} - S_{ab})/2 \sim \Exp(1)$ under the MSC, whose moment generating function is 
$1/(1-t)$ which gives the $1/(1+2\mu)$ term in the expression above. \\
\par
For the {\bf 3-leaf} case, note that 
\[
    \Lambda \coloneqq g_{ab} + g_{ac} - \delta_{ab,ac} =  g_{ab} + g_{bc} - \delta_{ab,bc} = g_{ac} + g_{bc} - \delta_{ac,bc},
\]
which is equal to the sum of the branches of the gene tree containing the three leaves $\{a,b,c\}$, regardless of the topology of 
the gene tree. Let $T$ be the first coalescent time in the gene tree i.e. $T = \frac{1}{2} \min \{ g_{ab}, g_{ac}, g_{bc}\}$. 
Then, the total branch length of the gene tree generated by $\{a,b,c\}$ is  
\begin{equation*}
    \Lambda = \begin{cases}
    T+\Delta + 2E_1, & 
    {\rm if} \,\, T \in \left( \frac{S_{ab}}{2}, \frac{\Delta}{2} \right) \\
    3T + 2E_1, & 
    {\rm if} \,\, T \geq \frac{\Delta}{2} \,
\end{cases}
\end{equation*}
where $E_1 \sim \Exp(1)$ is independent of $T$. 
Figure~\ref{fig:three-leaf-tree-total-branch-lengths} illustrates this result.
\\
\begin{figure}[!h]
    \centering
    \includegraphics[width=0.7\linewidth]{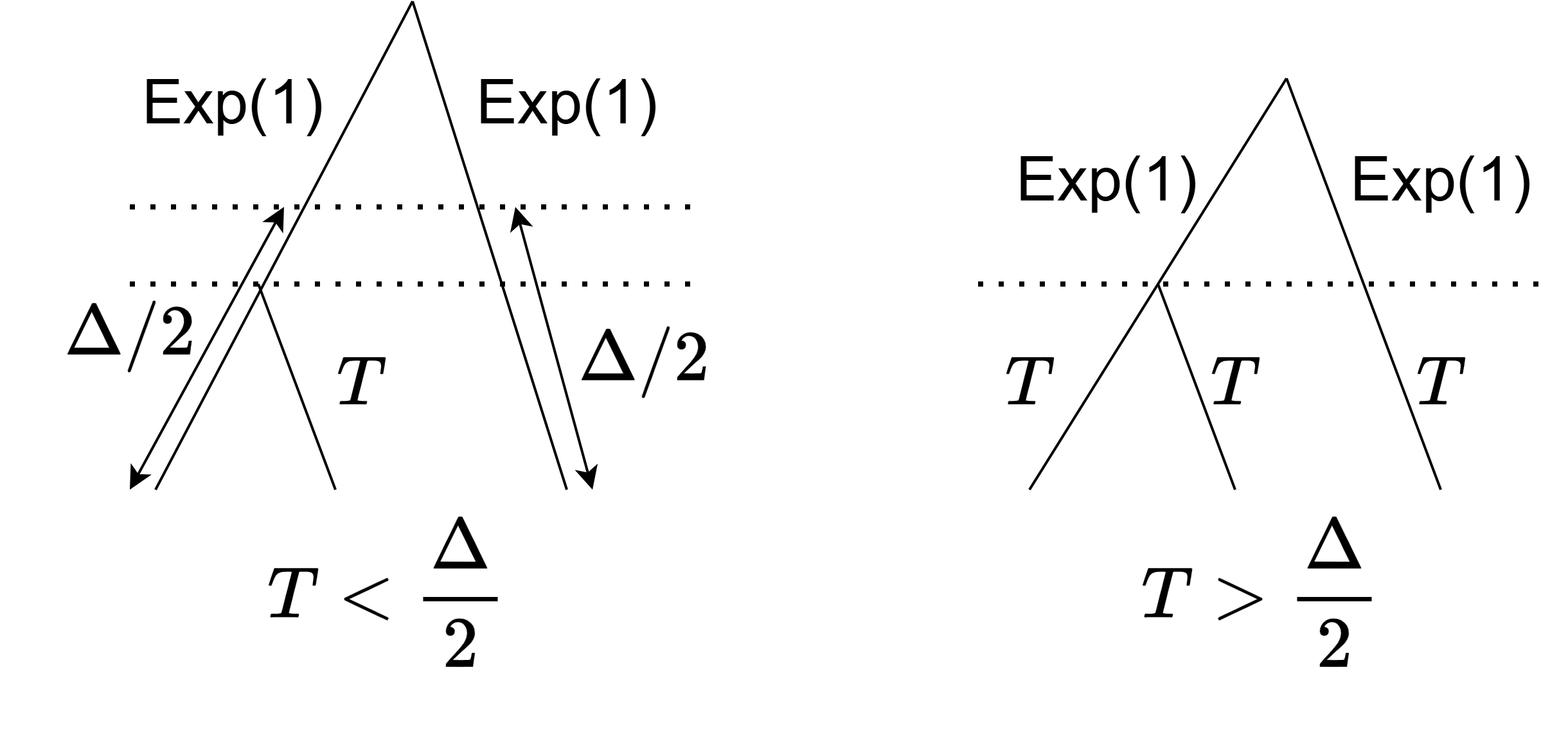}
    \caption{The total length of this gene tree is $T+2 \frac{\Delta}{2} + 2\Exp(1)$ 
    if $T < \Delta / 2$ (left) and $3T + 2\Exp(1)$ if $T \geq \frac{\Delta}{2}$ (right).}
    \label{fig:three-leaf-tree-total-branch-lengths}
\end{figure}
\par 
The pdf of $T$ is 
\begin{equation*}
f_{T}(T) = \begin{cases}
    \exp\left(-\left(T-\frac{S_{ab}}{2}\right)\right), & 
    {\rm if} \,\, T \in \left( \frac{S_{ab}}{2}, \frac{\Delta}{2} \right) \\
    3 \exp\left(- \frac{\Delta - S_{ab} }{2} - 3 \left( T-\frac{\Delta}{2} \right) \right), & 
    {\rm if} \,\, T \geq \frac{\Delta}{2} \\
    0, & {\rm otherwise}
\end{cases}
\end{equation*}
Therefore, we have that 
\begin{align}
    e^{(1)}_{ab,ac} = e^{(1)}_{ac,bc} = e^{(1)}_{ab,bc} &= 
    \mathbb{E}_T\left(\mathbb{E}(e^{t\Lambda }|T)\right) 
    \nonumber 
    \\ &= \frac{1}{1+2\mu} \left[
    \int_{\frac{S_{ab}}{2}}^{ \frac{\Delta}{2} }  
    \exp\left(-\left(T-\frac{S_{ab}}{2}\right) - \mu T - \Delta \right) \, dT  
    \right. \nonumber \\ & \left. + 
    \int_{ \frac{\Delta}{2} }^{\infty} 
    3 \exp\left(- \frac{\Delta - S_{ab} }{2} - 3 \left( T-\frac{\Delta}{2} \right) -3\mu T\right) \, dT
    \right] 
    \nonumber \\ & = \frac{e^{-\mu\left( \Delta + \frac{S_{ab}}{2} \right) } }{(1+\mu) (1+2\mu)},
    \label{eq:Lambda_MGF}
\end{align} as stated for the {\bf 3-leaf tree}.

For the {\bf 4-leaf cherry tree} with topology $((a,b), (c,d))$, again denote $T$ as the first coalescence time
and without loss of generality let $S_{ab} \leq S_{cd}$.
We also define the random variable $\Lambda_{0}$ as the sum of branch lengths of a 3-leaf gene tree, whose 
species tree is the star tree with zero diameter, and whose MGF is a special case of Equation~\eqref{eq:Lambda_MGF}
with $\Delta = S_{ab} = 0$, i.e. 
\begin{equation}
    \mathbb{E}\left(e^{-\mu \Lambda_{0}}\right) =  \frac{1}{(1+2\mu)(1+\mu)}.
    \label{eq:lambda_00}
\end{equation}

The central idea of this derivation is that after 2 leaves 
have coalesced into one, the problem is reduced to a 3-leaf case that is already computed.
Note that if $T< \frac{\Delta}{2}$, then either the first pair to coalesce is either $(a,b)$ 
or $(c,d)$. Either way, $\delta_{ab,cd} = 0$. Conditioning on $T$ and the corresponding first 
coalescent pair ${\bf p}$, 
\begin{equation*}
    g_{ab} + g_{cd} - \delta_{ab,cd}|(T,{\bf p}) \overset{d}{=} 
    \begin{cases}
        2T + \Delta + 2 E_1, &
        {\rm if} \,\, T \in \left( \frac{S_{ab}}{2}, \frac{\Delta}{2} \right)\\
        4T + 2 E_1,&
        {\rm if} \,\, T > \frac{\Delta}{2}, {\bf p} = (a,b) \text{ or }  (c,d)\\
        4T + \Lambda_{0}, & 
        {\rm if} \,\, T > \frac{\Delta}{2}, {\bf p} \neq (a,b) \text{ or }  (c,d)
    \end{cases}
\end{equation*}
where $E_1 \sim \Exp(1)$ and $\Lambda_0$ are independent of $T$. 
It follows that,
\begin{align*}
    e^{(1)}_{ab,cd} &= \mathbb{E}\left(e^{-\mu(g_{ab} + g_{cd} - \delta_{ab,cd}) } \right)
    \\&= 
    \int_{\frac{S_{ab}}{2}}^{\frac{\Delta}{2}} 
        \frac{\exp\left(
        -\left(T-\frac{S_{ab}}{2}\right) -\mu(2T + \Delta)
        \right)}{1+2\mu} \, dT 
    \\&\quad+ 
    e^{-\frac{\Delta-S_{ab}}{2} }
    \int_{\frac{S_{cd}}{2}}^{\frac{\Delta}{2}} 
        \frac{\exp\left(
        -\left(T-\frac{S_{cd}}{2}\right) -\mu(2T + \Delta)
        \right)}{1+2\mu} \, dT
    \\&\quad+ 
    e^{-\left(\Delta - \frac{S_{ab} + S_{cd}}{2} \right)} \left(
    \int_{\frac{\Delta}{2}}^{\infty} \frac{6\exp\left(
        -6(T-\Delta) - 4 \mu T)
    \right)}{1+2\mu}\, dT \cdot \mathbb{P}\left({\bf p}= (a,b),(c,d)\right) + 
    \right. \\ &\quad+ \left. \int_{\frac{\Delta}{2}}^{\infty}
        6\exp\left(
            -6(T-\Delta) - 4 \mu T)
        \right) \mathbb{E}\left(e^{-\mu \Lambda_{0}}\right) \, dT
        \cdot \mathbb{P}\left({\bf p}\neq (a,b),(c,d) \right)  
    \right)
\end{align*}
which yields the desired result, once we substitute Equation~\eqref{eq:lambda_00} and note that 
$\mathbb{P}({\bf p}  = (a,b),(c,d) ) =2/6 = 1/3$.
Similarly, note that 
\begin{equation*}
    g_{ac} + g_{bd} - \delta_{ac,bd}|(T,{\bf p}) \overset{d}{=} 
    \begin{cases}
        4T + \Lambda(\Delta-2T,S_{cd}-2T), &
        {\rm if} \,\, T \in \left( \frac{S_{ab}}{2}, \frac{S_{cd}}{2} \right), {\bf p} = (a,b)\\
        4T + \Lambda(\Delta-2T,0), &
        {\rm if} \,\, T \in \left(\frac{S_{cd}}{2}, \frac{\Delta}{2} \right), {\bf p} = (a,b), (c,d)\\
        4T + 2 E_1,&
        {\rm if} \,\, T > \frac{\Delta}{2}, {\bf p} = (a,c) \text{ or }  (b,d)\\
        4T + \Lambda_{0}, & 
        {\rm if} \,\, T > \frac{\Delta}{2}, {\bf p} \neq (a,c) \text{ or }  (b,d)
    \end{cases}
\end{equation*}
where $E_1 \sim \Exp(1)$ is independent of $T$. 
It follows that
\begin{align*}
    e^{(1)}_{ac,bd} &= \mathbb{E}\left(e^{-\mu(g_{ac} + g_{bd} - \delta_{ac,bd}) } \right)
    \\&= 
    \int_{\frac{S_{ab}}{2}}^{\frac{S_{cd}}{2}} 
        \exp\left(
        -\left(T-\frac{S_{ab}}{2}\right) -4\mu T)
        \right) 
        \frac{\exp\left(-\mu \left(\Delta+ \frac{S_{ab}}{2} \right) +3\mu T \right)}
        {(1+\mu)(1+2\mu)}
        \, dT 
    \\&\quad + 
    e^{-\frac{S_{cd}-S_{ab}}{2} }
    \int_{\frac{S_{cd}}{2}}^{\frac{\Delta}{2}} 
    2\exp\left(
        -2\left(T-\frac{S_{cd}}{2}\right) -4\mu T)
        \right) 
        \frac{\exp\left(-\mu \Delta + 2\mu T \right)}
        {(1+\mu)(1+2\mu)}
        \, dT 
    \\&\quad + 
    e^{-\left(\Delta - \frac{S_{ab} + S_{cd}}{2} \right)} \left(
    \int_{\frac{\Delta}{2}}^{\infty} \frac{6\exp\left(
        -6(T-\Delta) - 4 \mu T)
    \right)}{1+2\mu}\, dT \cdot \mathbb{P}\left({\bf p} = (a,c),(b,d)\right)
    \right. \\ &\quad + \left. \int_{\frac{\Delta}{2}}^{\infty}
        6\exp\left(
            -6(T-\Delta) - 4 \mu T)
        \right) \frac{dT}{(1+\mu)(1+2\mu)}  \cdot \mathbb{P}\left({\bf p} \neq (a,c),(b,d)\right)
    \right)
    \\&= 
    \frac{e^{-\mu\left(\Delta + \frac{S_{ab} + S_{ac}}{2} \right) }}{(1+2\mu)(1+\mu)^2} 
        + e^{-\mu(\Delta + S_{ac}) + \frac{\Delta - S_{ab}}{2}} 
        \frac{\mu(\mu + 2)}{(\mu + 1)^2(2\mu + 1)(2\mu + 3)}
\end{align*}
Note that since $a$ and $b$ are interchangeable in this tree topology, $e^{(1)}_{ad,bc} = e^{(1)}_{ac,bd}$.
This concludes the proof for the cherry tree case.\\
\par
Finally, for the {\bf 4-leaf comb tree} case with topology $(((a,b),c),d)$, and noting that
\begin{equation*}
    g_{ab} + g_{cd} - \delta_{ab,cd}|(T,{\bf p}) \overset{d}{=} 
    \begin{cases}
        2T + \Delta + 2 E_1, &
        {\rm if} \,\, T \in \left( \frac{S_{ab}}{2}, \frac{S_{ac}}{2} \right)\\
        2T + \Delta + 2 E_1, &
        {\rm if} \,\, T \in \left( \frac{S_{ac}}{2}, \frac{\Delta}{2} \right), {\bf p} = (a,b)\\
        4T + \Lambda(\Delta-2T,0) , &
        {\rm if} \,\, T \in \left( \frac{S_{ac}}{2}, \frac{\Delta}{2} \right), {\bf p} = (a,c),(b,c)\\
        4T + 2 E_1,&
        {\rm if} \,\, T > \frac{\Delta}{2}, {\bf p} = (a,c) \text{ or }  (b,d)\\
        4T + \Lambda_{0}, & 
        {\rm if} \,\, T > \frac{\Delta}{2}, {\bf p} \neq (a,c) \text{ or }  (b,d)
    \end{cases}
\end{equation*}
\begin{equation*}
    g_{ac} + g_{bd} - \delta_{ac,bd}|(T,{\bf p}) \overset{d}{=} 
    \begin{cases}
        4T + \Lambda(\Delta-2T,S_{ac}-2T), &
        {\rm if} \,\, T \in \left( \frac{S_{ab}}{2}, \frac{S_{ac}}{2} \right)\\
        2T + \Delta + 2 \Exp(1), &
        {\rm if} \,\, T \in \left( \frac{S_{ac}}{2}, \frac{\Delta}{2} \right), {\bf p} = (a,c)\\
        4T + \Lambda(\Delta-2T,0) , &
        {\rm if} \,\, T \in \left( \frac{S_{ac}}{2}, \frac{\Delta}{2} \right), {\bf p} = (a,b),(b,c)\\
        4T + 2 \Exp(1),&
        {\rm if} \,\, T > \frac{\Delta}{2}, {\bf p} = (a,c) \text{ or }  (b,d)\\
        4T + \Lambda_{0}, & 
        {\rm if} \,\, T > \frac{\Delta}{2}, {\bf p} \neq (a,c) \text{ or }  (b,d)
    \end{cases}
\end{equation*}
where $E_1 \sim \Exp(1)$ is independent of $T$, the desired results are derived by following the same process in the cherry tree case.

\subsection{Computing $e^{(2)}$ terms}
First, we will state all the results for $e^{(2)}_{ab,cd}$, and then provide key ideas 
for their proof.
\begin{itemize}
    \item {\bf Two-leaves.} If $\{a,b\} = \{c,d\}$, then $\delta_{ab,ab} = g_{ab}$ and so 
    the random variable in the exponent is zero, i.e.
    \begin{equation*}
        e^{(2)}_{ab,ab} = 1.
    \end{equation*}
    \item {\bf Three-leaf tree} with tree topology $((a,b),c)$ has 
    \begin{align*}
        e^{(2)}_{ab,ac} = e^{(2)}_{ab,bc} &= \frac{e^{-\mu S_{ac}}}{1+2\mu} \\
         e^{(2)}_{ac,bc} &= \frac{e^{-\mu S_{ab}}}{1+2\mu}
    \end{align*}
    \item {\bf Four leaves (Cherry tree)} with tree topology $((a,b),(c,d))$ has
    \begin{equation*}
        e^{(2)}_{ab,cd} = e^{(2)}_{ac,bd}= e^{(2)}_{ad,bc} = 
        \frac{e^{-\mu( S_{ab} + S_{cd}) }}{(1+2\mu)^2} 
        + e^{-(2\mu+1) \Delta + \frac{S_{ab} + S_{cd}}{2}} 
        \frac{4\mu}{(2\mu + 1)^2(2\mu + 3)} 
    \end{equation*}
    \item {\bf Four leaves (Comb tree)} with tree topology $(((a,b), c),d)$ has
    \begin{equation*}
        e^{(2)}_{ab,cd} = e^{(2)}_{ac,bd}= e^{(2)}_{ad,bc} = 
        \frac{e^{-\mu (S_{ab} + \Delta)}}{(1+2\mu)^2} 
        + e^{-\mu(\Delta + S_{ac}) - \frac{S_{ac} - S_{ab}}{2}} 
        \frac{4\mu}{(2\mu + 1)^2(2\mu + 3)} 
    \end{equation*}
\end{itemize}

To prove the three leaf case, observe that 
\begin{equation}
    g_{xy} + g_{xz} - 2 \delta_{xy,xz} = g_{yz}, \quad \forall (x,y,z) \in \binom L3
\end{equation}
In other words,
\begin{equation*}
    e^{(2)}_{xy,xz} = \mathbb{E}(e^{-\mu g_{yz}}) = \frac{e^{-\mu S_{yz}}}{1+2\mu}
\end{equation*}
For the four-leaf comb and cherry case, we observe that 
\begin{equation*}
    g_{ab} + g_{cd} - 2\delta_{ab,cd} = g_{\bf p} + g_{\bf p^c},
\end{equation*}
where ${\bf p} \in \binom{\{a,b,c,d\}}{2}$ is the first pair in the gene tree to coalesce 
and ${\bf p}^c$ is the complimentary pair; the remaining two leaves. Conditioning on $T,{\bf p}$
just like before yields the desired results.

\section{Algorithms} \label{sec:algorithms}

\par Algorithm \ref{alg:hamming} describes the procedure to compute the Hamming distance between two alignment sequences. This is a sub-routine within the distance-based approaches for species tree clustering - METAL, GLASS and STEAC - which are presented in Algorithm \ref{alg:species_tree_reconstruction}.

\begin{algorithm}[ht]
\caption{Hamming Distance Between Two Sequences}
\label{alg:hamming}
\begin{algorithmic}[1]
\Require Sequences $X$ and $Y$ of equal length $K$
\Require $X, Y \in \{\texttt{A}, \texttt{C}, \texttt{T}, \texttt{G}\}^K$
\Function{Hamming}{$X$, $Y$}
  \State $d \gets 0$
  \For{$i \gets 1$ to $K$}
    \If{$X[i] \neq Y[i]$}
      \State $d \gets d + 1$
    \EndIf
  \EndFor
  \State \Return $d / K$ \Comment{Normalize Hamming Distances}
\EndFunction
\end{algorithmic}
\end{algorithm}

\begin{algorithm}[ht]
\caption{Species‐Tree Estimation via Agglomerative Distance‐Based Methods}
\label{alg:species_tree_reconstruction}
\begin{algorithmic}[1]
\Require gene sequences $\chi_i^{(g)}$ for genes $g \in G = [m]$ and leaves/species $i \in L$
\Require method $\in \{\textbf{METAL},\textbf{STEAC},\textbf{GLASS}\}$
\Ensure species tree $T$

\If{method == \textbf{METAL}}
  \Comment{Concatenate gene sequences}
  \ForAll{leaves $i \in L$}
    \State $\chi_i \gets \mathrm{Concatenate}\left(\chi_i^{(1)}, \chi_i^{(2)}, \dots, \chi_i^{(|G|)}\right)$
  \EndFor
  \Comment{Compute distances on concatenated alignment}
  \ForAll{leaf pairs $(i,j) \in {\binom{L}{2}}$}
    \State $D[i,j] \gets \mathrm{Hamming}\left(\chi_i, \chi_j \right)$ \Comment{Normalized Hamming Distances}
  \EndFor
\Else
  \Comment{Compute gene-by-gene distances}
  \For{$g \in G$}
    \ForAll{leaf pairs $(i,j) \in {\binom{L}{2}} $}
      \State $\theta \gets \mathrm{Hamming}\left(\chi_i^{(g)} , \chi_j^{(g)}\right)$
      \If{$\theta \geq \frac{3}{4}$ {\bf \&} method == \textbf{STEAC}}
      \State $G$.delete($g$) \Comment{Delete gene}
      \State {\bf break}
      \ElsIf{method == \textbf{STEAC}}
      \State $D_g[i,j] \gets -\frac{3}{4}\log\left(1- \frac{4}{3} \theta \right)$ \Comment{Convert to coalescent units}
      \ElsIf{method == \textbf{GLASS}}
      \State $D_g[i,j] \gets \theta$ 
      \EndIf
    \EndFor
  \EndFor
  \Comment{Aggregate across genes}
  \ForAll{leaf pairs $(i,j) \in {\binom{L}{2}}$}
    \If{method == \textbf{STEAC}}
      \State $D[i,j] \gets \frac{1}{|G|}\sum_{g\in G} D_g[i,j]$
    \ElsIf{method == \textbf{GLASS}}
      \State $D[i,j] \gets \min_{g \in G} D_g[i,j]$
    \EndIf
  \EndFor
\EndIf
\State $T \gets \mathrm{HierarchicalClustering}(D)$ \Comment{Hierarchical clustering (e.g., UPGMA or NJ)}
\State {\bf return} T
\end{algorithmic}
\end{algorithm}

\section{Additional plots}\label{sec:plots}

\begin{figure}[h]
    \centering
    \includegraphics[width=0.9\linewidth]{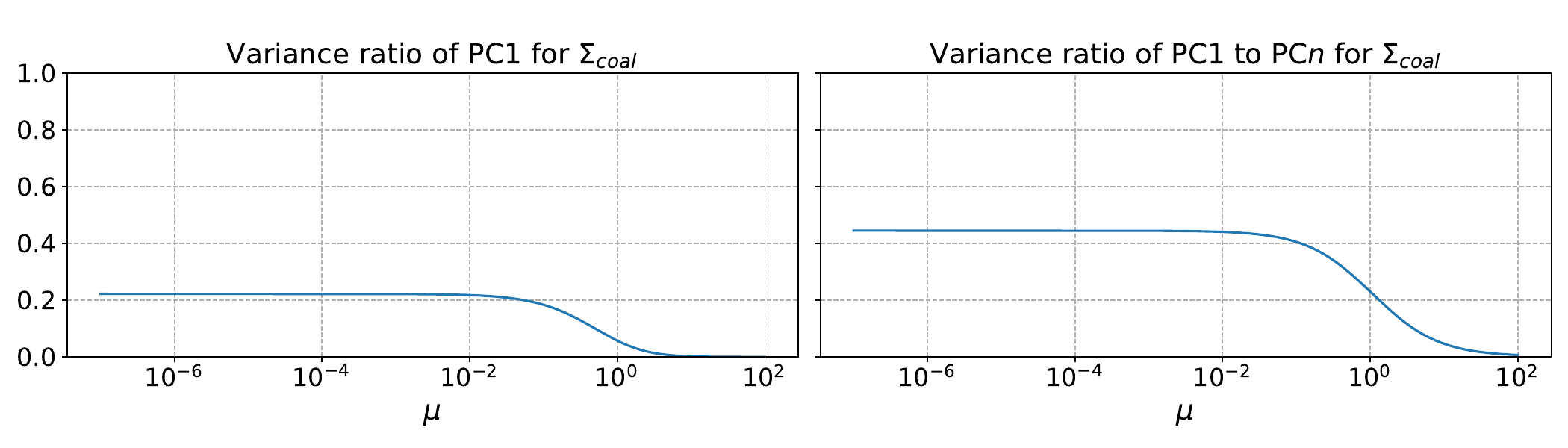} \\
    \includegraphics[width=0.9\linewidth]{variance_ratio_sub-eps-converted-to.pdf} \\
    \includegraphics[width=0.9\linewidth]{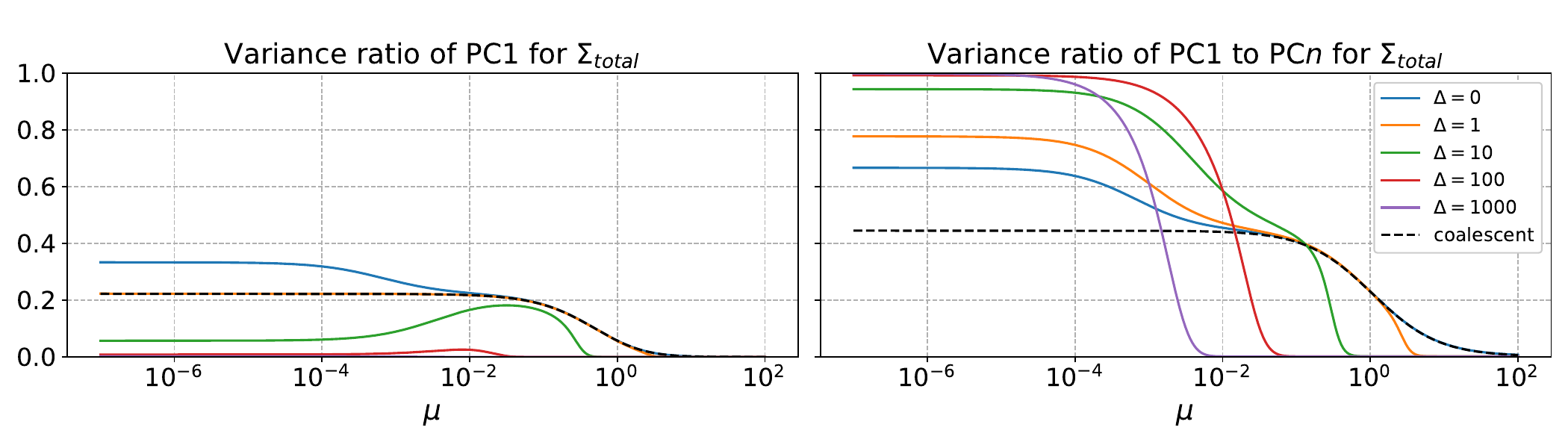} 
    \caption{Variance ratios for (top) \(\Sigma_{\rm coal}\), (middle) \(\Sigma_{\rm sub}\), and (bottom) \(\Sigma_{\rm total}\), plotted against mutation rate \(\mu\) and species‐tree diameter \(\Delta\). In each row, the left panel shows the fraction of total variance explained by the first principal component, and the right panel shows the fraction of total variance explained by the first \(n\) principal components. Here, the number of sites per gene is \(K=1000\). Since \(\Sigma_{\rm coal}\) does not depend on \(\Delta\), its row contains only a single curve. For \(\Sigma_{\rm total}\), note that when \(\mu \ll 1\) and \(\mu \Delta \gg 1\) (i.e., substitution variance dominates), its curves resemble those of \(\Sigma_{\rm sub}\), whereas for intermediate \(\mu\) (when MSC variance dominates), all three models’ curves converge to that of the coalescent covariance.}
    \label{fig:complete_variance_ratio}
\end{figure}

\begin{figure}[h]
    \centering
    \includegraphics[width=0.8\linewidth]{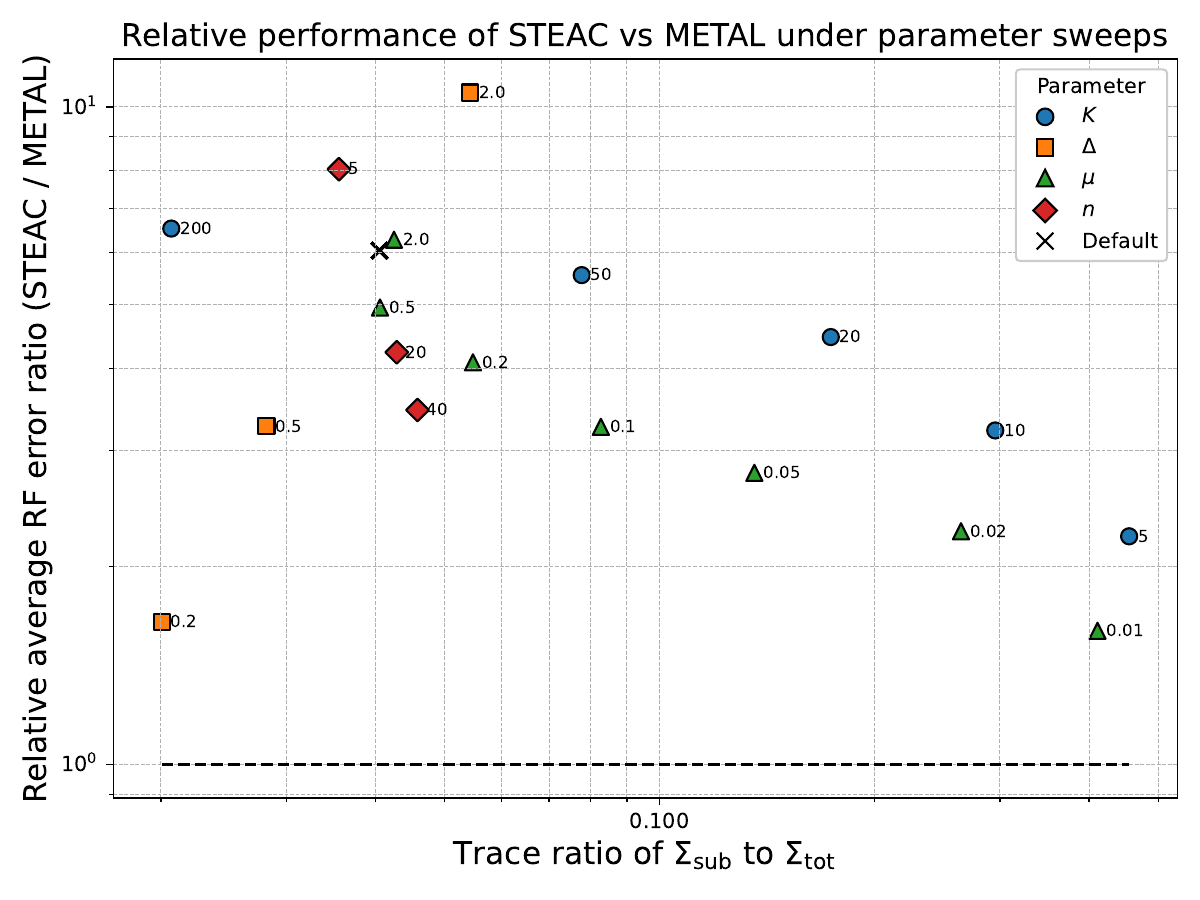}
    \caption{This Figures compares STEAC to METAL, similarly to Figure~\ref{fig:glass_vs_metal}.
    In all cases METAL outperforms STEAC}
    \label{fig:steac_vs_metal}
\end{figure}

\begin{figure}[h]
  \centering
  \begin{minipage}[b]{0.48\textwidth}
    \centering
    \includegraphics[width=\textwidth]{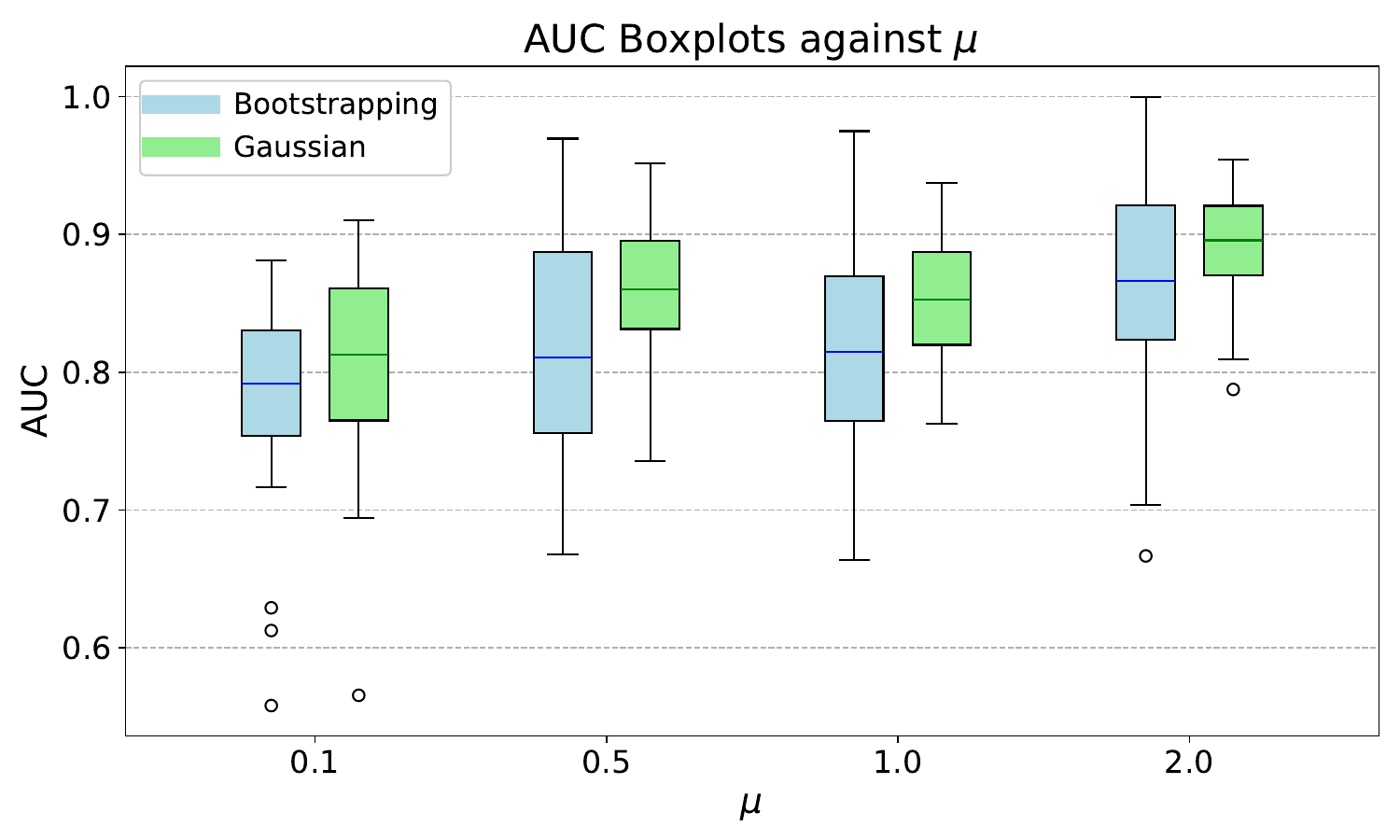}
  \end{minipage}
  \hfill
  \begin{minipage}[b]{0.48\textwidth}
    \centering
    \includegraphics[width=\textwidth]{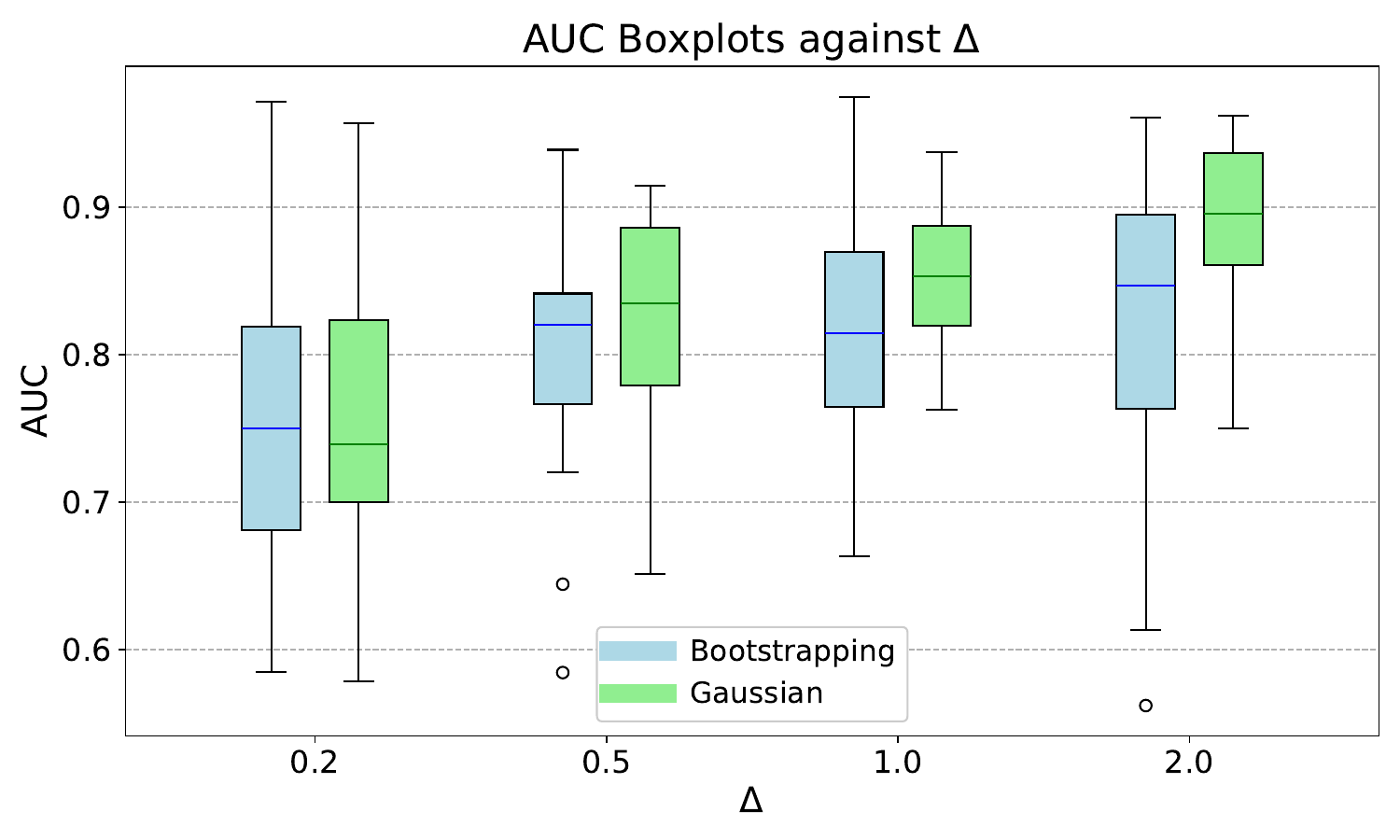}
  \end{minipage}
  \label{fig:appendix_auc_boxplots}
  \caption{Boxplots of AUC values 
  for different mutation rates $\mu$ (left), and different 
  species tree diameter $\Delta$ (right).}
\end{figure} 

\begin{figure}[h]
    \centering
    \includegraphics[width=0.7\linewidth]{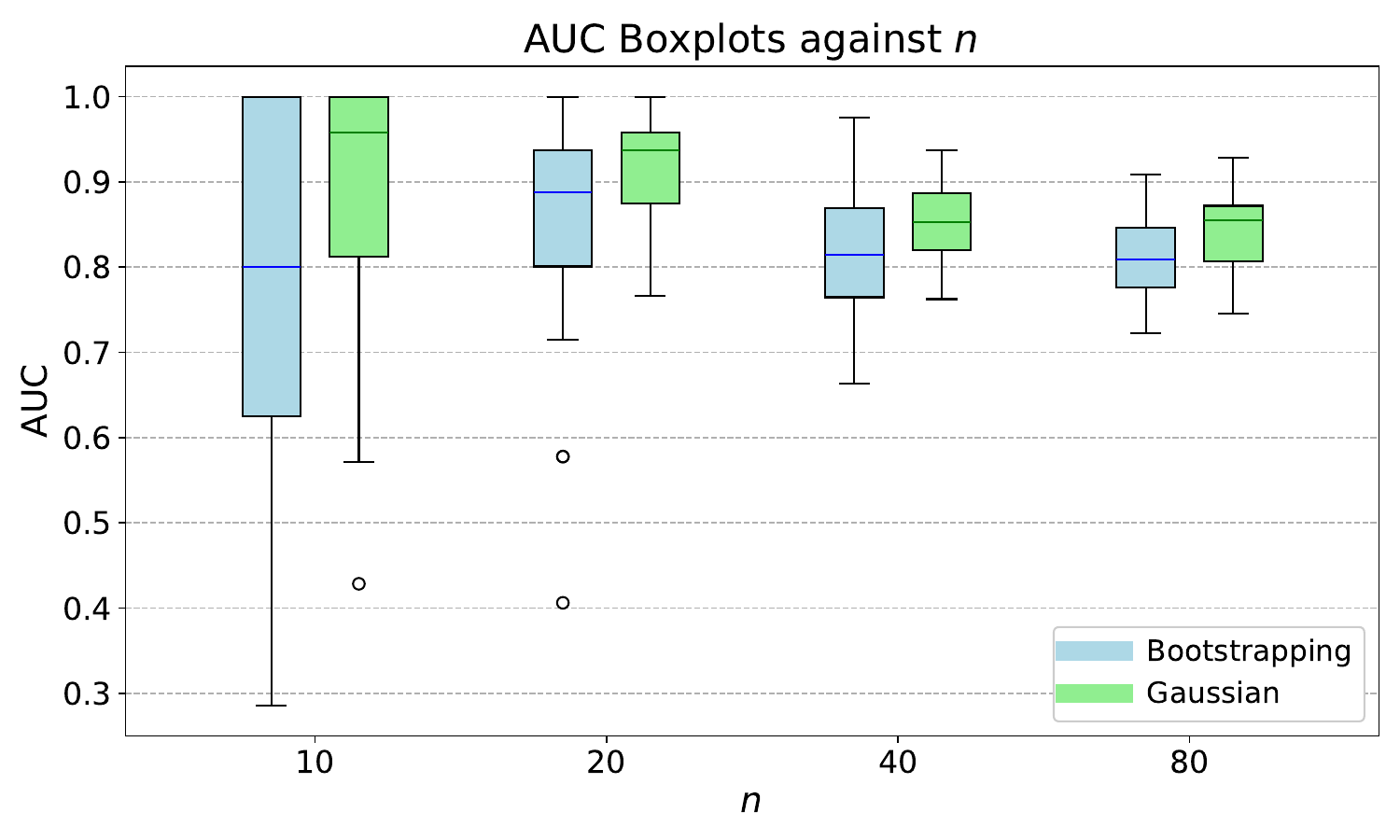}
    \label{fig:appendix_auc_boxplots_ntax}
    \caption{Boxplots of AUC values 
      for different number of taxa $n$.}
\end{figure}

\end{appendices}
\bibliographystyle{plain}
\bibliography{references}

\begin{thebibliography}{10}

\bibitem{aliatimis2024metal}
Georgios Aliatimis.
\newblock Metal covariance matrix.
\newblock \url{https://github.com/GeorgiosAliatimis/metal_covariance_matrix},
  2024.
\newblock Accessed: 2025-06-19.

\bibitem{aliatimis2024tropical}
Georgios Aliatimis, Ruriko Yoshida, Burak Boyac{\i}, and James~A Grant.
\newblock Tropical logistic regression model on space of phylogenetic trees.
\newblock {\em Bulletin of Mathematical Biology}, 86(8):99, 2024.

\bibitem{barnhill2023clustering}
David Barnhill and Ruriko Yoshida.
\newblock Clustering methods over the tropical projective torus.
\newblock {\em Mathematics}, 11(15):3433, 2023.

\bibitem{bouckaert2019beast}
Remco Bouckaert, Timothy~G Vaughan, Jo{\"e}lle Barido-Sottani, Sebasti{\'a}n
  Duch{\^e}ne, Mathieu Fourment, Alexandra Gavryushkina, Joseph Heled, Graham
  Jones, Denise K{\"u}hnert, Nicola De~Maio, et~al.
\newblock Beast 2.5: An advanced software platform for bayesian evolutionary
  analysis.
\newblock {\em PLoS computational biology}, 15(4):e1006650, 2019.

\bibitem{braun2024testing}
Edward~L Braun, Carl~H Oliveros, Noor~D White~Carreiro, Min Zhao, Travis~C
  Glenn, Robb~T Brumfield, Michael~J Braun, Rebecca~T Kimball, and Brant~C
  Faircloth.
\newblock Testing the mettle of metal: A comparison of phylogenomic methods
  using a challenging but well-resolved phylogeny.
\newblock {\em BioRxiv}, pages 2024--02, 2024.

\bibitem{dasarathy2014data}
Gautam Dasarathy, Robert Nowak, and Sebastien Roch.
\newblock Data requirement for phylogenetic inference from multiple loci: a new
  distance method.
\newblock {\em IEEE/ACM transactions on computational biology and
  bioinformatics}, 12(2):422--432, 2014.

\bibitem{degnan2009gene}
James~H Degnan and Noah~A Rosenberg.
\newblock Gene tree discordance, phylogenetic inference and the multispecies
  coalescent.
\newblock {\em Trends in ecology \& evolution}, 24(6):332--340, 2009.

\bibitem{dugad1998unsupervised}
Rakesh Dugad and Narendra Ahuja.
\newblock Unsupervised multidimensional hierarchical clustering.
\newblock In {\em Proceedings of the 1998 IEEE International Conference on
  Acoustics, Speech and Signal Processing, ICASSP'98 (Cat. No. 98CH36181)},
  volume~5, pages 2761--2764. IEEE, 1998.

\bibitem{flouri2018species}
Tom{\'a}{\v{s}} Flouri, Xiyun Jiao, Bruce Rannala, and Ziheng Yang.
\newblock Species tree inference with bpp using genomic sequences and the
  multispecies coalescent.
\newblock {\em Molecular biology and evolution}, 35(10):2585--2593, 2018.

\bibitem{hasegawa1985dating}
Masami Hasegawa, Hirohisa Kishino, and Taka-aki Yano.
\newblock Dating of the human-ape splitting by a molecular clock of
  mitochondrial dna.
\newblock {\em Journal of molecular evolution}, 22:160--174, 1985.

\bibitem{jukes1969evolution}
TH~Jukes.
\newblock Evolution of protein molecules.
\newblock {\em Mammalian Protein Metabolism}, 3, 1969.

\bibitem{lee2022tropical}
Wonjun Lee, Wuchen Li, Bo~Lin, and Anthea Monod.
\newblock Tropical optimal transport and wasserstein distances.
\newblock {\em Information Geometry}, 5(1):247--287, 2022.

\bibitem{liu2009estimating}
Liang Liu, Lili Yu, Dennis~K Pearl, and Scott~V Edwards.
\newblock Estimating species phylogenies using coalescence times among
  sequences.
\newblock {\em Systematic biology}, 58(5):468--477, 2009.

\bibitem{maddison1997gene}
Wayne~P Maddison.
\newblock Gene trees in species trees.
\newblock {\em Systematic biology}, 46(3):523--536, 1997.

\bibitem{Moreno2024}
Matthew~Andres Moreno, Mark~T. Holder, and Jeet Sukumaran.
\newblock Dendropy 5: a mature python library for phylogenetic computing.
\newblock {\em Journal of Open Source Software}, 9(101):6943, 2024.

\bibitem{mossel2008incomplete}
Elchanan Mossel and Sebastien Roch.
\newblock Incomplete lineage sorting: consistent phylogeny estimation from
  multiple loci.
\newblock {\em IEEE/ACM Transactions on Computational Biology and
  Bioinformatics}, 7(1):166--171, 2008.

\bibitem{ronquist2012mrbayes}
Fredrik Ronquist, Maxim Teslenko, Paul Van Der~Mark, Daniel~L Ayres, Aaron
  Darling, Sebastian H{\"o}hna, Bret Larget, Liang Liu, Marc~A Suchard, and
  John~P Huelsenbeck.
\newblock Mrbayes 3.2: efficient bayesian phylogenetic inference and model
  choice across a large model space.
\newblock {\em Systematic biology}, 61(3):539--542, 2012.

\bibitem{SpielmanWilke2015}
Stephanie~J. Spielman and Claus~O. Wilke.
\newblock Pyvolve: A flexible python module for simulating sequences along
  phylogenies.
\newblock {\em PLOS ONE}, 10(9):e0139047, 2015.

\bibitem{tavare1986some}
Simon Tavar{\'e}.
\newblock Some probabilistic and statistical problems on the analysis of dna
  sequence.
\newblock {\em Lecture of mathematics for life science}, 17:57, 1986.

\bibitem{yang2012molecular}
Ziheng Yang and Bruce Rannala.
\newblock Molecular phylogenetics: principles and practice.
\newblock {\em Nature reviews genetics}, 13(5):303--314, 2012.

\end{thebibliography}

\end{document}